\documentclass[sigconf]{acmart}

\usepackage{booktabs} 
\usepackage{xspace, tabularx}
\usepackage{color}
\usepackage{soul}
\usepackage[margin=4pt,caption=false]{subfig}
\usepackage{diagbox}
\usepackage{multirow}
\usepackage{algorithm}
\usepackage{graphicx}
\usepackage{amsmath, bm}
\usepackage{algpseudocode}
\usepackage{xcolor}
\usepackage{tikz}
\usetikzlibrary{positioning, shapes.geometric}
\usepackage{colortbl}
\newcommand*\emptycircle{\tikz{\draw circle (0.9mm)}}
\newcommand{\blackcircle}{\tikz{\filldraw[fill=black] circle (0.9mm)}}
\newcommand{\semifilled}{\tikz
    {
    \draw circle (0.9mm);
    \node [semicircle, fill=black, 
      inner sep=0pt, outer sep=0pt,minimum size=0.9mm, anchor=south, rotate=90]{};
     }}
\usepackage{comment}

\usepackage{amssymb}
\usepackage{alltt}
\usepackage{comment}
\usepackage{makecell}

\allowdisplaybreaks
\usepackage{xparse}
\DeclareMathOperator*{\argmax}{arg\,max}

\usepackage[utf8]{inputenc}
\usepackage[english]{babel}
\definecolor{lavender}{rgb}{0.75, 0.58, 0.89}
\usepackage{mathtools}

\newtheorem{condition}{Condition}

\newcommand{\tts}[1]{{\small\texttt{#1}}}


\makeatletter
\newcommand{\removelatexerror}{\let\@latex@error\@gobble}
\makeatother

\newif\ifshowcomment
\showcommenttrue

\ifshowcomment
    \newcommand{\Jian}[1]{\textsf{\color{green}{[{Jian: #1}]}}}
    \newcommand{\dawn}[1]{\textsf{\color{red}{[{Dawn: #1}]}}}
    \newcommand{\yongdae}[1]{\textsf{\color{blue}{[{Yongdae: #1}]}}}
    \newcommand{\yujin}[1]{\textsf{\color{lavender}{[{Yujin: #1}]}}}

\else
    \newcommand{\Jian}[1]{}
    \newcommand{\dawn}[1]{}
    \newcommand{\yongdae}[1]{}
    \newcommand{\yujin}[1]{}
\fi

\newcommand{\para}[1]{\smallskip\noindent\textbf{#1}}
\usepackage{pifont}
\newcommand{\xmark}{\ding{55}}%
\newcommand{\tri}{\ding{115}}%

\allowdisplaybreaks
\setcopyright{acmcopyright}

\begin{document}
\title{Impossibility of Full Decentralization in Permissionless Blockchains} 

\copyrightyear{2019} 
\acmYear{2019} 
\acmConference[AFT '19]{1st ACM Conference on Advances in Financial Technologies}{October 21--23, 2019}{Zurich, Switzerland}
\acmBooktitle{1st ACM Conference on Advances in Financial Technologies (AFT '19), October 21--23, 2019, Zurich, Switzerland}
\acmPrice{15.00}
\acmDOI{10.1145/3318041.3355463}
\acmISBN{978-1-4503-6732-5/19/10}

\author{Yujin Kwon*, Jian Liu$^\dagger$, Minjeong Kim*, Dawn Song$^\dagger$, Yongdae Kim*}
\affiliation{\vspace{2mm}
\institution{*KAIST}}
\email{{dbwls8724, mjkim9394, yongdaek}@kaist.ac.kr}
\affiliation{ \vspace{2mm} \institution{$^\dagger$UC Berkeley}}
\email{jian.liu@eecs.berkeley.edu, dawnsong@cs.berkeley.edu}

\renewcommand{\shortauthors}{Kwon et al.}

\begin{abstract}
Bitcoin uses the \textit{proof-of-work} (PoW) mechanism where nodes earn rewards in return for the use of their computing resources.
Although this incentive system has attracted many participants, power has, at the same time, been significantly biased towards a few nodes, called \textit{mining pools}. 
In addition, poor decentralization appears not only in PoW-based coins but also in coins that adopt \textit{proof-of-stake} (PoS) and \textit{delegated proof-of-stake} (DPoS) mechanisms.

In this paper, we address the issue of centralization in the consensus protocol. 
To this end, we first define \textit{$(m, \varepsilon, \delta)$-decentralization} as a state satisfying that 1) there are at least $m$ participants running a node, and 2) the ratio between the total resource power of nodes run by the richest and the $\delta$-th percentile participants is less than or equal to $1+\varepsilon$. 
Therefore, when $m$ is sufficiently large, and $\varepsilon$ and $\delta$ are 0, $(m, \varepsilon, \delta)$-decentralization represents \textit{full decentralization}, which is an ideal state. 
To ascertain if it is possible to achieve good decentralization, 
we introduce conditions for an incentive system that will allow a blockchain to achieve $(m, \varepsilon, \delta)$-decentralization. 
When satisfying the conditions, a blockchain system can reach full decentralization with probability 1, regardless of its consensus protocol.
However, to achieve this, the blockchain system should be able to assign a positive Sybil cost, where the Sybil cost is defined as the difference between the cost for one participant running multiple nodes and the total cost for multiple participants each running one node. 
Conversely, we prove that if there is no Sybil cost,
the probability of achieving $(m, \varepsilon, \delta)$-decentralization is bounded above by a function of $f_\delta$, where $f_\delta$ is the ratio between the resource power of the $\delta$-th percentile and the richest participants. 
Furthermore, the value of the upper bound is close to 0 for small values of $f_\delta$. 
Considering the current gap between the rich and poor, this result implies that it is almost impossible for a system without Sybil costs to achieve good decentralization. 
In addition, because it is yet unknown how to assign a Sybil cost without relying on a TTP in blockchains, it also represents that currently, a contradiction between achieving good decentralization in the consensus protocol and not relying on a TTP exists.
\end{abstract}



\maketitle

\section{Introduction}

Traditional currencies have a centralized structure, and thus there exist several problems such as a single point of failure and corruption.
For example, the global financial crisis in 2008 was aggravated by the flawed policies of banks that eventually led to many bank failures, followed by an increase in the distrust of these institutions. 
With this background, Bitcoin~\cite{nakamoto2008bitcoin}, which is the first decentralized digital currency, has received considerable attention; given that it is a decentralized cryptocurrency, there is no organization that controls the system, unlike traditional financial systems. 

To operate the system without any central authority, Bitcoin uses \textit{blockchain} technology. 
Blockchain is a public ledger that stores transaction history, while nodes record the history on the blockchain by generating blocks through a consensus protocol that provides a synchronized view among nodes. 
Bitcoin adopts a consensus protocol using the PoW mechanism in which nodes utilize their computational power in order to participate.
Moreover, nodes receive coins as a reward for the use of their computational power, and this reward increases with the amount of computational power used.
This incentive system has attracted many participants. 
At the same time, however, computational power has been significantly biased towards a few participants (i.e., mining pools).
As a result, the decentralization of the Bitcoin system has become poor, thus deviating from its original goal~\cite{gervais2014bitcoin, bonneau2015sok, gencer2018decentralization}.

Since the success of Bitcoin, many  cryptocurrencies have been developed. They have attempted to address several drawbacks of Bitcoin, such as low transaction throughput, waste of  energy owing to the utilization of vast computational power, and poor decentralization.
Therefore, some cryptocurrencies use consensus mechanisms different from PoW, such as PoS and DPoS, in which nodes should have stakes for participation instead of a computing resource.
While these new consensus mechanisms have addressed several of the drawbacks of Bitcoin, the problem of poor decentralization still remains unsolved. 
For example, similar to PoW systems, stakes are also significantly biased towards a few participants.
This has caused concern about poor decentralization in PoS and DPoS coins. 

Currently, many coins suffer from two problems that degrade the level of decentralization: 1) an insufficient number of independent participants because of the coalition of participants (e.g., mining pools) and 2) a significantly biased power distribution among them. 
Therefore, many developers have attempted to create a well-decentralized system~\cite{quadratic_voting, brunjes2018reward}. 
In addition, researchers such as Micali have noted that ``incentives are the hardest thing to do" and believe that inappropriate incentive systems may cause blockchain systems to be significantly centralized~\cite{debate}. 
This implies that it is currently an open problem as to whether we can design an incentive system that allows for good or full decentralization to be achieved. 

\smallskip
\noindent\textbf{Full decentralization. }
In this paper, the conditions for full decentralization are studied for the first time. 
To this end, we define \textit{$(m, \varepsilon, \delta)$-decentralization} as a state that satisfies that 1) \textit{the number of participants running nodes in a consensus protocol is not less than $m$} and 2) \textit{the ratio between the effective power of the richest and the $\delta$-th percentile participants is not greater than $1+\varepsilon$,} where the effective power of a participant represents the total resource power of the nodes run by that participant. 
The case when $m$ is sufficiently large and $\varepsilon$ and $\delta$ are 0 represents full decentralization in which everyone has the same power. 
To investigate if a high level of decentralization is possible, we model a blockchain system (Section~\ref{sec:model}), 
and find four sufficient conditions of the incentive system such that the blockchain system \textit{converges in probability} to $(m, \varepsilon, \delta)$-decentralization.
If an incentive system that satisfies these four conditions exists, the blockchain system can achieve $(m, \varepsilon, \delta)$-decentralization with probability 1, regardless of the underlying consensus protocol. 
The four conditions are: 1) \textit{at least $m$ nodes earn rewards}, 2) \textit{it is not more profitable for participants to delegate their resource power to fewer participants than it is to run their own nodes}, 
3) \textit{it is not more profitable for a participant to run multiple nodes than to run one node}, and 4) \textit{the ratio between the resource power of the richest and the $\delta$-th percentile nodes converges in probability to a value of less than $1+\varepsilon.$} 

\smallskip
\noindent\textbf{Impossibility. }
Based on these conditions, we find an incentive system that enables a system to achieve full decentralization. 
\textit{In this incentive system, for the third condition to be met, the cost for one participant running multiple nodes should be greater than the total cost for multiple participants each running one node.
The difference between the former cost and the latter cost is called a Sybil cost in this paper. }
This implies that a system where Sybil costs exist can be fully decentralized with probability 1. 

When a system does not have Sybil costs, there is no incentive system that satisfies the four conditions (Section~\ref{sec:impossible}). 
More specifically, the probability of reaching $(m, \varepsilon, \delta)$-decentralization is bounded above by a function $G(f_\delta)$ that is close to 0 for a small ratio $f_\delta$ between the resource power of the $\delta$-th percentile and the richest participants. 
This implies that achieving good decentralization in a system without Sybil costs depends totally on the rich-poor gap in the real world. 
As such, the larger the rich-poor gap, the closer the probability is to zero. 
To determine the approximate ratio $f_\delta$ in actual systems, we investigate hash rates in Bitcoin and observe that $f_0\,\,(\delta=0)$ and $f_{15}\,\,(\delta=15)$ are less than $10^{-8}$ and $1.5\times 10^{-5}$, respectively. 
In this case, $f_0$ indicates the ratio between the resource power of the poorest and the richest participants. 

Unfortunately, it is not yet known how permissionless blockchains that have no \textit{real identity management} can have Sybil costs.
Indeed, to the best of our knowledge, all permissionless blockchains that do not rely on a TTP do not currently have any Sybil costs. 
Taking this into consideration, \textit{it is almost impossible for permissionless blockchains to achieve good decentralization, and there is a contradiction between achieving good decentralization in the consensus protocol and not relying on a TTP}. 
The existence of mechanisms to enforce a Sybil cost in permissionless blockchains is left as an open problem.  
The solution to this issue would be the key to determining how blockchains can achieve a high level of decentralization. 

\smallskip
\noindent\textbf{Protocol analysis in the top 100 coins. }
Next, to find out what condition each system does not satisfy, we extensively analyze incentive systems of all existing PoW, PoS, and DPoS coins among the top 100 coins in CoinMarketCap~\cite{coinmarket} according to the four conditions (Section~\ref{app:protocol}). 
According to this analysis, PoW and PoS systems do not have both enough participants running nodes and an even power distribution among the participants.
However, unlike PoW and PoS coins, DPoS coins can have an even power distribution among a fixed number of participants when Sybil costs exist. 
If the Sybil costs do not exist, however, rational participants would run multiple nodes for higher profits. 
In that case, DPoS systems cannot guarantee that any participants possess the same power.

\smallskip
\noindent\textbf{Data analysis in top 100 coins. }
To validate the result of the protocol analysis and our theory, we also conduct data analysis of the same list of coins using three metrics: the number of block generators, the Gini coefficient, and Shannon entropy (Section~\ref{app:level}). 
Through this empirical study, we can observe the expected rational behaviors in most existing coins. 
In addition, we \textit{quantitatively confirm} that the coins do not currently achieve good decentralization. 
As a result, this data analysis not only investigates the actual level of decentralization, but also empirically confirms the analysis results of incentive systems.
We discuss the debate surrounding incentive systems and whether we can relax the conditions for full decentralization (Section~\ref{sec:discuss}). 
Finally, we conclude and provide two directions to go (Section~\ref{sec:conclude}).
\section{Importance of decentralization}\label{sec:background}

Decentralization is an essential factor that should be inherent in the design of blockchain systems. 
However, most of the computational power of PoW-based systems is currently concentrated in only a few nodes, called \textit{mining pools},\footnote{More specifically, this refers to centralized mining pools. Even though there are decentralized mining pools, given that centralized pools are major pools, we will, hereafter, simply refer to them as mining pools.} where individual miners gather together for mining. 
This causes concern not only about the level of decentralization, but also about the security of systems since the mining-power distribution is a critical aspect to be considered in the security of PoW systems.
In general, when a participant has large amounts of resource power, their behavior will significantly influence others in the consensus protocol. 
In other words, the more resources a participant has, the greater their influence on the system.
Therefore, the resource power distribution implicitly represents the level of decentralization in the system. 

At this point, we can consider the following questions: ``What can influential participants do in practice?" and ``Can this behavior harm other nodes?" 
Firstly, there are attacks such as double spending and selfish mining, which can be executed by attackers with over 50\% and 33\% of the resource power, respectively. 
These attacks would result in significant financial damage~\cite{btg_accident}.  
In addition, in a consensus protocol combined with PBFT~\cite{castro1999practical}, malicious behavior of nodes that possess over 33\% resource power can cause the consensus protocol to become stuck. 
It would certainly be more difficult for such attacks to be executed through collusion with others if the resource power is more evenly distributed. 
In addition, nodes participating in the consensus protocol verify transactions and generate blocks. More specifically, when generating a block, nodes choose which transactions to include in that block.
Therefore, they can choose only the advantageous transactions while ignoring the disadvantageous transactions. 
For example, participants can exclude transactions issued by rivals in the process of generating blocks and, if they possess large amounts of power, validation of these transactions will often be delayed because the malicious participant has many opportunities to choose the transactions that will be validated. 
Even though the rivals can also retaliate against them, the damage from the retaliation depends on the power gap between the malicious participants and their rivals. 

Furthermore, transaction issuers are required to pay transaction fees. 
The fees are usually determined by economic interactions~\cite{fees}. 
This implies that the fees can depend on the behavior of block generators.
For example, if they verify only transactions that have fees above a specific amount, the overall transaction fees can increase because users would have to pay a higher fee for their transactions to be validated.
Considering this, the more the system is centralized, the closer it may become to oligopolies. 

In fully decentralized systems, however, it would be significantly more difficult for the above problems to occur. Moreover, the system would certainly be fair to everyone.
This propels the desire to achieve a fully decentralized system.
Even though there have been many discussions and attempts to achieve good decentralization, existing systems except for a few coins~\cite{gencer2018decentralization, minjeong} have rarely been analyzed. 
This paper not only studies the possibility of full decentralization, but also extensively investigates the existing coins.
\section{System Model}
\label{sec:model}
In this section, we model a consensus protocol and an incentive system. 
Moreover, we introduce the notation used throughout this paper (see Tab.~\ref{tab:par}). 

\smallskip
\noindent\textbf{Consensus protocol. }
A blockchain system has a consensus protocol where player $p_i$ participates and generates blocks by running their own nodes. 
The set of all nodes in the consensus protocol is denoted by $\mathcal{N},$ and that of the nodes run by player $p_i$ is denoted by $\mathcal{N}_{p_i}$.
Moreover, we define $\mathcal{P}$ as the set of all players running nodes in the consensus protocol (i.e., $\mathcal{P}=\{p_i|\, \mathcal{N}_{p_i}\not = \emptyset \}$).
Therefore, $|\mathcal{N}|$ is not less than $|\mathcal{P}|.$ In particular, if a player has multiple nodes, $|\mathcal{N}|$ would be greater than $|\mathcal{P}|.$

For nodes to participate in the consensus protocol, they should possess specific resources, and their influence significantly depends on their resource power. 
The resource power in consensus protocols using PoW and PoS mechanisms is in the form of computational power and stakes, respectively.
Node $n_i\in \mathcal{N}$ possesses resource power $\alpha_{n_i}(>0)$.
Moreover, $\bm{\bar{\alpha}}$ denotes the vector of the resource power for all nodes (i.e., $\bm{\bar{\alpha}}=(\alpha_{n_i})_{{n_i}\in \mathcal{N}}$).
We also denote the resource power owned by player $p_i$ as $\alpha_{p_i}$ and the set of players with positive resource power as $\mathcal{P}_{\alpha}$ (i.e., $\mathcal{P}_{\alpha}=\{p_i|\, \alpha_{p_i}>0\}$).
Here, we note that these two sets, $\mathcal{P}_{\alpha}$ and $\mathcal{P},$ can be different because when players delegate their own power to others, they do not run nodes but possess the resource power (i.e., the fact that $\alpha_{p_i}>0$ does not imply that $\mathcal{N}_{p_i}\not=\emptyset$). 
For clarity, we describe a mining pool as an example. 
In the pool, there are workers and an operator, where the workers own their resource power but delegate it to the operator without running a full node. 
Therefore, pool workers belong to $\mathcal{P}_\alpha$ but not $\mathcal{P}$ while the operator belongs to both $\mathcal{P}_\alpha$ and $\mathcal{P}$.

In fact, the influence of player $p_i$ on the consensus protocol depends on the total resource power of the nodes run by the player rather than just its resource power $\alpha_{p_i}.$ 
Therefore, we define $EP_{p_i}$, the \textit{effective power} of player $p_i$ as $\sum_{n_i\in\mathcal{N}_{p_i}}\alpha_{n_i}.$ 
Again, considering the preceding example of mining pools, the operator's effective power is the sum of the resource power of all pool workers while the workers have zero effective power.
The maximum and $\delta$-th percentile of $\{EP_{p_i} |\, p_i\in\mathcal{P}\}$ are denoted by $EP_{\tt max}$ and $EP_{\delta}$, respectively, and $\bm{\bar{\alpha}_{\mathcal{N}_{p_i}}}$ represents a vector of the resource power of the nodes owned by player $p_i$ (i.e., $\bm{\bar{\alpha}_{\mathcal{N}_{p_i}}}=(\alpha_{n_i})_{n_i\in \mathcal{N}_{p_i}}$). 
Note that $EP_{\tt max}$ and $EP_{100}$ are the same.
In addition, we consider the average time to generate one block as a \textit{time unit} in the system.
We use the superscript $t$ to express time $t.$ For example, $\alpha^t_{n_i}$ and $\bm{\bar{\alpha}^t}$ represent the resource power of node $n_i$ at time $t$ and the vector of the resource power possessed by the nodes at time $t$, respectively.

\smallskip
\noindent\textbf{Incentive system. } To incentivize players to participate in the consensus protocol, the blockchain system must have an incentive system.
The incentive system would assign rewards to nodes, depending on their resource power. Here, we define the utility function $U_{n_i}(\alpha_{n_i}, \bm{\bar{\alpha}_{-n_i}})$ of the node $n_i$ as the expected net profit per time unit,   
where $\bm{\bar{\alpha}_{-n_i}}$ represents the vector of other nodes' resource power and the net profit indicates earned revenues with all costs subtracted. 
Specifically, the utility function $U_{n_i}(\alpha_{n_i}, \bm{\bar{\alpha}_{-n_i}})$ of node $n_i$ can be expressed as 
\vspace{-2mm}
\begin{equation*}
\resizebox{.95\hsize}{!}{$
    U_{n_i}=E[R_{n_i}\,|\,\bm{\bar{\alpha}}]=
\begin{cases}
\sum_{R_{n_i}} R_{n_i}\times\Pr(R_{n_i}|\,\bm{\bar{\alpha}}) & \text{if } R_{n_i} \text{ is discrete}\\
\int_{R_{n_i}} R_{n_i}\times\Pr(R_{n_i}|\,\bm{\bar{\alpha}}) &\text{otherwise,}
\end{cases}$} \vspace{-1mm} 
\end{equation*}
where $R_{n_i}$ is a random variable with probability distribution $\Pr(R_{n_i}|\,\bm{\bar{\alpha}})$ for a given $\bm{\bar{\alpha}}.$ 
This equation for $U_{n_i}$ and $R_{n_i}$ indicates that $U_{n_i}$ is the arithmetic mean of the random variable $R_{n_i}$ for given $\bm{\bar{\alpha}}.$ 
In addition, while function $U_{n_i}$ indicates the expected net profit that node $n_i$ can earn for the time unit, random variable $R_{n_i}$ represents all possible values of the net profit that node $n_i$ can obtain for the time unit.
For clarity, we give an example of the Bitcoin system, whereby $R_{n_i}$ and $\Pr(R_{n_i}|\,\bm{\bar{\alpha}})$ are defined as:
\vspace{-1mm}
$$ R_{n_i}=\begin{cases}
12.5 \text{ BTC}-c_{n_i} & \text{if } n_i \text{ generates a block}\\
-c_{n_i} & \text{otherwise,}
\end{cases}$$
\vspace{-3mm}
$$\Pr(R_{n_i}=a|\,\bm{\bar{\alpha}})=\begin{cases}
\frac{\alpha_{n_i}}{\sum_{n_j\in \mathcal{N}}\alpha_{n_j}} & \text{if } a=12.5\text{ BTC}-c_{n_i}\\
1-\frac{\alpha_{n_i}}{\sum_{n_j\in \mathcal{N}}\alpha_{n_j}} & \text{otherwise,}
\end{cases} 
$$
where $c_{n_i}$ represents all costs associated with running node $n_i$ during the time unit.
This is because a node currently earns 12.5 BTC as the block reward, and the probability of generating a block is proportional to its computing resource. 
Moreover, $R_{n_i}$ cannot be greater than a constant $R_{\tt max},$ determined in the system. 
In other words, the system can provide nodes with a limited value of rewards at a given time. 
Indeed, the reward that a node can receive for a time unit cannot be infinity, and problems such as inflation would occur if the reward were significantly large. 

In addition, if nodes can receive more rewards when they have larger resource power, then players would increase their resources by spending a part of the earned profit. 
In that case, for simplicity, we assume that all players increase their resource power per earned net profit $R_{n_i}$ at rate $r$ every time.
For example, if a node earns a net profit $R_{n_i}^t$ at time $t,$ the node's resource power would increase by $r\cdot R_{n_i}^t$ after time $t.$

We also define the \textit{Sybil cost function} $C(\bm{\bar{\alpha}_{\mathcal{N}_{p_i}}})$ as an additional cost that a player should pay per time unit to run multiple nodes compared to the total cost of when those nodes are run by different players. 
The cost $C(\bm{\bar{\alpha}_{\mathcal{N}_{p_i}}})$ would be 0 if $|\mathcal{N}_{p_i}|$ is 1 (i.e., the player $p_i$ runs one node). 
Moreover, the case where $C(\bm{\bar{\alpha}_{\mathcal{N}_{p_i}}})>0$ for any set $\mathcal{N}_{p_i}$ such that $|\mathcal{N}_{p_i}|>1$ indicates that the cost for one player to run $M(>1)$ nodes is always greater than the total cost for $M$ players each running one node.
Note that this definition does not just imply that it is expensive to run many nodes, the cost of which is usually referred to as Sybil costs in the consensus protocol~\cite{sybil}, 
\textit{this function implies that the total cost for running multiple nodes depends on whether one player runs those nodes.}

Finally, we assume that all players are rational. Thus, they act in the system for higher utility. 
More specifically, if there is a coalition of players in which the members can earn a higher profit, they delegate their power to form such a coalition (formally, it is referred to as a cooperative game). 
In addition, if it is more profitable for a player to run multiple nodes as opposed to one node, the player would run multiple nodes.

\begin{table}[ht] 
\centering
\caption{List of parameters.}
\vspace{-2mm}
\label{tab:par}
\begin{small}
\renewcommand{\tabcolsep}{1pt}
\begin{tabular}{|>{\centering\arraybackslash}m{2cm}|m{6cm}|}
\hline
\textbf{Notation} & \multicolumn{1}{c|}{\textbf{Definition}} \\ \hline\hline
$p_i$ & Player of index $i$ \\ \hline 
$\mathcal{P}$ &  The set of players running nodes in the consensus protocol \\ \hline
$n_i$ & Node of index $i$ \\ \hline
$\mathcal{N}$ & The set of nodes in the consensus protocol \\ \hline
$\mathcal{N}_{p_i}$ & The set of nodes owned by $p_i$ \\ \hline
$\alpha_{n_i}$, $\alpha_{p_i}$ & The resource power of node $n_i$ and player $p_i$ \\ \hline
$\bm{\bar{\alpha}}$ & The vector of resource power $\alpha_{n_i}$ for all nodes \\ \hline
$\mathcal{P}_{\alpha}$ & The set of players with positive resource power \\ \hline
$EP_{p_i}$ &   The effective power of nodes run by $p_i$    \\ \hline
$EP_{\tt max}, EP_{\delta}$ &    The maximum and $\delta$-th percentile of effective power of players running nodes   \\ \hline
$\bm{\bar{\alpha}_{\mathcal{N}_{p_i}}}$ & The vector of resource power of nodes run by $p_i$ \\ \hline
$\alpha^t_{n_i}$ &   The resource power of $n_i$ at time $t$ \\ \hline
\vspace{1mm}$\bm{\bar{\alpha}^t}$ &   The vector of resource power at time $t$ \\ \hline
$\bm{\bar{\alpha}_{-n_i}}$ & The vector of resource power of nodes other than $n_i$ \\ \hline
$U_{n_i}(\alpha_{n_i}, \bm{\bar{\alpha}_{-n_i}})$ & Utility function of $n_i$ \\ \hline
$R_{n_i}$ &    Random variable for a net reward of $n_i$ per time unit   \\ \hline
$R_{\tt max}$ & The maximum value of random variable $R_{n_i}$ \\ \hline
$r$ & Increasing rate of resource power per the net profit \\ \hline
$C(\bm{\bar{\alpha}_{\mathcal{N}_{p_i}}})$ &    Sybil cost function of $p_i$    \\ \hline
\end{tabular}\end{small}
\end{table}
\section{\large Conditions for full decentralization}
\label{sec:condition}

In this section, we study the circumstances under which a high level of decentralization can be achieved. 
To this end, we first formally define $(m, \varepsilon, \delta)$-decentralization and introduce the sufficient conditions of an incentive system that will allow a blockchain system to achieve  $(m, \varepsilon, \delta)$-decentralization. 
Then, based on these conditions, we find such an incentive system. 

\subsection{Full Decentralization}

The level of decentralization largely depends on two elements: the number of players running nodes in a consensus protocol and the distribution of effective power among the players. 
In this paper, full decentralization refers to the case where a system satisfies that 1) the number of players running nodes is as large as possible and 2) the distribution of effective power among the players is even. 
Therefore, if a system does not satisfy one of these requirements, it cannot become fully decentralized. 
For example, in the case where only two players run nodes with the same resource power, only the second requirement is satisfied.
As another example, a system may have many nodes run by independent players with the resource power being biased towards a few nodes. 
Then, in this case, only the first requirement is satisfied. 
Clearly, both of these cases have poor decentralization.
Note that, as described in Section~\ref{sec:background}, blockchain systems based on a peer-to-peer network can be manipulated by partial players who possess in excess of 50\% or 33\% of the effective power.
Next, to reflect the level of decentralization, we formally define $(m, \varepsilon, \delta)$-decentralization as follows. 

\begin{definition}[$(m, \varepsilon, \delta)$-Decentralization] \noindent 
\label{def:perfect}
For $1\leq m,$ $0\leq\varepsilon,$ and $0\leq \delta\leq 100$, a system is $\bm{(m, \varepsilon, \delta)}$\textbf{-decentralized} if it satisfies that 
\begin{enumerate}
    \item The size of $\mathcal{P}$ is not less than $m$ (i.e., $|\mathcal{P}|\geq m$),  
    \item 
    The ratio between the effective power of the richest player, $EP_{\tt max}$, and the $\delta$-th percentile player, $EP_{\delta}$, is less than or equal to $1+\varepsilon$ (i.e., $\frac{EP_{\tt max}}{EP_\delta}\leq 1+\varepsilon$).
\end{enumerate}
\end{definition}

In Def.~\ref{def:perfect}, the first requirement indicates that not only there are players that possess resources, but also that at least $m$ players should run their own nodes. 
In other words, too many players do not combine into one node (i.e., many players do not delegate their resources to others.). 
Note that delegation decreases the number of players running nodes in the consensus protocol.
The second requirement ensures an even distribution of the effective power among players running nodes. 
Specifically, for the richest and the $\delta$-th percentile players running nodes, the gap between their effective power should be small. 
According to Def.~\ref{def:perfect}, it is evident that as $m$ increases and $\varepsilon$ and $\delta$ decrease, the level of decentralization increases. 
Therefore, $(m,0,0)$-decentralization for a sufficiently large $m$ indicates full decentralization where there is a sufficiently large number of independent players and everyone has the same power. 

\subsection{Sufficient Conditions}
\label{subsec:condition}

Next, we introduce four sufficient conditions of an incentive system that will allow a blockchain system to achieve $(m, \varepsilon, \delta)$-decentralization with probability 1. 
We first revisit the two requirements of $(m, \varepsilon, \delta)$-decentralization.
For the first requirement in Def.~\ref{def:perfect}, the size of $\mathcal{N}$ should be greater than or equal to $m$ because the size of $\mathcal{P}$ is never greater than that of $\mathcal{N}.$ 
This can be achieved by assigning rewards to at least $m$ nodes. This approach is presented in Condition~\ref{con:resource} (GR-$m$). 
In addition, it should not be more profitable for too many players to combine into a few nodes than it is when they run their nodes directly. 
If delegating is more profitable than not delegating, many players with resource power would delegate their power to a few players, resulting in $|\mathcal{P}| < m.$
Condition~\ref{con:coll} (ND-$m$) states that it should not be more profitable for nodes run by independent (or different) players to combine into fewer nodes when the number of all players running nodes is not greater than $m$. 

\begin{condition}[\textbf{Giving Rewards (GR-$m$)}]
At least $m$ nodes should earn net profit.
Formally, for any $\bar{\bm{\alpha}},$ $|\mathcal{N}^+|\geq m,$ where
$$\mathcal{N}^{+}=\{n_i\in \mathcal{N}\,|\,U_{n_i}(\alpha_{n_i}, \bm{\bar{\alpha}_{-n_i}})>0\}.$$
\vspace{-2mm}
\label{con:resource}
\end{condition}

This condition states that some players can earn a reward by running a node, which makes the number of existing nodes equal to or greater than $m.$
Meanwhile, if the system does not give net profit, rational players would not run a node because the system requires a player to possess a specific resource (i.e., $\alpha_{n_i}>0$) in order to run a node unlike other peer-to-peer systems such as Tor. 
Simply put, players should invest their resource power elsewhere for higher profits instead of participating in a consensus protocol with no net profit, which is called an opportunity cost~\cite{need}. 
As a result, to reach $(m,\delta,\epsilon)$-decentralization, it is also necessary for a system to give net profit to some nodes.

\begin{condition}[\textbf{Non-Delegation (ND-$m$)}] \label{con:coll}
Nodes run by different players do not combine into fewer nodes unless the number of all players running their nodes is greater than $m.$
Before defining it formally, we denote a set of nodes run by different players by $\mathcal{S}^d$. That is, 
for any $n_i, n_j\in\mathcal{S}^d$, the two players running $n_i$ and $n_j$ are different.
We also let $s^d$ denote a proper subset of $\mathcal{S}^d$ such that 
$|\mathcal{P}(\mathcal{N}\backslash \mathcal{S}^d \cup s^d)|< m$, where  
$$\mathcal{P}(\mathcal{N}\backslash \mathcal{S}^d \cup s^d)=\{p_i\in\mathcal{P}\,|\,\exists n_i\in (\mathcal{N}\backslash \mathcal{S}^d \cup s^d) \text{ s.t. } n_i\in\mathcal{N}_{p_i}\}.$$
Then, for any set of nodes $\mathcal{S}^d,$ 

\vspace{-2mm}
\begin{eqnarray}
    &&  \hspace*{-10mm}\sum_{n_i\in \mathcal{S}^d}U_{n_i}(\alpha_{n_i},\bm{\bar{\alpha}_{-n_i}})\geq\notag\\ &&\hspace*{-12mm}  \max_{\substack{s^d \varsubsetneq \mathcal{S}^d \\
    \bm{\bar{\alpha}_d}\in s^d_\alpha}}\Bigl\{\sum_{\alpha_{n_i}\in\bm{\bar{\alpha}_d}}U_{n_i}(\alpha_{n_i},\bm{\alpha_{-n_i}^{-}}(\mathcal{S}^d\symbol{92} s^d))\Bigr\}, \label{eq:coll}
\end{eqnarray}
where,
\vspace{-2mm}
$$s^d_\alpha=\Bigl\{\bm{\bar{\alpha}_d}=(\alpha_{n_i})_{n_i\in s^d}\,\Big|\,\sum_{\alpha_{n_i}\in \bm{\bar{\alpha}_d}}\alpha_{n_i}= \sum_{n_i\in \mathcal{S}_d}\alpha_{n_i}\Bigr\}, 
$$
and 
$\bm{\alpha_{-n_i}^-}(\mathcal{S}^d\symbol{92} s^d)=(\alpha_{n_j})_{n_j\not\in \mathcal{S}^d\symbol{92} s^d, n_j\not=n_i}.$ 
\end{condition}

The set $\mathcal{P}(\mathcal{N}\backslash \mathcal{S}^d \cup s^d)$ represents all players running nodes that do not belong to $\mathcal{S}^d\backslash s^d$. 
In Eq.~\eqref{eq:coll}, the left-hand side represents the total utility of the nodes in $\mathcal{S}^d$ that are individually run by different players. 
Here, given that $\mathcal{S}^d \subseteq \mathcal{N},$ we note that $\bm{\bar{\alpha}_{-n_i}}$ includes the resource power of the nodes in $\mathcal{S}^d$ except for node $n_i.$  
The right-hand side represents the maximum total utility of the nodes in $s^d$ when the nodes in $\mathcal{S}^d$ are combined into fewer nodes belonging to $s^d$ by delegation of resource power of players. 
Note that $|s^d| < |\mathcal{S}^d|$ because $s^d \varsubsetneq \mathcal{S}^d.$ 
Therefore, Eq.~\eqref{eq:coll} indicates that the utility in the case where multiple players delegate their power to fewer players is not greater than that for the case where the players directly run nodes.
As a result, ND-$m$ prevents delegation that results in the number of players running nodes being less than $m,$ and the first requirement of $(m,\varepsilon,\delta)$-decentralization can be met when GR-$m$ and ND-$m$ hold. 

Next, we consider the second requirement in Def.~\ref{def:perfect}. 
One way to achieve an even distribution of effective power among players is to cause the system to have an even resource power distribution among nodes while each player has only one node. 
Note that in this case where each player has only one node, 
an even distribution of their effective power is equivalent to an even resource power distribution among nodes. 
Condition~\ref{con:sybil} (NS-$\delta$) states that, for any player with above the $\delta$-th percentile effective power, running multiple nodes is not more profitable than running one node. 
In addition, to reach a state where the richest and the $\delta$-th percentile nodes possess similar resource power, the ratio between the resource power of these two nodes should \textit{converge in probability} to a value of less than $1+\varepsilon.$
This is presented in Condition~\ref{con:density} (ED-$(\varepsilon, \delta)$).

\begin{condition}[\textbf{No Sybil nodes (NS-$\delta$)}] \label{con:sybil}
For any player with effective power not less than $EP_{\delta}$, participation with multiple nodes is not more profitable than participation with one node.
Formally, for any player $p_i$ with effective power $\alpha\geq EP_{\delta},$ 

\vspace{-2mm}
\begin{align}
    \max_{\substack{\{\mathcal{N}_{p_i}:\,|\mathcal{N}_{p_i}|>1\}\\ \bm{\bar{\alpha}_{\mathcal{N}_{p_i}}}\in\mathcal{S}_\alpha^{p_i}}}&\Bigl\{\sum_{\alpha_{n_i}\in\bm{\bar{\alpha}_{\mathcal{N}_{p_i}}}} U_{n_i}\left(\alpha_{n_i},\bm{\alpha_{-n_i}^{+}}(\mathcal{N}_{p_i})\right)-C(\bm{\bar{\alpha}_{\mathcal{N}_{p_i}}})\Bigr\}\notag
    \\ \vspace{-2mm} 
    &\leq U_{n_j}(\alpha_{n_j}=\alpha,\bm{\bar{\alpha}_{-\mathcal{N}_{p_i}}}), \label{eq:sybil}
\end{align}

\noindent where node $n_j\in\mathcal{N}_{p_i},$ the set $\small \bm{\bar{\alpha}_{-\mathcal{N}_{p_i}}}=(\alpha_{n_k})_{n_k\not\in\mathcal{N}_{p_i}}$,
$\bm{\alpha_{-n_i}^{+}}(\mathcal{N}_{p_i})=\bm{\bar{\alpha}_{-\mathcal{N}_{p_i}}}\Vert (\alpha_{n_k})_{n_k\in \mathcal{N}_{p_i}, n_k\not=n_i}$, and

$$\mathcal{S}_\alpha^{p_i}=\Bigl\{\bm{\bar{\alpha}_{\mathcal{N}_{p_i}}}=(\alpha_{n_i})_{n_i\in \mathcal{N}_{p_i}}\Big|\sum_{\alpha_{n_i}\in \bm{\bar{\alpha}_{\mathcal{N}_{p_i}}}}\alpha_{n_i}= \alpha\Bigr\}. \vspace{-1mm}
$$
\end{condition}
In Eq.~\eqref{eq:sybil}, the left and right-hand sides represent the maximum utility of the case where a player runs multiple nodes of which the total resource power is $\alpha,$ and the utility of the case where the player runs only one node $n_j$ with resource power $\alpha,$ respectively.
Therefore, Eq.~\eqref{eq:sybil} indicates that a player with equal to or greater than the $\delta$-th percentile effective power can earn the maximum utility when running one node.

\begin{condition}[\textbf{Even Distribution (ED-$(\varepsilon, \delta)$)}]
The ratio between the resource power of the richest and the $\delta$-th percentile nodes should \textbf{converge in probability} to a value less than $1+\varepsilon.$
Formally, when $\alpha_{\tt max}^t$ and $\alpha_{\delta}^t$ represent the maximum and the $\delta$-th percentile of $\{\alpha_{n_i}^t|n_i\in\mathcal{N}^t\}$, respectively, 
\begin{equation*}
    \lim_{t\rightarrow\infty} {\Pr}\Bigl[\,\frac{\alpha_{\tt max}^t}{\alpha_{\delta}^t}\leq 1+\varepsilon \Bigr]=1.
\vspace{-3mm}
\end{equation*}
\label{con:density}
\end{condition}

The above condition indicates that when enough time is given, the ratio between the resource power of the richest and the $\delta$-th percentile nodes reaches a value less than $1+\varepsilon$ with probability 1. 
We note that $\alpha^t_{n_i}$ changes over time, depending on the behavior of each player. 
In particular, if it is profitable for a player to increase their effective power, $\alpha^t_{n_i}$ would be a random variable related to $R_{n_i}^t$ because a player would reinvest part of their net profit $R_{n_i}^t$ to increase their resources. 
More specifically, in that case, $\alpha^t_{n_i}$ increases to $\alpha^t_{n_i}+rR_{n_i}^t$ after time $t$ as described in Section~\ref{sec:model}.

As a result, these four conditions allow blockchain systems to reach $(m, \varepsilon, \delta)$-decentralization with probability 1, as is presented in the following theorem. 
The proof of the theorem is omitted because it follows the above logic.

\begin{theorem}
For any initial state, a system satisfying GR-$m$, ND-$m$, NS-$\delta$, and ED-$(\varepsilon, \delta)$ converges in probability to $(m,\varepsilon,\delta)$-decentralization.
\end{theorem}

\subsection{\normalsize Possibility of Full Decentralization in Blockchain}
\label{subsec:possible}

To determine whether blockchain systems can achieve full decentralization, we study the existence of an incentive system satisfying these four conditions for a sufficiently large $m$, $\delta=0$, and $\varepsilon=0.$ 
We provide an example of an incentive system that satisfies the four conditions, thus allowing full decentralization to be achieved. 

It is also important to increase the total resource power involved in the consensus protocol from the perspective of security. 
This is because if the total resource power involved in the consensus protocol is small, an attacker can easily subvert the system. 
Therefore, to prevent this, we construct $U_{n_i}(\alpha_{n_i}, \bm{\bar{\alpha}_{-n_i}})$ as an increasing function of $\alpha_{n_i}$, which implies that players continually increase their resource power. 
In addition, we construct random variable $R_{n_i}$ with probability $\Pr(R_{n_i} | \bm{\bar{\alpha}})$ as follows: 
\begin{align}
R_{n_i}=&
\begin{cases}
B_r &\text{if $n_i$ generates a block}\\
0  &\text{otherwise}
\end{cases},\label{eq:br}\\
\Pr(R_{n_i}=a\,|\,\bm{\bar{\alpha}})=&
\begin{cases}
\frac{\sqrt{\alpha_{n_i}}}{\sum_{n_j\in\mathcal{N}}\sqrt{\alpha_{n_j}}} & \text{if }a=B_r\label{eq:prob}\\
1-\frac{\sqrt{\alpha_{n_i}}}{\sum_{n_j\in\mathcal{N}}\sqrt{\alpha_{n_j}}}  &\text{otherwise}
\end{cases},
\end{align}
\begin{equation}
U_{n_i}(\alpha_{n_i}, \bm{\bar{\alpha}_{-n_i}})=\frac{B_r\cdot \sqrt{\alpha_{n_i}}}{\sum_{n_j\in\mathcal{N}}\sqrt{\alpha_{n_j}}},\label{eq:example}     
\end{equation}
where the superscript $t$ representing time $t$ is omitted for convenience.
This incentive system indicates that when a node generates a block, it earns the block reward $B_r,$ and the probability of generating a block is proportional to the square root of the node's resource power.
Under these circumstances, we can easily check that the utility function $U_{n_i}$ is a mean of $R_{n_i}.$

Next, we show that this incentive system satisfies the four conditions. 
Firstly, the utility satisfies GR-$m$ for any $m$ because it is always positive. 
ND-$m$ is also satisfied because the following equation is satisfied:
This can be easily proven by using the fact that the utility is a concave function. 

\vspace{-2mm}
\begin{equation*}
    \sum_{i=1}^{m}U_{n_i}(\alpha_{n_i},\bm{\bar{\alpha}_{-n_i}})> U_{n_i}\Bigl(\sum_{i=1}^{m}\alpha_{n_i}\,\Bigg|\,(\alpha_{n_j})_{j>m}\Bigr)
\end{equation*}

Thirdly, to make NS-0 true, we can choose a proper Sybil cost function $C$ of Eq.~\eqref{eq:sybil}, which satisfies the following:
\begin{small}
\begin{equation*}
    \sum_{i=1}^{M}U_{n_i}(\alpha_{n_i},\bm{\bar{\alpha}_{-n_i}})- U_{n_i}\Bigl(\sum_{i=1}^{M}\alpha_{n_i}\,\Big|\,(\alpha_{n_j})_{j>M}\Bigr)\leq C((\alpha_{n_i})_{i\leq M})
\end{equation*}
\end{small}
Under this Sybil cost function, the players would run only one node. 
Finally, to show that this incentive system satisfies ED-$(0,0)$, we use the following theorem, whose proof is presented in Section~\ref{app:example}. 

\begin{theorem}
\label{thm:example}
Assume that $R_{n_i}$ is defined as follows:
\vspace{-1mm}
\begin{equation*}
    R_{n_i}=\begin{cases}
f(\bm{\bar{\alpha}}) &\text{if $n_i$ generates a block}\\
0  &\text{otherwise}
\end{cases},
\vspace{-1mm}
\end{equation*}
where $f: \mathbb{R}^{|\mathcal{N}|} \mapsto\mathbb{R}^+.$
If $U_{n_i}(\alpha_{n_i},\bm{\bar{\alpha}_{-n_i}})$ is a strictly increasing function of $\alpha_{n_i}$ and the following equation is satisfied for all $\alpha_{n_i}>\alpha_{n_j},$ ED-$(\varepsilon, \delta)$ is satisfied for all $\varepsilon$ and $\delta$.
\vspace{-1mm}
\begin{equation}
\frac{U_{n_i}(\alpha_{n_i}, \bm{\bar{\alpha}_{-n_i}})}{\alpha_{n_i}}<\frac{U_{n_j}(\alpha_{n_j}, \bm{\bar{\alpha}_{-n_j}})}{\alpha_{n_j}} \label{eq:density2}
\vspace{-1mm}
\end{equation}
On the contrary, if $U_{n_i}(\alpha_{n_i},\bm{\bar{\alpha}_{-n_i}})$ is a strictly increasing function of $\alpha_{n_i}$ and Eq.~\eqref{eq:density2} is not satisfied for all $\alpha_{n_i}>\alpha_{n_j},$ ED-$(\varepsilon, \delta)$ cannot be met for all $0\leq\varepsilon<\frac{EP_{\tt max}^0}{EP_{\delta}^0}-1$ and $0\leq\delta<100$. 
\end{theorem}

Thm.~\ref{thm:example} states that when the utility is a strictly increasing function of $\alpha_{n_i}$ and Eq.~\eqref{eq:density2} is satisfied under the assumption that the block reward is constant for a given $\bm{\bar{\alpha}},$ an even power distribution is achieved. 
Meanwhile, if Eq.~\eqref{eq:density2} is not met, the gap between rich and poor nodes cannot be narrowed. 
Specifically, for the case where $\frac{U_{n_i}(\alpha_{n_i},\bm{\bar{\alpha}_{-n_i}})}{\alpha_{n_i}}$ is constant, the large gap between rich and poor nodes can be continued\footnote{Formally speaking, the probability of achieving an even power distribution among players is less than 1, and in Thm.~\ref{thm:impossible}, we will address how small the probability is.}. 
Moreover, the gap would widen when $\frac{U_{n_i}(\alpha_{n_i},\bm{\bar{\alpha}_{-n_i}})}{\alpha_{n_i}}$ is a strictly increasing function of $\alpha_{n_i}$. 
In fact, here $\frac{U_{n_i}(\alpha_{n_i},\bm{\bar{\alpha}_{-n_i}})}{\alpha_{n_i}}$ can be considered as an increasing rate of resource power of a node.  
Thus, Eq.~\eqref{eq:density2} indicates that the resource power of a poor node increases faster than that of a rich node. 

Now, we describe why the incentive system defined by Eq.~\eqref{eq:br}, \eqref{eq:prob}, and \eqref{eq:example} satisfies ED-$(0,0)$. 
Firstly, Eq.~\eqref{eq:br} is a form of $R_{n_i}$ described in Thm.~\ref{thm:example}, and Eq.~\eqref{eq:example} implies that $U_{n_i}$ is a strictly increasing function of $\alpha_{n_i}$.
Therefore, ED-$(0,0)$ is met by Thm.~\ref{thm:example} because Eq.~\eqref{eq:example} satisfies Eq.~\eqref{eq:density2} for all $\alpha_{n_i}>\alpha_{n_j}.$  
As a result, the incentive system defined by Eq.~\eqref{eq:br}, \eqref{eq:prob}, and \eqref{eq:example} satisfies the four sufficient conditions,
\textit{implying that full decentralization is possible under a proper Sybil cost function $C.$}
Moreover, Thm.~\ref{thm:example} describes the existence of infinitely many incentive systems that can facilitate full decentralization. 
Interestingly, we have found that an incentive scheme similar to this is being considered by the Ethereum foundation, who have also indicated that \textit{real identity management} can be important~\cite{quadratic_voting}.
This finding is in accordance with our results.

\section{\large Impossibility of Full Decentralization in Permissionless Blockchains}
\label{sec:impossible}

In the previous section, we showed that blockchain systems can be fully decentralized under an appropriate Sybil cost function $C$, where the Sybil cost represents the additional costs for a player running multiple nodes when compared to the total cost for multiple players each running one node. 
In order for a system to implement the Sybil cost, we can easily consider real identity management where a trusted third party (TTP) manages the \textit{real identities} of players. 
When real identity management exists, it is certainly possible to implement a Sybil cost. 
However, the existence of a TTP contradicts the concept of decentralization, and thus, we cannot adopt such identity management for good decentralization.
Currently, it is not yet known how permissionless blockchains without such identity management can implement Sybil cost. 
In fact, many cryptocurrencies are based on permissionless blockchains, and many people want to design permissionless blockchains on the basis of their nature.   
Unfortunately, as far as we know, the Sybil cost function $C$ of all permissionless blockchains is currently zero. 
Taking this into consideration (i.e., $C=0$), we examine whether blockchains without Sybil costs can achieve good decentralization in this section.

\subsection{Almost Impossible Full Decentralization}

To determine whether it is possible for a system without Sybil costs to achieve full decentralization, we describe the following theorem.

\begin{theorem}
Consider a system without Sybil costs (i.e., $C=0$).
Then, the probability of the system achieving $(m, \varepsilon, \delta)$-decentralization is always less than or equal to 
\vspace{-1mm}
$$\max_{s\in \mathcal{S}} \Pr[\textup{System }s \textup{ reaches } (m, \varepsilon, \delta)\textup{-decentralization}], \quad \text{where}\vspace{-1mm}$$
$\mathcal{S}$ is the set of all systems satisfying GR-$|\mathcal{N}|$, ND-$|\mathcal{P}_\alpha|$, and NS-$0$. 
\label{thm:suff_nec}
\end{theorem}

\noindent 
GR-$|\mathcal{N}|$ means that all nodes can earn net profit, and the satisfaction of both ND-$|\mathcal{P}_\alpha|$ and NS-$0$ indicates that all players run only one node without delegating. 
\textbf{The above theorem implies that the maximum probability for a system, which satisfies GR-$|\mathcal{N}|$, ND-$|\mathcal{P}_\alpha|$, and NS-$0$, to reach $(m, \varepsilon, \delta)$-decentralization is equal to the global maximum probability.} 
Moreover, according to Thm.~\ref{thm:suff_nec}, there is a system satisfying GR-$|\mathcal{N}|$, ND-$|\mathcal{P}_\alpha|$, NS-$0$, and ED-$(\varepsilon, \delta)$ \textit{if and only if} there is a system that converges in probability to $(m, \varepsilon, \delta)$-decentralization. 
In other words, the fact that a system satisfying GR-$|\mathcal{N}|$, ND-$|\mathcal{P}_\alpha|$, NS-$0$, and ED-$(\varepsilon, \delta)$ should exist is \textbf{sufficient and necessary} to create a system converging in probability to $(m, \varepsilon, \delta)$-decentralization. 

The proof of Thm.~\ref{thm:suff_nec} is presented in Section~\ref{app:suff_nec}. 
In the proof, we use the fact that the system can optimally change the state (i.e., the effective power distribution among players above the $\delta$-th percentile) for $(m, \varepsilon, \delta)$-decentralization when the system can recognize the current state (i.e., the current effective power distribution among players above the $\delta$-th percentile). 
Then we prove that, to learn the current state, players above the $\delta$-th percentile should run only one node, or coalition of some players should be more profitable. 
In that case, to make a system most likely to reach $(m, \varepsilon, \delta)$-decentralization, resources of rich nodes should not increase through delegation of others. 
Considering this, we can derive Thm.~\ref{thm:suff_nec}.

According to Thm.~\ref{thm:suff_nec}, to find out if a system without Sybil costs can reach a high level of decentralization, it is sufficient to determine the maximum probability for a system satisfying GR-$|\mathcal{N}|$, ND-$|\mathcal{P}_\alpha|$, and NS-$0$ to reach $(m, \varepsilon, \delta)$-decentralization.
Therefore, we first find a utility function that satisfies GR-$|\mathcal{N}|$, ND-$|\mathcal{P}_\alpha|$, and NS-$0$ through the following lemma. 

\begin{lemma}
\label{lem:linear}
When the Sybil cost function $C$ is zero, GR-$|\mathcal{N}|$, ND-$|\mathcal{P}_\alpha|$, and NS-$0$ are met if and only if
\begin{equation}
U_{n_i}(\alpha_{n_i}, \bm{\bar{\alpha}_{-n_i}})=F\Bigl(\sum_{n_j\in\mathcal{N}}\alpha_{n_j}\Bigr)\cdot\alpha_{n_i}, \,\,\text{where}\,\, 
F: \mathbb{R}^+\mapsto\mathbb{R}^+. \label{eq:linear}
\end{equation}
\end{lemma}

\noindent Eq.~\eqref{eq:linear} implies that the utility function is linear when the total resource power of all nodes is given.
Under this utility (i.e., net profit), a player would run one node with its own resource power because delegation of its resource and running multiple nodes are not more profitable than running one node with its resource power.
Lem.~\ref{lem:linear} is proven using a proof by induction, and it is presented in 
Section~\ref{app:1}. 

We then consider whether Eq.~\eqref{eq:linear} can satisfy ED-$(\varepsilon, \delta)$. 
Note that when ED-$(\varepsilon, \delta)$ is satisfied, the probability of achieving $(m, \varepsilon, \delta)$-decentralization is 1. 
Therefore, it is sufficient to answer the following question: ``What is the probability of a system defined by Eq.~\eqref{eq:linear} to reach $(m, \varepsilon, \delta)$-decentralization?"
Thm.~\ref{thm:impossible} gives the answer by providing the upper bound of the probability. 
Before describing the theorem, we introduce several notations. 
Given that players, in practice, start running their nodes in the consensus protocol at different times, $\mathcal{P}$ would differ depending on the time. 
Thus, we use notations $\mathcal{P}^t$ and $\mathcal{P}_{\delta}^t$ to reflect this, where $\mathcal{P}_{\delta}^t$ is defined as: 
$$\mathcal{P}_{\delta}^t=\{p_i\in \mathcal{P}^t| EP_{p_i}^t\geq EP_{\delta}^t\}.
$$
That is, $\mathcal{P}_{\delta}^t$ indicates the set of all players who have above the $\delta$-th percentile effective power at time $t.$
Moreover, we define $\alpha_{\tt MAX}$ and $f_\delta$ as 
\vspace{-2mm}
$$\alpha_{\tt MAX}=\max\left\{\alpha_{p_i}^{t^0_i}\big|\, p_i \in \lim_{t\rightarrow\infty}\mathcal{P}^t\right\},$$
$$f_\delta=\min\biggl\{\frac{\alpha^{t^0_{ij}}_{p_i}}{\alpha_{p_j}^{t^0_{ij}}}\Bigg |\,\, p_i, p_j \in \lim_{t\rightarrow\infty}\mathcal{P}_{\delta}^t,\, {t^0_{ij}}=\max\{t^0_i, t^0_j\}\biggr\}, \vspace{-1mm}
$$
where $t_i^0$ denotes the time at which player $p_i$ starts to participate in a consensus protocol. 
The parameter $\alpha_{\tt MAX}$ indicates the initial resource power of the richest player among the players who remain in the system for a long time. 
Furthermore, $f_\delta$ represents the ratio between the $\delta$-th percentile and the largest initial resource power of the players who remain in the system for a long time.
Taking these notations into consideration, we present the following theorem. 

\begin{theorem}
When the Sybil cost function $C$ is zero, the following holds for any incentive system that satisfies Eq.~\eqref{eq:linear}:
\begin{equation}
    \lim_{t\rightarrow\infty}\Pr\left[\frac{EP_{\tt max}^t}{EP^t_{\delta}}\leq 1+\varepsilon\right]<G^\varepsilon\left(f_\delta, \frac{rR_{\tt max}}{\alpha_{\tt MAX}}\right),
    \label{eq:impossible}
\end{equation}
where $\lim_{f_\delta\rightarrow 0}G^\varepsilon(f_\delta, \frac{rR_{\tt max}}{\alpha_{\tt MAX}})$ and 
$\lim_{\alpha_{\tt MAX}\rightarrow \infty}G^\varepsilon(f_\delta, \frac{rR_{\tt max}}{\alpha_{\tt MAX}})$ are 0. 
Specifically, the function $G^\varepsilon(f_\delta, \frac{rR_{\tt max}}{\alpha_{\tt MAX}})$ is defined by Eq.~\eqref{eq:final2}. 
\label{thm:impossible}
\end{theorem}

This theorem implies that the probability of achieving $(m, \varepsilon, \delta)$-decentralization is less than $G^\varepsilon(f_\delta, \frac{rR_{\tt max}}{\alpha_{\tt MAX}})$. 
Here, note that $rR_{\tt max}$ represents the maximum resource power that can be increased by a player per time unit.
Given that $\lim_{f_\delta\rightarrow 0}G^\varepsilon(f_\delta, \frac{rR_{\tt max}}{\alpha_{\tt MAX}})=0,$ \textit{the upper bound would be smaller when the rich-poor gap in the current state is larger.}  
In addition, the fact that $\lim_{\alpha_{\tt MAX}\rightarrow \infty}G^\varepsilon(f_\delta, \frac{rR_{\tt max}}{\alpha_{\tt MAX}})$ implies that 
the greater the difference between the resource power of the richest player and the maximum value that can be increased by a player per time unit, the smaller the upper bound. 

In fact, to make a system more likely to reduce the rich-poor gap, poor nodes should earn a small reward with a high probability for some time, while rich nodes should get the reward $R_{\tt max}$ with a small probability. 
This is proved in the proof of Thm.~\ref{thm:impossible}, which is presented in Section~\ref{app:3}. 
Note that, in that case, rich nodes would rarely increase their resources, but poor nodes would often increase their resources.

To determine how small $G^\varepsilon(f_\delta, \frac{rR_{\tt max}}{\alpha_{\tt MAX}})$ is for a small value of $f_\delta$, we adopt a Monte Carlo method. 
This is because a large degree of complexity is required to compute a value of $G^\varepsilon(f_\delta, \frac{rR_{\tt max}}{\alpha_{\tt MAX}})$ directly. 
Fig.~\ref{fig:sim1} displays the value of $G^\varepsilon(f_\delta, \frac{rR_{\tt max}}{\alpha_{\tt MAX}})$ with respect to $f_\delta$ and $\varepsilon$ when $\frac{rR_{\tt max}}{\alpha_{\tt MAX}}$ is 0.1. 
For example, we can see that $G^0(10^{-4}, 0.1)$ is about $10^{-5},$ and this implies that a state where the ratio between resource power of the $\delta$-th percentile player and the richest player is $10^{-4}$ can reach $(m,0,\delta)$-decentralization with a probability less than $10^{-5}$ even if infinite time is given.  
Note that $\varepsilon=9, 99,$ and $999$ indicate that the effective power of the richest player is 10 times, 100 times, and 1000 times that of the $\delta$-th percentile player in $(m,\varepsilon,\delta)$-decentralization, respectively. 

Fig.~\ref{fig:sim1} shows that the probability of achieving $(m,\varepsilon,\delta)$-decentralization is smaller when $f_\delta$ and $\varepsilon$ are smaller.
From Fig.~\ref{fig:sim1}, one can see that the value of $G^\varepsilon(f_\delta, \frac{rR_{\tt max}}{\alpha_{\tt MAX}})$ is significantly small for a small value of $f_\delta$.
This result means that the probability of achieving good decentralization is close to 0 if there is a large gap between the rich and poor, and the resource power of the richest player is large (i.e., the ratio $\frac{rR_{\tt max}}{\alpha_{\tt MAX}}$ is not large\footnote{The ratio $\frac{rR_{\tt max}}{\alpha_{\tt MAX}}$ does not need to be small.}).
The values of $G^\varepsilon(f_\delta, \frac{rR_{\tt max}}{\alpha_{\tt MAX}})$ when $\frac{rR_{\tt max}}{\alpha_{\tt MAX}}$ is $10^{-2}$ and $10^{-4}$ are presented in Section~\ref{app:sim},  
and the values are certainly smaller than those presented in Fig.~\ref{fig:sim1}. 

\begin{figure}[ht]
    \centering
    \includegraphics[width=\columnwidth]{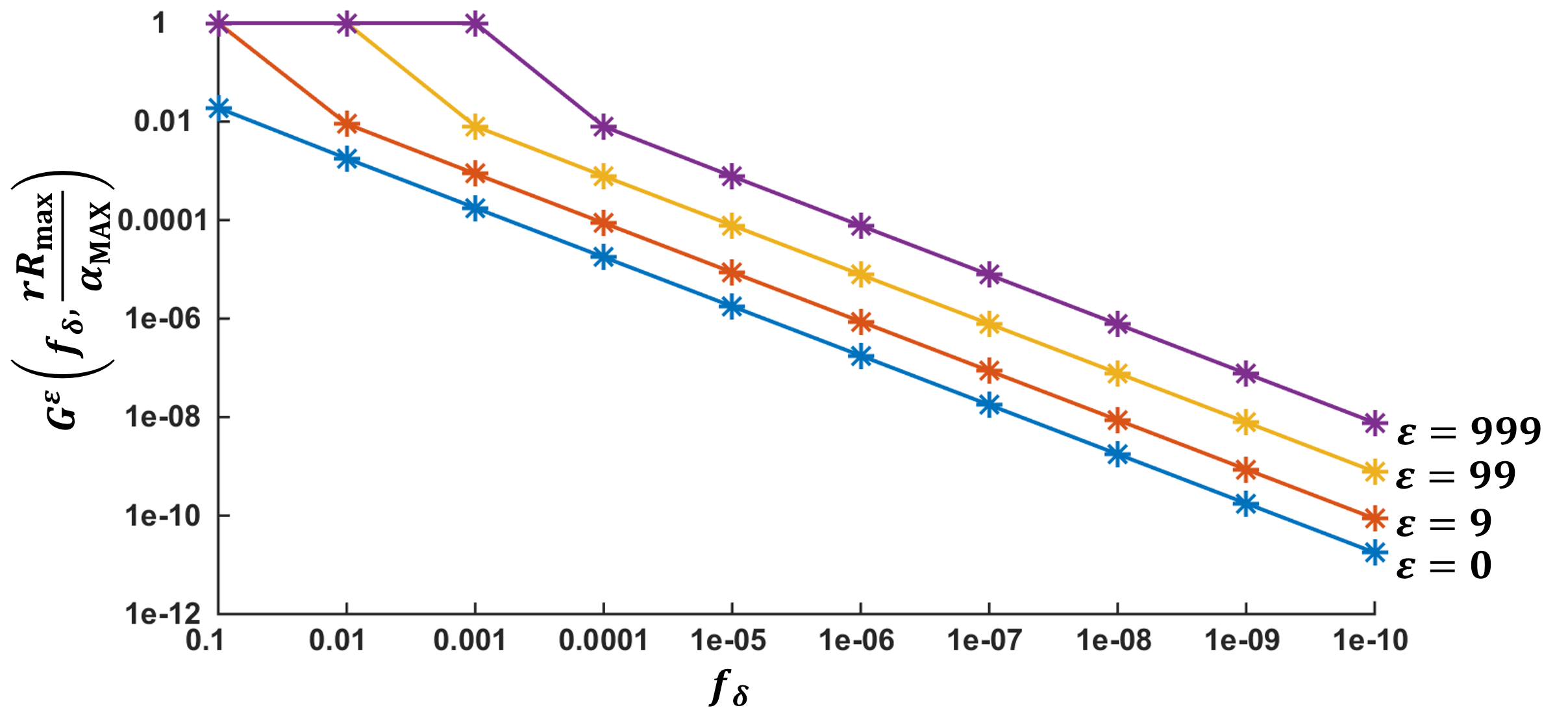}
    \vspace{-3mm}
    \caption{In this figure, when $\frac{rR_{\tt max}}{\alpha_{\tt MAX}}$ is 0.1, $G^\varepsilon(f_\delta, \frac{rR_{\tt max}}{\alpha_{\tt MAX}})$ ($y$-axis) is presented with respect to $f_\delta$ ($x$-axis) and $\varepsilon$. }
    \label{fig:sim1}
\end{figure}

To determine how small the ratio $f_\delta$ is at present, we use the hash rate of all users in Slush mining pool~\cite{slush} in Bitcoin as an example. 
We find miners with hash rates lower than 3.061 GH/s and greater than 404.0 PH/s at the time of writing.
Referring to these data, we can see that the ratio $f_0$ (i.e., the ratio between the resource power of the poorest and richest players) is less than $\frac{3.061\times 10^9}{404.0\times10^{15}}$ $(\approx 7.58\times10^{-9}).$
We also observe that the 15-th percentile and 50-th percentile hash rates are less than 5.832 TH/s and 25.33 TH/s, respectively. 
Therefore, the ratios $f_{15}$ and $f_{50}$ are less than approximately $1.44\times10^{-5}$ and $6.27\times10^{-5},$ respectively. 
This example indicates that the rich-poor gap is significantly large.
Moreover, we observe an upper bound of $\frac{rR_{\tt max}}{\alpha_{\tt MAX}}$ in the Bitcoin system. 
Given that the block reward is 12.5 BTC ($\approx \$65,504$), the maximum value of $rR_{\tt max}$ is approximately 384 TH. 
This maximum value can be derived, assuming that a player reinvests all earned rewards to increase their hash power. 
Then, an upper bound of $\frac{rR_{\tt max}}{\alpha_{\tt MAX}}$ would be $9.5\times 10^{-4}$, which is certainly less than the value of 0.1 used in Fig.~\ref{fig:sim1}. 
\textbf{As a result, Thm.~\ref{thm:impossible} implies that, currently, it is almost impossible for a system without Sybil costs to achieve good decentralization. 
In other words, the achievement of good decentralization in the consensus protocol and a non-reliance on a TTP, which are required for good decentralization of systems, contradict each other.}

\subsection{Intuition and Implication}

Here, we describe intuitively why a permissionless blockchain, which does not rely on any TTP, cannot reach good decentralization. 
Because a player with great wealth can possess more resources, the initial distribution of the resource power in a system depends significantly on the distribution of wealth in the real world when the system does not have any constraint of participation and can attract many players. 
Therefore, if wealth is equally distributed in the real world and many players are incentivized to participate in the consensus protocol, full decentralization can be easily achieved, even in permissionless blockchains where anyone can join without any permission processes.
However, according to many research papers and statistics, the rich-poor gap is significant in the real world~\cite{inequality,stone2015guide,sikorski2015rich}.
In addition, the wealth inequality is well known as one of the most glaring deficiencies in today's capitalist society, and resolving this problem is difficult.  

In a permissionless blockchain, players can freely participate without any restrictions, and large wealth inequality would appear initially. 
Therefore, for the system to achieve good decentralization, its incentive system should be designed to gradually narrow the rich-poor gap. 
To this end, we can consider the following incentive system: Nodes receive net profit in proportion to the square root of their resource power on average (e.g., Eq.~\eqref{eq:example}). 
This incentive system can result in the resource power distribution among nodes being more even (see Section~\ref{subsec:possible}). 
However, this alone cannot satisfy NS-$\delta$ when there is no Sybil cost (i.e., $C=0$). 
Therefore, to satisfy NS-$\delta$, we can establish that the expected net profit decreases when the number of existing nodes increases. 
For example, $B_r$ in Eq.~\eqref{eq:example} can be a decreasing function of the number of existing nodes. 
In this case, players with large resources would not run Sybil nodes because when they do so, their utilities decrease with the increase in the number of nodes.
However, this approach has a side effect in that players ultimately delegate their power to a few other players in order to earn higher profits. 
This is because this rational behavior on the part of the players decreases the number of nodes.
As a result, the above example intuitively describes that \textit{the four conditions are contradictory when a Sybil cost does not exist}\footnote{This does not imply the impossibility of full decentralization. It only implies that the probability of achieving full decentralization is less than 1.}, and whether a permissionless blockchain can achieve good decentralization depends completely on how wide the gap is between the rich and the poor in the real world.
This finding is supported by Thm.~\ref{thm:impossible}.

Conversely, if we can establish a method of implementing Sybil costs without relying on a TTP in blockchains, we would be able to resolve the contradiction between achieving good decentralization in the consensus protocol and non-reliance on a TTP. 
This allows for designing a blockchain that achieves good decentralization. 
We leave this as an open problem. 

\subsection{Question and Answer}

In this section, to further clarify the implications of our results, we present questions that academic reviewers or blockchain engineers have considered in the past and provide answers to them.

\smallskip
\noindent\textbf{[Q1] ``Creating more nodes does not increase your mining power, so why is this a problem?" }
Firstly, note that decentralization is significantly related to \textit{real identities}. 
That is, when the number of independent players is large and the power distribution among them is even, the system has good decentralization. 
In this paper, we \textit{do not claim} that the higher the number of Sybil nodes, the lower the level of decentralization.
We simply assert that a system should have knowledge of the current power distribution among players to achieve good decentralization, 
and a system without real identity management can know the distribution when each player runs only one node. 
Moreover, we prove that, to achieve good decentralization as far as possible, all players should run only one node (Thm.~\ref{thm:suff_nec}). 

\smallskip
\noindent\textbf{[Q2] ``Would a simple puzzle for registering as a block-submitter not be a possible Sybil cost, without identity management?" }
According to the definition of Sybil cost (Section~\ref{sec:model}), the cost to run one node should depend on whether a player runs another node.
More specifically, the cost to run one node for a player who has other nodes should be greater than that for a player with no other nodes. 
Therefore, the proposed scheme cannot constitute a Sybil cost. 
Again, note that the Sybil cost described in this paper is different from that usually mentioned in PoW and PoS systems~\cite{sybil}.

\smallskip
\noindent\textbf{[Q3] ``If mining power is delivered in proportion to the resources one has available (which would be an ideal situation in permissionless systems), achievement of good decentralization is clearly an impossibility. This seems rather self-evident." }
Naturally, a system would be centralized in its initial state because the rich-poor gap is large in the real world and only a few players may be interested in the system in the early stages.  
Considering this, our work investigates \textit{whether there is a mechanism to achieve good decentralization.}
Note that our goal is to reduce the gap between the effective power of the rich and poor, not the gap between their resource power. 
In other words, even if the rich possess significantly large resource power, the decentralization level can still be high if the rich participate in the consensus protocol with only part of their resource power and so not large effective power. 
To this end, we can consider a utility function, which is a decreasing function for a large input (e.g., a concave function). 
However, this function cannot still achieve good decentralization because it does not satisfy NS-$\delta.$
Note that, with a mechanism satisfying the four conditions, a system can \textit{always} reach good decentralization regardless of the initial state. 
Unfortunately, our finding is that there is no mechanism satisfying the four conditions, which implies that the probability of achieving good decentralization is less than 1. 
To make matters worse, Thm.~\ref{thm:impossible} states that the probability is bounded above by a value close to 0.
\textit{As a result, this implies that it is almost impossible for us to create a system with good decentralization without any Sybil cost, even if infinite time is given.}

\smallskip
\noindent\textbf{[Q4] ``I think when the rich invest a lot of money in a system, the system can become popular. So, if the large power of the rich is not involved in the system, can it become popular?"}
In this paper, we focus on the decentralization level in a consensus protocol, which performs a role as the government of a system. 
Therefore, good decentralization addressed in this paper implies a fair government rather than indicating that there are no rich or poor in the entire system. 
If the rich invest a lot of money in business (e.g., an application based on the smart contract) running on the system instead of the consensus protocol, the system may have a fair government and become popular. 
Indeed, the efforts to make a fair government also appear in the real world since people are extremely afraid of an unfair system in which the rich influence the government through bribes. 


\section{Protocol Analysis}
\label{app:protocol}

In this section, to determine if what condition each system satisfies or not, we analyze the incentive systems of the top 100 coins extensively according to the four conditions.
Based on this analysis, we can determine whether each system has a sufficient number of independent players and an even distribution of effective power among the players.
This analysis also describes what each blockchain system requires in order to achieve good decentralization. 

\subsection{Top 100 Coins}

Before analyzing the incentive systems based on the four conditions, we classified the top 100 coins 
in CoinMarketCap~\cite{coinmarket} according to their consensus protocols. 
Most of them use one of the following three consensus protocols: PoW, PoS, and DPoS.
Specifically, there exist 44 PoW, 22 PoS, and 11 DPoS coins. 
In addition, there are 15 coins that use other consensus protocols such as Federated Byzantine Agreement (FBA), Proof of Importance, Proof of Stake and Velocity~\cite{ren2014proof}, and hybrid.  
Furthermore, we classify five coins including XRP~\cite{schwartz2014ripple}, NEO~\cite{neo}, VeChain~\cite{vechain}, Ontology~\cite{ontology}, and GoChain~\cite{gochain} into permissioned systems. 
This is because in these systems, only players that are chosen by the coin foundation can run nodes in the consensus protocol. 
Finally, there exist one token, Huobi Token, and two cryptocurrencies that are non-operational, i.e., BitcoinDark and Boscoin. 
Table~\ref{tab:classification} summarizes the classification of the aforementioned top 100 coins.

\begin{small}
\begin{table*}[ht]
    \renewcommand{\tabcolsep}{1pt}
    \centering
    \caption{Classification of top 100 coins (Sep. 11, 2018)}
    \begin{tabular}{|>{\centering\arraybackslash}m{2cm}|m{14cm}|c|}
         \hline
         \cline{1-3} Consensus & \multicolumn{1}{c|}{Coins} & Count \\
         \hline
         \hline
         \cline{1-3} PoW & Bitcoin (1)~\cite{nakamoto2008bitcoin}, 
         Ethereum (2)~\cite{wood2014ethereum}, Bitcoin Cash (4)~\cite{bch}, Litecoin (7)~\cite{litecoin}, Monero (9)~\cite{monero}, Dash (10)~\cite{duffield2015dash}, IOTA (11)~\cite{popov2016tangle}, Ethereum Classic (13)~\cite{etc}, Dogecoin (18)~\cite{dogecoin}, Zcash (19)~\cite{zcash}, Bytecoin (21)~\cite{bytecoin}, Bitcoin Gold (22)~\cite{bitcoingold}, Decred (25)~\cite{decred}, Bitcoin Diamond (26)~\cite{diamond}, DigiByte (28)~\cite{digibyte}, Siacoin (33)~\cite{vorick2014sia}, Verge (34)~\cite{verge},
         Metaverse ETP (35)~\cite{metaverse}, Bytom (36)~\cite{bytom}, 
         MOAC (43)~\cite{moac}, Horizen (47)~\cite{viglione2017zen}, MonaCoin (51)~\cite{monacoin}, Bitcoin Private (52)~\cite{bitcoinprivate}, ZCoin (56)~\cite{zcoin}, Syscoin (60)~\cite{sidhu2017syscoin}, Electroneum (61)~\cite{electroneum}, Groestlcoin (64)~\cite{groestl}, Bitcoin Interest (67)~\cite{interest}, Vertcoin (70)~\cite{vertcoin}, Ravencoin (71)~\cite{fenton2018ravencoin}, Namecoin (72)~\cite{namecoin}, BridgeCoin (74)~\cite{bridgecoin}, SmartCash (75)~\cite{smart}, Ubiq (77)~\cite{ubiq}, DigitalNote (82)~\cite{digitalnote}, ZClassic (83)~\cite{zclassic}, Burst (85)~\cite{burst}, Primecoin (86)~\cite{king2013primecoin}, Litecoin Cash (90)~\cite{litecoincash}, Unobtanium (91)~\cite{unobtanium}, Electra (92)~\cite{electra}, Pura (96)~\cite{pura}, Viacoin (97)~\cite{viacoin}, Bitcore (100)~\cite{bitcore}
         & 44 \\
         \hline
         \cline{1-3} PoS & Cardano (8)~\cite{kiayias2017ouroboros}, Tezos (15)~\cite{goodman2014tezos}, Qtum (24)~\cite{dai2017smart}, Nano (29)~\cite{lemahieu2018nano}, Waves (31)~\cite{waves}, Stratis (37)~\cite{stratis}, Cryptonex (38)~\cite{cryptonex}, Ardor (42)~\cite{ardor}, Wanchain (44)~\cite{wanchain}, Nxt (50)~\cite{nxt}, PIVX (57)~\cite{pivx}, PRIZM (63)~\cite{prizm}, WhiteCoin (76)~\cite{whitecoin}, Blocknet (79)~\cite{blocknet}, Particl (80)~\cite{particl}, Neblio (81)~\cite{neblio}, BitBay (87)~\cite{bitbay}, GCR (89)~\cite{gcr}, NIX (93)~\cite{nix}, SaluS (94)~\cite{salus}, LEO (98)~\cite{leo}, ION (99)~\cite{ion}
          & 22  \\
         \hline
         \cline{1-3} DPoS & EOS (5)~\cite{eos}, TRON (12)~\cite{tron}, Lisk (20)~\cite{lisk}, BitShare (27)~\cite{bitshare}, Steem (32)~\cite{steem}, GXChain (48)~\cite{gxchain}, Ark (49)~\cite{ark}, WaykiChain (68)~\cite{waykichain}, Achain (84)~\cite{achain}, Asch (88)~\cite{asch}, Steem Dollars (95)~\cite{steem}
          & 11 \\
         \hline
         \cline{1-3} Others & Stellar (6)~\cite{mazieres2015stellar}, NEM (16)~\cite{nem}, ICON (30)~\cite{icon}, Komodo (39)~\cite{komodo}, ReddCoin (40)~\cite{ren2014proof}, Hshare (41)~\cite{hshare}, Nebulas (53)~\cite{nebulas}, Emercoin (54)~\cite{emercoin}, Elastos (55)~\cite{elastos}, Nexus (58)~\cite{nexus}, Byteball Bytes (59)~\cite{churyumov2016byteball}, Factom (62)~\cite{snow2015factom}, Skycoin (69)~\cite{skycoin}, Nexty (66)~\cite{nexty}, Peercoin (73)~\cite{peercoin}
          & 15 \\
         \hline
         \cline{1-3} Permissioned & XRP (3)~\cite{schwartz2014ripple}, NEO (14)~\cite{neo}, VeChain (17)~\cite{vechain}, Ontology (23)~\cite{ontology}, GoChain (65)~\cite{gochain}
         & 5 \\
         \hline
        \cline{1-3} Token & Huobi Token (45)
         & 1 \\
        \hline
         \cline{1-3} Non-operational & BitcoinDark (46), Boscoin (78)
          & 2 \\
        \hline
    \end{tabular}
    \label{tab:classification}
\end{table*}
\end{small}

\subsection{Analysis}

Next, we analyze the blockchain systems of the top 100 coins according to the four sufficient conditions. 
In this study, we focus on the analysis of the coins that use PoW, PoS, and DPoS mechanisms, which are the major consensus mechanisms of non-permissioned blockchains,
to identify which conditions are not currently satisfied in each system. 
If a system satisfies both GR-$m$ and ND-$m$, we can expect that many players participate in its consensus protocol and run nodes.
In addition, if the system satisfies both NS-$\delta$ and ED-$(\varepsilon, \delta)$, the effective power would be more evenly distributed among the players. 
Table~\ref{tab:incentive_analysis} presents the results of the analysis, where the black circle (\blackcircle) and the half-filled circle (\semifilled) indicate the full and partial satisfaction of the corresponding condition, respectively. 
The empty circle (\emptycircle) indicates that the corresponding condition is not satisfied at all.
In addition, we mark each coin system with a triangle (\tri) or an X (\xmark) depending on whether it partially implements or does not implement a Sybil cost, respectively. 
Here, partial Sybil cost means that the payment of the Sybil cost can be avoided by pretending that the multiple nodes run by one player are run by different players (i.e., players who have different real identities). 
Note that PoW, PoS, and DPoS coins cannot have perfect Sybil costs because they are non-permissioned blockchains. 
Even it is currently unknown as to how Sybil costs are implemented in blockchain systems without real identity management. 
We present detailed analysis results in the following sections.

\subsubsection{Proof of Work}
\label{subsubsec:pow}
Most PoW systems are designed to give nodes a block reward proportional to the ratio of the computational power of each node to the total power. 
In addition, there are electric bills that are dependent on the computational power, as well as the other costs associated with running a node, such as a large memory for the storage of blockchain data.
The cost required to run a node is, therefore, independent of the computational power. 
Considering this, we can express a utility (i.e., an expected net profit) $U_{n_i}(\alpha_{n_i}, \bm{\bar{\alpha}_{-n_i}})$ of node $n_i$ as follows: 

\begin{equation}
U_{n_i}(\alpha_{n_i}, \bm{\bar{\alpha}_{-n_i}})=B_r\cdot\frac{\alpha_{n_i}}{\sum_{n_j}\alpha_{n_j}}-c_1\cdot\alpha_{n_i}-c_2. \label{eq:pow}
\end{equation}
In Eq.~\eqref{eq:pow}, $B_r$ represents the block reward (e.g., 12.5 BTC in the Bitcoin system) that a node can earn for a time unit, and $c_1 (>0)$ and $c_2 (>0)$ represent the electric bill per computational power and the other costs incurred during the time unit, respectively.  
In particular, the cost $c_2$ is independent of the computational power. 
The values of the three coefficients, $B_r, c_1,$ and $c_2,$ determine whether the four conditions are satisfied. 

\begin{small}
\begin{table}[t]
\renewcommand{\tabcolsep}{2pt}
  \centering
  \caption{Analysis of incentive systems}
  \label{tab:incentive_analysis}
  \begin{tabular}{|c|c|c|c|c|c|c|}
    \hline
    \cline{1-7}
    Coin name&Con 1&Con 2&Con 3& Con 4& $N_{\tt dpos}$ & \makecell{Sybil cost}\\
    \hline
    \hline
    \multicolumn{7}{|c|}{\makecell{\textbf{PoW \& PoS coins}}} \\ 
    \hline
    \cline{1-7} All PoW\&PoS$\dagger$ & \semifilled & \emptycircle & \blackcircle & \emptycircle & $-$ & \xmark \\ 
    \hline
    \cline{1-7} IOTA & \emptycircle & \emptycircle & \blackcircle & \blackcircle & $-$ & \xmark \\ 
    \hline
    \cline{1-7} BridgeCoin & \emptycircle & \emptycircle & \blackcircle & \blackcircle & $-$ & \xmark \\
    \hline
    \cline{1-7} Nano & \emptycircle & \emptycircle & \blackcircle & \blackcircle & $-$ & \xmark \\ 
    \hline
    \cline{1-7} Cardano & \semifilled & \semifilled & \semifilled & \semifilled & $-$ & \xmark \\ 
    \hline
    \hline
    \multicolumn{7}{|c|}{\makecell{\textbf{DPoS coins}}} \\ \hline
    \cline{1-7} EOS & \semifilled & \semifilled & \hspace{1.5mm}\semifilled$^\star$ & \semifilled & 21 & \tri \\ 
     \hline
     \cline{1-7} TRON & \semifilled & \semifilled & \hspace{1.5mm}\semifilled$^\star$ & \semifilled & 27 & \tri \\
    \hline 
    \cline{1-7} Lisk & \semifilled & \semifilled & \semifilled & \semifilled & 101 & \xmark \\ 
    \hline 
    \cline{1-7} BitShare & \semifilled & \semifilled & \semifilled & \semifilled & 27& \xmark \\ 
    \hline 
    \cline{1-7} Steem & \semifilled & \semifilled & \hspace{1.5mm}\semifilled$^\star$ & \semifilled & 20 & \tri \\ 
    \hline
    \cline{1-7} GXChain & \semifilled & \semifilled & \semifilled & \semifilled & 21 & \xmark \\ 
    \hline
    \cline{1-7} Ark & \semifilled & \semifilled & \semifilled & \semifilled & 51 & \xmark \\ 
    \hline
    \cline{1-7} WaykiChain & \semifilled & \semifilled & \semifilled & \semifilled & 11 & \xmark \\ 
    \hline
    \cline{1-7} Achain & \semifilled & \semifilled & \semifilled & \semifilled & 99 & \xmark \\
    \hline
    \cline{1-7} Asch & \semifilled & \semifilled & \semifilled & \semifilled & 91 & \xmark \\ 
    \hline
    \cline{1-7} Steem Dollars & \semifilled & \semifilled & \hspace{1.5mm}\semifilled$^\star$ & \semifilled & 20 & \tri \\ 
    \hline
    \end{tabular}
    \begin{tabular}{p{\columnwidth}}
    \\
      $\dagger=$ except for IOTA, BridgeCoin, Cardano, and Nano; 
      \blackcircle $=$ fully satisfies the condition; \semifilled $=$ partially satisfies the condition; 
      \emptycircle $=$ does not satisfy the condition; 
    \tri $=$ has partial Sybil costs; \xmark $=$ does not have Sybil costs;
    \end{tabular}
      \vspace*{-4mm}
\end{table}
\end{small}

Firstly, for the system to satisfy GR-$m$ for any $m$, it should be able to assign rewards to nodes with small computational power.
Considering Eq.~\eqref{eq:pow} for appropriate values of $B_r,$ there is $\bm{\bar{\alpha}}=(\alpha_{n_i})_{n_i\in\mathcal{N}}$ such that $U_{n_i}(\alpha_{n_i}, \bm{\bar{\alpha}_{-n_i}})>0$ for all nodes $n_i.$  
However, there also exists $\alpha_{n_i}$ such that $U_{n_i}(\alpha_{n_i}, \bm{\bar{\alpha}_{-n_i}})<0$ for a given $\bm{\bar{\alpha}_{-n_i}}$, which implies that the PoW system cannot satisfy GR-$m$ for some values of $m$. 
For example, if $\sum_{n_j}\alpha_{n_j}$ is significantly large and $\alpha_{n_i}$ is small enough, Eq.~\eqref{eq:pow} would be negative because the first term of the right-hand side of Eq.~\eqref{eq:pow} is close to 0. 

We can observe this situation in practical PoW systems. 
In these systems, nodes can generate blocks using CPUs, GPUs, FPGAs, and ASICs, with computational power ranging from low at the CPU level to high at the ASIC level.
In particular, the value of $c_1$ decreases from CPUs to ASICs. 
In other words, ASICs have better efficiency than the others. 
Currently, PoW coins can be divided into ASIC-resistant coins and coins that allow ASIC miners. 
The latter (e.g., Bitcoin and Litecoin) allow miners to use ASIC hardware, which has rapidly increased their total computational power.  
However, as a side effect, CPU mining has become unprofitable because the electric bill for CPU miners is larger than their earned rewards. 
For this reason, several coins, such as Ethereum, were developed to resist ASIC miners; however, ASIC-resistant algorithms cannot be a fundamental solution. 
These algorithms only prevent the rapid growth of the total computational power; nodes with small computational power can still suffer losses.
For example, even though Ethereum has the ASIC-resistant algorithm, Ethash~\cite{ethash}, CPU miners cannot earn net profit by mining Ethereum~\cite{cpuminingeth}. 
Therefore, these PoW coins only partially satisfy GR-$m$ because there exists $\bm{\bar{\alpha}}$ such that $U_{n_i}(\alpha_{n_i}, \bm{\bar{\alpha}_{-n_i}})<0$ for some nodes $n_i$. 
As special cases, we consider IOTA and BridgeCoin, where there is no block reward because coin mining does not exist or has already been completed. 
These systems do not satisfy GR-$m$ at all because the utility $U_{n_i}$ is negative for all $\bm{\bar{\alpha}}$. 

In addition, PoW systems cannot satisfy ND-$m$. 
This is because when $m$ players run their own nodes, they must pay the additional cost of $(m-1)\cdot c_2$ as compared to the case in which they run only one node by cooperating with one another. 
This cooperation is commonly observed in the form of centralized mining pools. 
Of course, the variance of rewards can decrease when players join the mining pools, which may be another reason that many of them join these pools. 
However, although there are decentralized pools (e.g., P2Pool~\cite{p2pool} and SMARTPOOL~\cite{luu2017smart}) in which players can reduce the variance of rewards and run a full node, 
most players do not join these decentralized pools owing to the cost of running a full node\footnote{One can see that the percentage of resource power possessed by the decentralized pools is significantly small.}.

Meanwhile, for the aforementioned reason, the systems can satisfy NS-$\delta$. 
Finally, ED-$(\varepsilon, \delta)$ cannot be satisfied in PoW systems. 
Firstly, Eq.~\eqref{eq:pow} is a strictly increasing function of $\alpha_{n_i}$ for a proper value of $\sum_{n_j}\alpha_{n_j}$ and does not satisfy Eq.~\eqref{eq:density2}. 
Thus, according to Thm.~\ref{thm:example}, ED-$(\varepsilon, \delta)$ cannot be satisfied for the proper range of $\sum_{n_j}\alpha_{n_j}.$
In addition, for a significantly large value of $\sum_{n_j}\alpha_{n_j},$ all nodes would reduce their resource power since all of them suffer a loss regardless of their resource power. 
Note that this behavior does not affect the power distribution, which represents relative resource power. 
As a result, PoW systems with an incentive system defined by Eq.~\eqref{eq:pow} cannot satisfy ED-$(\varepsilon, \delta)$.
\textbf{Through this analysis of PoW systems, we expect that the current PoW systems have neither a sufficient number of independent players nor an even power distribution among the players.} 

Meanwhile, IOTA and Bridgecoin, which do not have any incentives, satisfy both NS-$\delta$ and ED-$(\varepsilon, \delta)$ as trivial cases because rational players would not run nodes.

\subsubsection{Proof of Stake}
In PoS systems, nodes receive block rewards proportional to their stake. 
Therefore, in these systems, we can express the utility $U_{n_i}$ as follows: 
\begin{equation}
\label{eq:pos}
U_{n_i}(\alpha_{n_i}, \bm{\bar{\alpha}_{-n_i}})=
B_r\cdot\frac{\alpha_{n_i}}{\sum_{j}\alpha_{n_j}}-c \quad\text{if } \alpha_{n_i} \geq S_b.
\end{equation}
$B_r$ and $c$ in Eq.~\eqref{eq:pos} represent the block reward that a node can earn for a time unit and the cost required to run one node, respectively. 
$S_b$ indicates the least amount of stakes required to run one node. 
Therefore, Eq.~\eqref{eq:pos} implies that only nodes with stakes above $S_b$ can be run and earn a reward proportional to their stake fraction. 

Similar to PoW systems, the systems only satisfy GR-$m$ for some $m$ (i.e., partially satisfy GR-$m$) because there exists a large value of $\sum\alpha_{n_j}$ such that $U_{n_i}(\alpha_{n_i}, \bm{\bar{\alpha}_{-n_i}})<0$ in PoS systems.
In addition, it is more profitable for multiple players to run one node when compared to running each different node.
For example, if a player has a stake below $S_b,$ rewards cannot be earned by running nodes in the consensus protocol. 
However, the player can receive a reward by delegating their stake to others. 
In addition, if multiple players run only one node, they can reduce the cost required to run nodes. 
Therefore, PoS systems do not satisfy ND-$m$. 
These behaviors are observed through PoS pools~\cite{pospools, stakeminers} or leased PoS~\cite{leased_pos} in practice.
This fact also implies that it is less profitable for one player to run multiple nodes than it is to run one node; thus, PoS systems satisfy NS-$\delta$. 
Finally, the system cannot satisfy ED-$(\varepsilon, \delta)$.  
To explain this, we should consider when $B_r$ is a constant and when it is not, where PIVX~\cite{pivx} is associated with the latter. 
If $B_r$ is a constant, the utility $U_{n_i}$ is a strictly increasing function of $\alpha_{n_i}.$ Because Eq.~\eqref{eq:density2} is not met, according to Thm.~\ref{thm:example}, this case cannot satisfy ED-$(\varepsilon, \delta)$. 
Meanwhile, in the PIVX system, $B_r$ is a decreasing function of  $\sum_{n_j}\alpha_{n_j}$ owing to the \textit{seesaw effect}~\cite{pivx}. 
Therefore, for a large value of $\sum_{n_j}\alpha_{n_j},$ nodes earn fewer rewards compared to the case when $\sum_{n_j}\alpha_{n_j}$ is small. 
In this case, there is an equilibrium where all nodes reduce their resource power for higher profits and, in addition, a strategy that allows a state to reach the equilibrium exists. 
This does not change the power distribution among nodes, which is only related to the relative resource power of the nodes. 
As a result, PIVX also does not satisfy ED-$(\varepsilon, \delta)$. 

As shown in Table~\ref{tab:incentive_analysis}, the results are similar to those for PoW coins. \textbf{Therefore, as with PoW coins, PoS coins would have a restricted number of independent players and a biased power distribution among them.} 
Note that we excluded Wanchain in this analysis because the specifications of its PoS protocol had not yet been provided at the time of writing~\cite{wanchain_announce}.  
Similar to IOTA and BridgeCoin, Nano does not provide incentives to run nodes. 
Therefore, the result of Nano is the same with IOTA and BridgeCoin.
In addition, Cardano is planning to implement an incentive system different from that of the usual PoS systems~\cite{brunjes2018reward}. 
The system has the goal that there should be $k$ nodes with similar resource power for a given $k$. 
In fact, this incentive system has a similar property to DPoS systems, which will be described in the following section.

\subsubsection{Delegated Proof of Stake}
\label{subsubsec:dpos}
DPoS systems are significantly different from PoW and PoS systems.
Unlike these systems, DPoS systems do not give nodes block rewards proportional to their resource power.
Instead, stake holders elect block generators through a voting process, where the voting power is proportional to the stake owned by the stake holders (i.e., voters). 
Then, the block generators have an equal opportunity to generate blocks and earn the same block rewards. 
Therefore, when we arrange $\bm{\bar{\alpha}}=\{\alpha_{n_i}|\,1\leq i \leq n\}$ in descending order, we can express the utility $U_{n_i}$ in DPoS systems as follows: 

\begin{equation}
\label{eq:dpos}
U_{n_i}(\alpha_{n_i}, \bm{\bar{\alpha}_{-n_i}})=
\begin{cases}
B_r-c & \text{ if } i\leq N_{\tt dpos}\\
\,-c & \text{ otherwise }
\end{cases},
\end{equation}
where $B_r$ is a block reward that a node can earn on average per a time unit, and $c$ represents the cost associated with running one node. 
In addition, $N_{\tt dpos}$ is a constant number given by the DPoS system. 
Eq.~\eqref{eq:dpos} implies that only $N_{\tt dpos}$ nodes with the most votes can earn rewards by generating blocks. 
However, not all DPoS systems have the same incentive scheme as Eq.~\eqref{eq:dpos}. 
For example, EOS with $N_{\tt dpos}=21$ gives small rewards to nodes ranked within the 100-th place~\cite{eos_incentive}. 
In addition, Steem with $N_{\tt dpos}=20$ randomly chooses one node, ranked outside the 20-th place, as a block generator~\cite{steem}. Thus, the system also gives small rewards to nodes ranked outside the 20-th position. 
In WaykiChain, the incentive system is significantly different from the typical incentive scheme used in DPoS systems because nodes with small votes can also earn non-negligible rewards~\cite{waykichain_incentive}. 
Although incentive systems different from Eq.~\eqref{eq:dpos} exist, we describe the analysis results of the DPoS coins with respect to Eq.~\eqref{eq:dpos} because their properties are similar. 

Firstly, the DPoS system attracts players who can obtain high voting power because it provides them with a block reward. 
Meanwhile, rational players who are unable to obtain high voting power cannot earn any rewards. 
Therefore, the system partially satisfies GR-$m$. 
Moreover, it is rational for multiple players with small stakes to delegate their stakes to one player by voting for that player, which is why this system is called a \textit{delegated} PoS system. 
Meanwhile, rational players with high stakes would run their own nodes by voting for themselves. 
For example, if two players have sufficiently high stakes and run two nodes, they can earn a total value of $2(B_r-c)$ as net profit. 
However, when they run only one node, they earn only $B_r-c.$
As a result, it is rational only for those players with small stakes to delegate all their resource power to others, and ND-$m$ is partially satisfied. 

Next, we consider NS-$\delta$. As described above, a player with small stakes would not run multiple nodes, but instead would delegate their stakes to others. 
However, for a player with high stakes, this is divided into two cases: when weak identity management exists and when it does not. 
Weak identity management implies that nodes should reveal a pseudo-identity such as a public URL or a social ID.
Firstly, in the latter case, the player with high stakes can earn a higher profit by running multiple nodes because there is no Sybil cost. 
Therefore, a DPoS system in which identity management does not exist partially satisfies NS-$\delta$ because only players with high stakes would run multiple nodes. 
Meanwhile, when the system includes weak identity management, voters can partially recognize whether different nodes are run by the same player. 
Therefore, the voters can avoid voting for these multiple nodes run by the same player.
This is because they may want to achieve good decentralization in the system, and recognize that the system can be centralized towards a few players when they vote for the nodes controlled by the same player. 
This means that it is not more profitable for one player to run multiple nodes than it is to run one node (i.e., Sybil costs exist), and these DPoS systems satisfy NS-$\delta$. 
Note that because the identity management is not perfect, a rich player can still run multiple nodes by \textit{creating multiple pseudo-identities.} 
Thus, strictly speaking, systems with weak identity management still do not fully satisfy NS-$\delta$.
However, because it is certainly more expensive for a rich player to run multiple nodes in systems with weak identity management when compared to systems without identity management, we mark such systems with \semifilled$^\star$ for NS-$\delta$ in Table~\ref{tab:incentive_analysis} to distinguish them from systems with no identity management. 

Currently, EOS, TRON, Steem, and Steem Dollars have weak identity management. 
EOS and TRON propose some requirements in order for a player to register as a delegate, even though the requirements are not official~\cite{eos_producer, eos_requirement, tron_voter}.
These requirements include a public website, technical specifications, and team members, which can be regarded as pseudo-identities. 
Steem and Steem Dollars provide the information for activities in Steemit 
~\cite{steemit, steem_witness, steem_witness2}.
Note that Steem and Steem Dollars are transacted under the same consensus protocol. 

Finally, we examine whether the DPoS system satisfies ED-$(\varepsilon, \delta)$. 
To this end, we consider two cases: when a delegate shares the block reward with voters (e.g., TRON~\cite{tron_share} and Lisk~\cite{lisk_share}), and when they do not share (e.g., EOS\footnote{A debate exists as to whether delegates should share their rewards with voters or not. Currently, some delegates have announced that they will share the rewards~\cite{eos_share1, eos_share2}.}). 
In the former case, if a delegator receives $V$ votes, the voters who voted for the delegator can, in general, earn reward $\frac{B_r}{V}-f$ per vote, where $f$ represents a fee per vote paid to the delegator. 
Here, note that the larger $V$ is, the smaller the reward is that the voters earn.
Therefore, when voters are biased towards a delegator, some voters can move their vote to other delegators for higher profits. 
In the latter case, delegators would increase their effective power by voting for themselves with more stakes to maintain or increase their ranking, and Eq.~\eqref{eq:density2} is met in the DPoS system. 
This allows for a more even power distribution among the delegators. 
Therefore, in the two cases, the power distribution among delegators can converge to an even distribution. 
However, the wealth gap between nodes obtaining small voting power and nodes obtaining high voting power would increase, thus implying that the probability of poor nodes generating blocks becomes smaller gradually.
Consequently, the DPoS system partially satisfies ED-$(\varepsilon, \delta)$. 

Table~\ref{tab:incentive_analysis} presents the analysis result for the DPoS coins according to the four conditions. 
\textbf{DPoS systems may potentially ensure even power distribution among a limited number of players when weak identity management exists.
However, the system has a limited number of players running nodes in the consensus protocol, 
which implies that they cannot have good decentralization.}

\section{Empirical Study}
\label{app:level}
In this section, we extensively collect and quantitatively analyze the data for the PoW, PoS, and DPoS coins not only to establish the degree to which they are currently centralized, but also to validate the protocol analysis result and four conditions. 
Through this study, we empirically observe rational behaviors, such as the delegation of resources to a few players and the running of multiple nodes, which eventually hinder full decentralization.

\subsection{Methodology}

We considered the past 10,000 blocks before Oct. 15, 2018, for PoW and PoS systems and the past 100,000 blocks before Oct. 15, 2018, for DPoS systems since some DPoS systems do not renew the list of block generators within 10,000 blocks. 
We parsed addresses of block generators from each blockchain explorer for 68 coins. 
Because IOTA and Nano are based on DAG technology instead of blockchain technology, the analysis of these two systems will be presented in Section~\ref{subsubsec:DAG}.

We determined the number $NB_{A_i}$ of blocks generated by each address $A_i$, where the set of addresses is denoted by $\mathcal{A}.$
We then constructed a dataset $\mathcal{NB}=\left\{NB_{A_i}| A_i\in\mathcal{A}\right\}$ and rearranged $\mathcal{NB}$ and $\mathcal{A}$ in descending order of $NB_{A_i}$. 
Then, we analyzed the dataset using three metrics: the total number of addresses ($|\mathcal{A}|$), the Gini coefficient, and the entropy ($H$), where the Gini coefficient is the most commonly used term to measure wealth distribution in economics. 
Regarding the security in blockchain systems, it is meaningful to analyze not only how evenly the total power is distributed but also how evenly 50\% and 33\% of the power are distributed, since a player who possesses above 50\% or 33\% power can execute attacks as described in Section~\ref{sec:background}.
Therefore, we also measure the level of decentralization for 50\% and 33\% power in the systems using the three metrics. 
To do this, we first define subset $\mathcal{A}^{x}$ of the address set $\mathcal{A},$ and subset $\mathcal{NB}^x$ of the data set $\mathcal{NB}$ as follows:
\begin{equation*}
    \begin{aligned}
    \mathcal{A}^{x}&=\Bigl\{A_i\in \mathcal{A}\,\Big|\,\,\frac{\sum_{j=1}^{i-1}NB_{A_i}}{\sum_{A_i\in\mathcal{A}}NB_{A_i}}<x\Bigr\},\\
\mathcal{NB}^{x}&=\{NB_{A_i}|\, A_i \in \mathcal{A}^x\},
    \end{aligned}
\end{equation*}
where $0\leq x \leq 1$. Here, note that if $x$ is 0, the two sets are empty, and if $x$ is 1, they are equal to $\mathcal{A}$ and $\mathcal{NB},$ respectively.
The Gini coefficient and the entropy ($H$) are then defined as:

$$Gini(\mathcal{NB}^x)= \frac{\sum_{A_i, A_j\in\mathcal{A}^x}|NB_{A_i}-NB_{A_j}|}{2|\mathcal{A}|\sum_{A\in\mathcal{A}^x}NB_{A_i}}, $$

$$H(\mathcal{NB}^x)=-\sum_{A_i\in\mathcal{A}^x}\frac{NB_{A_i}}{\sum_{A_i\in\mathcal{A}^x}NB_{A_i}}\log_2 \Bigl(\frac{NB_{A_i}}{\sum_{A_i\in\mathcal{A}^x}NB_{A_i}}\Bigr).
$$
The Gini coefficient measures the spread of the data set $\mathcal{NB}^x.$ If the deviation of $\mathcal{NB}^x$ is small, its value is close to 0. Otherwise, the coefficient is close to 1.  
The entropy depends on both $|\mathcal{A}^x|$ and the Gini coefficient. 
As $|\mathcal{A}^x|$ gets larger and the Gini coefficient gets smaller, the entropy gets larger. 
Therefore, entropy implicitly represents the level of decentralization, and large entropy implies a high level of decentralization.  
In fact, because a player can have multiple addresses, the measured values may not accurately represent the actual level of decentralization.
However, since entropy is a concave function of the relative ratio of $NB_{A_i}$ to the total number of generated blocks  (i.e., $\frac{NB_{A_i}}{\sum_{A_i\in\mathcal{A}^x}NB_{A_i}}$), the results show an upper bound of the current level of decentralization. 
Therefore, if the measured values of entropy are low, the current systems do not have good decentralization.

\begin{small}
\begin{table}[ht]
\renewcommand{\tabcolsep}{1.5pt}
  \centering
  \caption{PoW Coins}
  \label{tab:pow_level}
  \begin{tabular}{|c|c|c|c|c|c|c|c|c|c|c|}
    \hline
    \multicolumn{1}{|c|}{}&\multicolumn{3}{c|}{100 \%}&\multicolumn{3}{c|}{50\%}&\multicolumn{3}{c|}{33\%}\\
    \cline{2-10}
    Coin name&$|\mathcal{A}|$&Gini&H& $|\mathcal{A}^{\frac{1}{2}}|$&Gini$^{\frac{1}{2}}$ &H$^{\frac{1}{2}}$& $|\mathcal{A}^{\frac{1}{3}}|$&Gini$^{\frac{1}{3}}$ &H$^{\frac{1}{3}}$\\
    \hline
    \hline
    \cline{1-10} Bitcoin & 62 & 0.8192 & 3.89 & 4 & 0.1143 & 1.98 & 3 & 0.1103 & 1.57 \\ 
     \hline
     \cline{1-10} Ethereum & 65 & 0.8634 & 3.38 & 3 & 0.1402 & 1.53 & 2 & 0.0415 & 1.00 \\ 
    \hline
    \cline{1-10} Bitcoin Cash & 15& 0.5729 & 3.06 & 3 & 0.2572 & 1.51 & 2 & 0.0859 & 0.12 \\ 
    \hline 
    \cline{1-10} Litecoin & 35 & 0.8094 & 3.10 & 3 & 0.0176 & 1.58 & 2 & 0.0146 & 1.00 \\ 
    \hline 
    \cline{1-10} Dash & 109 & 0.9005 & 3.79 & 4 & 0.2050 & 1.90 & 2 & 0.0770 & 0.98 \\ 
    \hline 
    \cline{1-10} Ethereum Classic & 83 & 0.8916 & 3.17 & 2 & 0.1538 & 0.93 & 1 & 0 & 0 \\ 
    \hline
    \cline{1-10} Dogecoin & 400 & 0.8686 & 4.95 & 4 & 0.2123 & 1.89 & 2 & 0.1098 & 0.96 \\ 
    \hline
    \cline{1-10} Zcash & 75 & 0.8932 & 3.36 & 3 & 0.0615 & 1.52 & 2 & 0.0546 & 0.15 \\ 
    \hline
    \cline{1-10} Bitcoin Gold & 29 & 0.8585 & 2.36 & 1 & 0 & 0 & 1 & 0 & 0 \\ 
    \hline
    \cline{1-10} Decred & 17 & 0.7751 & 2.33 & 2 & 0.1471 & 0.35 & 2 & 0.1471 & 0.35 \\ 
    \hline
    \cline{1-10} Bitcoin Diamond & 16 & 0.7401 & 2.44 & 2 & 0.0707 & 0.99 & 2 & 0.0707 & 0.99 \\ 
    \hline
    \cline{1-10}  DigiByte & 125 & 0.7791 & 5.09 & 7 & 0.2724 & 2.63 & 4 & 0.1879 & 1.90 \\ 
    \hline
    \cline{1-10} \rowcolor[gray]{0.8} Siacoin & 1406 & 0.8582 & 3.02 & 2 & 0.1551 & 0.98 & 2 & 0.1551 & 0.98 \\ 
    \hline
    \cline{1-10} Verge & 82 & 0.7261 & 4.92 & 8 & 0.1762 & 3.03 & 5 & 0.0820 & 2.46 \\ 
    \hline
    \cline{1-10} Metaverse ETP & 36 & 0.7964 & 3.25 & 3 & 0.2914 & 1.49 & 2 & 0.1927 & 0.97 \\ 
    \hline
    \cline{1-10} Bytom & 12 & 0.7978 & 1.54 & 1 & 0 & 0 & 1 & 0 & 0 \\ 
    \hline
    \cline{1-10} MOAC & 28 & 0.7067 & 3.46 & 3 & 0.2330 & 1.53 & 2 &  0.1615 & 0.98 \\ 
    \hline
    \cline{1-10} Horizen & 96 & 0.9109 & 3.39 & 3 & 0.0882 & 1.56 & 2 & 0.0189 & 1.00 \\ 
    \hline
    \cline{1-10} MonaCoin & 44 & 0.8185 & 3.39 & 3 & 0.1373 & 1.56 & 2 & 0.0920 & 0.99 \\ 
    \hline
    \cline{1-10} Bitcoin Private & 135 & 0.8557 & 4.48 & 5 & 0.1260 & 2.28 & 3 & 0.0766 & 1.57 \\ 
    \hline
    \cline{1-10} Zcoin & 361 & 0.9562 & 1.75 & 1 & 0 & 0 & 1 & 0 & 0 \\ 
    \hline
    \cline{1-10} \rowcolor[gray]{0.8} Syscoin & 5979 & 0.2529 & 10.37 & 1978 & 0.5055 & 6.78 & 644 & 0.7571 & 3.61 \\ 
    \hline
    \cline{1-10} Groestlcoin  & 10 & 0.4969 & 2.67 & 3 & 0.3408 & 1.47 & 2 & 0.4110 & 0.45 \\ 
    \hline
    \cline{1-10} Bitcoin Interest & 19 & 0.7267 & 2.66 & 2 & 0.3109 & 0.70 & 1 & 0 & 0 \\ 
    \hline
    \cline{1-10} Vertcoin & 60 & 0.8390 & 3.61 & 3 & 0.2639 & 1.40 & 2 & 0.2064 & 0.87 \\
    \hline
    \cline{1-10} Ravencoin & 71 & 0.8014 & 4.12 & 4 & 0.2057 & 1.90 & 2 & 0.0488 & 0.99 \\ 
    \hline
    \cline{1-10} \rowcolor[gray]{0.8} Namecoin & 3390 & 0.5693 & 8.00 & 49 & 0.8613 & 2.52 & 3 & 0.1913 & 1.48 \\ 
    \hline
    \cline{1-10} BridgeCoin & 1 & 0 & 0 & 1 & 0 & 0 & 1 & 0 & 0 \\ 
    \hline
    \cline{1-10} SmartCash & 7& 0.6885 & 1.47 & 1 & 0 & 0 &1 & 0 & 0\\ 
    \hline
    \cline{1-10} Ubiq & 34 & 0.8440 & 2.58 & 1 & 0 &0 &1 & 0 &0 \\ 
    \hline
    \cline{1-10} Zclassic & 41 & 0.7762 & 3.54 & 3 & 0.2394 & 1.43 & 2 & 0.0899 & 0.98 \\
    \hline
    \cline{1-10} Burst & 143 & 0.9054 & 3.45 & 2 & 0.2473 & 0.82 &1 & 0 &0 \\ 
    \hline
    \cline{1-10}\rowcolor[gray]{0.8} Prime & 7477 & 0.2525 & 10.46 & 2476 & 0.5048 & 6.63 & 809 & 0.7565 & 3.22 \\ 
    \hline
    \cline{1-10} Litecoin Cash & 33 & 0.6788 & 3.78 & 5 & 0.0711 & 2.31 & 3 & 0.0557 & 1.58 \\ 
    \cline{1-10} Unobtanium & 30 & 0.9463 & 0.89 & 1& 0 &0 &1 & 0 &0 \\ 
    \hline
    \cline{1-10} \rowcolor[gray]{0.8} Electra & 1268  & 0.6608 & 8.34 & 46 & 0.5262 &4.87 &12 & 0.2622 & 3.53 \\ 
    \hline
    \cline{1-10} Pura & 19 & 0.6521 & 3.08 & 3 & 0.0778 & 1.58 & 2 & 0.0905 & 0.99 \\ 
    \hline
    \cline{1-10} Viacoin & 33 & 0.9141 & 1.78 & 1 & 0 & 0 & 1 & 0 & 0 \\ 
    \hline
    \cline{1-10} Bitcore & 116 & 0.9337 & 3.11 & 2 & 0.0956 & 0.97 & 2 & 0.0956 & 0.97 \\ 
    \hline
  \end{tabular}
 \end{table}
\end{small}

\subsection{Data Analysis}

\subsubsection{Quantitative analysis}

Tables~\ref{tab:pow_level}, \ref{tab:pos_level}, and \ref{tab:dpos_level} represent the results for PoW, PoS, and DPoS coins, respectively. 
Coins such as Monero~\cite{monero}, Bytecoin (21)~\cite{bytecoin}, Electroneum~\cite{electroneum}, DigitalNote~\cite{digitalnote}, and PIVX~\cite{pivx} include \textit{stealth} or \textit{anonymous} addresses that cannot be traced. Therefore, we excluded them from this data analysis. 
As such, we conduct the data analysis for 39 PoW, 19 PoS, and 10 DPoS coins in this section. 
In addition, the datasets for certain coins have \textit{too much noise} to establish their actual level of decentralization because they include \textit{short-lived addresses}, which are only used for a short time and discarded later. 
We shaded these coins in gray in the tables. 
Moreover, in the case of Cardano and WaykiChain, only trusted nodes are allowed to participate in the protocol at the time of writing since they have not yet implemented their public consensus protocols~\cite{kiayias2017ouroboros, cardano_notopen, waykichain_plan}. 
In the tables, we shaded to these coins in blue. 
We do not consider these shaded coins when interpreting the results below.

\begin{small}
\begin{table}[t]
\renewcommand{\tabcolsep}{0.5pt}
  \centering
  \caption{PoS Coins}
  \label{tab:pos_level}
  \begin{tabular}{|c|c|c|c|c|c|c|c|c|c|c|}
    \hline
    \multicolumn{1}{|c|}{}&\multicolumn{3}{c|}{100 \%}&\multicolumn{3}{c|}{50\%}&\multicolumn{3}{c|}{33\%}\\
    \cline{2-10}
    Coin name&$|\mathcal{A}|$&Gini&H& $|\mathcal{A}^{\frac{1}{2}}|$&Gini$^{\frac{1}{2}}$ &H$^{\frac{1}{2}}$& $|\mathcal{A}^{\frac{1}{3}}|$&Gini$^{\frac{1}{3}}$ &H$^{\frac{1}{3}}$\\
    \hline
    \hline
    \cline{1-10} \rowcolor{blue!20} Cardano & 7& 0.0039 & 2.81 & 3 & 0.0083 & 2.11 & 2 & 0.0111 & 1.50 \\ 
    \hline
    \cline{1-10} Tezos & 245 & 0.8391 & 5.54 & 9 & 0.1061 & 3.13 & 6 & 0.1168 & 2.55 \\ 
     \hline
     \cline{1-10} Qtum & 1853 & 0.7404 & 8.07 & 32 & 0.5923 & 4.12 & 7 & 0.2512 & 2.69 \\ 
    \hline 
    \cline{1-10} Waves & 110 & 0.8606 & 4.24 & 4 & 0.1545 & 1.93 & 3 & 0.1628 & 1.51 \\ 
    \hline 
    \cline{1-10} Stratis & 527 & 0.8113 & 6.78 & 20 & 0.2626 & 4.15 & 10 & 0.2007 & 3.23 \\ 
    \hline
    \cline{1-10} Cryptonex & 122 & 0.9231 & 3.30 & 4 & 0.0103 & 2.00 & 3 & 0.0078 & 1.58 \\ 
    \hline 
    \cline{1-10} Ardor & 247 & 0.8623 & 4.91 & 8 & 0.5376 & 2.20 & 6 & 0.4554 & 1.95 \\ 
    \hline
    \cline{1-10} Nxt & 165 & 0.9150 & 3.30 & 2 & 0.0326 & 1.00 & 2 & 0.0326 & 1.00 \\ 
    \hline
    \cline{1-10} PRIZM & 82 & 0.8672 & 3.68 & 4 & 0.0053 & 2.00 & 3 & 0.0022 & 1.58 \\ 
    \hline 
    \cline{1-10} Whitecoin & 239 & 0.6273 & 6.84 & 32 & 0.2954 & 4.75 & 15 & 0.2740 & 3.71 \\ 
    \hline
    \cline{1-10} Blocknet & 584 & 0.7965 & 6.54 & 10 & 0.3891 & 2.96 & 4 & 0.1778 & 1.92 \\ 
    \hline
    \cline{1-10} \rowcolor[gray]{0.8} Particl & 1801 & 0.5989 & 9.48 & 141 & 0.4436 & 6.56 & 48 & 0.3713 & 5.21 \\ 
    \hline
    \cline{1-10} Neblio & 1177 & 0.8258 & 6.00 & 5 & 0.4523 & 1.74 & 2 & 0.3123 & 0.70 \\ 
    \hline
    \cline{1-10} Bitbay & 313 & 0.7839 & 6.02 & 9 & 0.3075 & 2.94 & 4 & 0.0890 & 1.97 \\ 
    \hline
    \cline{1-10} GCR & 263 & 0.8192 & 5.84 & 11 & 0.2515 & 3.43 & 6 & 0.1779 & 2.68 \\ 
    \hline
    \cline{1-10} \rowcolor[gray]{0.8} NIX & 1130 & 0.4520 & 9.62 & 255 & 0.2224 & 7.86 & 135 & 0.2180 & 6.96 \\ 
    \hline
    \cline{1-10} SaluS & 27 & 0.6974 & 3.41 & 4 & 0.1577 & 1.97 & 3 & 0.1342 & 1.56 \\ 
    \hline
    \cline{1-10} \rowcolor[gray]{0.8} Leocoin & 879 & 0.5988 & 8.72 & 106 & 0.3639 & 6.33 & 44 & 0.3268 & 5.16 \\ 
    \hline
    \cline{1-10} ION & 287 & 0.8998 & 4.24 & 2 & 0.0335 & 1.00 & 2 & 0.0335 & 1.00 \\ 
    \hline
  \end{tabular}
\end{table}
\end{small}

\begin{small}
\begin{table}[t]
\renewcommand{\tabcolsep}{0.5pt}
  \centering
  \caption{DPoS Coins}
  \label{tab:dpos_level}
  \begin{tabular}{|c|c|c|c|c|c|c|c|c|c|c|}
    \hline
    \multicolumn{1}{|c|}{}&\multicolumn{3}{c|}{100 \%}&\multicolumn{3}{c|}{50\%}&\multicolumn{3}{c|}{33\%}\\
    \cline{2-10}
    Coin name&$|\mathcal{A}|$&Gini&H& $|\mathcal{A}^{\frac{1}{2}}|$&Gini$^{\frac{1}{2}}$ &H$^{\frac{1}{2}}$& $|\mathcal{A}^{\frac{1}{3}}|$&Gini$^{\frac{1}{3}}$ &H$^{\frac{1}{3}}$\\
    \hline
    \hline
    \cline{1-10} EOS & 22 & 0.0447 & 4.43 & 11 & 0.0002 & 3.46 & 7 & 0.0003 & 2.81 \\ 
     \hline
     \cline{1-10} TRON & 28& 0.0358 & 4.79 & 14 & 0.0009 & 3.81 & 9 & 0.0008 & 3.17 \\ 
    \hline 
    \cline{1-10} Lisk & 101& 0.0023 & 6.66 & 51 & 0.0011 & 5.67 & 34 & 0.0010 & 5.09 \\ 
    \hline 
    \cline{1-10} BitShare & 27& 0.0009 & 4.75 & 14 & 0.0007 & 3.81 & 9 & 0.0003 & 3.17 \\ 
    \hline 
    \cline{1-10} Steem & 140 & 0.8324 & 4.68 & 11 & 0.0002 & 3.46 & 7 & 0.0002 & 2.81 \\ 
    \hline
    \cline{1-10} GXChain & 21& 0.0328 & 4.39 & 10 & 0.0016 & 3.32 & 7 & 0.0013 & 2.81 \\ 
    \hline
    \cline{1-10} Ark & 52 & 0.0200 & 5.69 & 25 & 0.0005 & 4.64 & 16 & 0.0003 & 4.00 \\ 
    \hline
    \cline{1-10} \rowcolor{blue!20} WaykiChain & 11 & 0.1688 & 3.27 & 5 & 0.0021 & 2.32 & 4 & 0.0022 & 2.00 \\ 
    \hline
    \cline{1-10} Achain & 99 & 0.0018 & 6.63 & 49 & 0.0009 & 5.61 & 32 & 0.0008 & 5.00 \\ 
    \hline
    \cline{1-10} Asch & 92 & 0.0769 & 6.50 & 42 & 0.0267 & 5.39 & 27 & 0.0184 & 4.75 \\
    \hline
  \end{tabular}
\end{table}
\end{small}

Firstly, one can see that there is an insufficient number of block generators in PoW, PoS, and DPoS coins. 
In particular, $|\mathcal{A}^{\frac{1}{2}}|$ and $|\mathcal{A}^{\frac{1}{3}}|$ in PoW and PoS are quite small. 
However, PoS systems generally have more block generators than PoW systems. 
This may be because the pool concept is more common in PoW systems. 
Indeed, most PoS systems are currently in an early stage, and some of them do not have staking pools yet. 
For example, Qtum does not have staking pools at the time of writing and has a relatively large number of block generators compared to others\footnote{Note that the value of $|\mathcal{A}|$ in Table~\ref{tab:pos_level} does not accurately represent the number of block generators because a player can create multiple addresses.}. 
This fact certainly allows the level of decentralization in Qtum to increase. 
However, we cannot assure that this situation will continue. There have already been some requests for pools and intentions to run a business for Qtum staking pools~\cite{qtum_request, qtum_request2, qtum_pool, qtum_pool2}. 
Considering this, we expect that staking pools will become more popular in PoS systems. 
Note that Tezos and Waves, already allowing the delegation of stakes, have a smaller number of block generators. 
PoW protocols also did not originally have a pool concept. However, mining pools have become significantly popular, and most miners currently join mining pools. 
As a special case, BridgeCoin, which does not satisfy GR-$m$ at all, has only one player. This implies that it cannot attract the participation of players. 
For the case of DPoS systems, they (except for Steem) have $|\mathcal{A}|$ similar to $N_\tts{dpos}.$ 
The reason for $|\mathcal{A}|$ in Steem being relatively large when compared to $N_\tts{dpos}=20$ is that one block generator is randomly chosen among all nodes as described in Section~\ref{subsubsec:dpos}. 
However, in all DPoS systems, $|\mathcal{A}^{\frac{1}{2}}|$ and $|\mathcal{A}^{\frac{1}{3}}|$ are close to $\frac{N_\tts{dpos}}{2}$ and $\frac{N_\tts{dpos}}{3}$, respectively. 
This indicates that only a small number of players have been block generators even though block generators are frequently elected, implying that the barriers to becoming a block generator are quite high. 

\begin{small}
\begin{table}[t]
\renewcommand{\tabcolsep}{0.5pt}
  \centering
  \caption{Resource Power in DPoS Coins}
  \label{tab:dpos_power}
  \resizebox{\columnwidth}{!}{%
  \begin{tabular}{|c|c|c|c|c|c|c|c|c|c|c|c|c|c|}
    \hline
    \multicolumn{1}{|c|}{}&\multicolumn{3}{c|}{Delegates}&\multicolumn{3}{c|}{100 \%}&\multicolumn{3}{c|}{50\%}&\multicolumn{3}{c|}{33\%}\\
    \cline{2-13}
    Coin name&$|\mathcal{N}^{\tt D}|$&Gini$^{\tt D}$&H$^{\tt D}$&$|\mathcal{N}|$&Gini&H& $|\mathcal{N}^{\frac{1}{2}}|$&Gini$^{\frac{1}{2}}$ &H$^{\frac{1}{2}}$& $|\mathcal{N}^{\frac{1}{3}}|$&Gini$^{\frac{1}{3}}$ & H$^{\frac{1}{3}}$ \\
    \hline
    \hline
    \cline{1-13} EOS & 21 & 0.048 & 4.39 & 439 & 0.846 & 6.47 & 28 & 0.063 & 4.80 & 18 & 0.047 & 4.16 \\ 
     \hline
     \cline{1-13} TRON & 27& 0.198 & 4.54 & 165 & 0.849 & 4.84 & 12 & 0.258 & 3.29 & 6 & 0.324 & 2.23 \\ 
    \hline 
    \cline{1-13} Lisk & 101& 0.031 & 6.65 & 1179 & 0.907 & 6.99 & 52 & 0.013 & 5.70 & 35 & 0.011 & 5.13 \\ 
    \hline 
    \cline{1-13} BitShare & 27 & 0.070 & 4.74 & 140 & 0.550 & 6.35 & 21 & 0.051 & 4.34 & 14 & 0.038 & 3.80 \\ 
    \hline 
    \cline{1-13} Steem & 20 & 0.052 & 4.32 & 150 & 0.588 & 6.37 & 23 & 0.061 & 4.52 & 15 & 0.042 & 3.90 \\ 
    \hline
    \cline{1-13} GXChain & 21& 0.000 & 4.39 & $-$ & $-$ & $-$ & $-$ & $-$ & $-$ & $-$ & $-$ & $-$ \\ 
    \hline
    \cline{1-13} Ark & 51 & 0.053 & 5.66 & 196 & 0.734 & 5.86 & 26 & 0.054 &4.69 & 17 & 0.055 & 4.08 \\ 
    \hline
    \cline{1-13} WaykiChain & $-$ & $-$ & $-$ & $-$ & $-$ & $-$ & $-$ & $-$ & $-$ & $-$ & $-$ & $-$ \\
    \hline
    \cline{1-13} Achain & $-$ & $-$ & $-$ & $-$ & $-$ & $-$ & $-$ & $-$ & $-$ & $-$ & $-$ & $-$ \\ 
    \hline
    \cline{1-13} Asch & 91 & 0.041 & 6.49 & 633 & 0.745 & 7.63 & 71 & 0.028 & 6.15 & 46 & 0.032 & 5.52 \\ 
    \hline
  \end{tabular}}
\end{table}
\end{small}

Next, we describe the power distribution among nodes. 
As shown in Tables~\ref{tab:pow_level} and \ref{tab:pos_level}, PoW and PoS coins certainly have a high value of the Gini coefficient, which implies that they have a significantly biased power distribution. 
Meanwhile, DPoS coins, except for Steem, have a low Gini coefficient, and all DPoS coins have low values of Gini$^{\frac{1}{2}}$ and Gini$^{\frac{1}{3}}.$ 
This is because the elected block generators have the same opportunity to generate blocks in DPoS systems. 
Again, note that in Steem, one block generator is randomly chosen among all nodes, which makes the Gini coefficient for \textit{all} block generators in Steem high. 

Unlike Table~\ref{tab:pow_level} and \ref{tab:pos_level}, Table~\ref{tab:dpos_level} does not present the resource power of the nodes, where the resource power indicates the number of stakes delegated to each node, because the probability of generating blocks is not proportional to the resource power in DPoS systems. 
Thus, to present the distribution of resource power among nodes, we analyze the instantaneous number of stakes delegated to each node through block explorers. 
Table~\ref{tab:dpos_power} represents the distribution of stakes used to vote for nodes as of Nov. 19, 2018, where we mark with ``$-$" the values that cannot be determined in the block explorer for the corresponding coin. 
In particular, the voting process in WaykiChain has not yet been implemented at the time of writing~\cite{waykichain_plan}. 

In Table~\ref{tab:dpos_power}, $|\mathcal{N}^{\tt x}|,$ Gini$^{\tt x}$, and $H^{\tt x}$ represent the size of $\mathcal{N}^{\tt x}$, Gini coefficient, and entropy for $\mathcal{N}^{\tt x},$ respectively. 
The columns labeled \textit{Delegates}, \textit{100\%}, \textit{50\%}, and \textit{33\%} provide information regarding the number of nodes, the Gini coefficient, and the entropy for the delegates ($\mathcal{N}^D$), and for the nodes whose total resource power is 100\% ($\mathcal{N}$),  50\% ($\mathcal{N}^\frac{1}{2}$), and 33\% ($\mathcal{N}^\frac{1}{3}$), respectively. 
Gini$^{\tt D}$ is low for all DPoS systems, indicating that delegates possess similar resource power. 
In Section~\ref{subsubsec:dpos}, we explained that DPoS systems can converge in probability to the state where delegates have similar resource power. 
Here, the reason Gini$^{\tt D}$ of TRON is relatively high compared to the others is that the node~\cite{tron_node} operated by the TRON foundation is ranked in the first place by a relatively large margin. 
However, we observe that delegates, except for this node, possess almost the same resource power in TRON. 
Conversely, the value of Gini for all nodes is high, implying a large gap between the rich and the poor players. 
Moreover, Table~\ref{tab:dpos_power} shows that the resource power is significantly biased toward the delegates. 

\textit{As a result, the quantitative data analysis validates our theory and the analysis result of the incentive systems in Section~\ref{app:protocol}. }

\subsubsection{Multiple nodes run by the same player}
In DPoS systems that do not have weak identity management, 
a rich player can easily earn a higher profit by running multiple nodes. 
However, because they do not have any real identity management, it can be difficult to detect this rational behavior in practice.
Nevertheless, we show that one player runs multiple nodes in several coins: GXChain, Ark, and Asch.

\para{GXChain. }
GXChain has 21 delegates in the consensus protocol. 
We can see the activities of the delegates via the official GXChain block explorer~\cite{gxchain_explorer}, including their creator.
At the time of writing, we observed that two players with \tts{nathan} and \tts{opengate} accounts run 16 and 5 active delegates, respectively. 
More specifically, the \tts{nathan} account created the delegates \tts{aaron, caitlin, kairos, sakura, taffy,} and \tts{miner1$\sim$11,} and the \tts{opengate} account created the delegates \tts{hrrs, dennis1, david12, marks-lee,} and \tts{robin-red.}
This implies that the system is currently controlled by at most only two players. 

\para{Ark. }
We discover that two nodes, \tts{biz\_classic} and \tts{biz\_private}, are run by the same player. 
Firstly, we can see that a player who has address \tts{AHsuUuhTNCGCbnPNkwJbeH27E4sDdcnmgp} votes for \tts{biz\_classic}, and the delegate \tts{biz\_classic} share rewards with the voter by issuing transactions.
Because transactions issued in the Ark system include some messages, we were able to observe the following two messages sent from \tts{biz\_classic} to the voter~\cite{ark_mult1, ark_mult2}: 
\begin{enumerate}
    \item You meet the minimum for \tts{biz\_private}. Switch for higher payouts.
    \item FYI: Change your vote to \tts{biz\_private} for higher payouts :)
\end{enumerate}
Therefore, we can speculate that \tts{biz\_classic} and \tts{biz\_private} are owned by the same player. 

\para{Asch. }
There are 87 active delegates, and we were able to find 30 and 50 delegates with names such as \tts{asch\_team\_i} and \tts{at i}, respectively, where \tts{i} is replaced by a number. 
For example, there exist delegate nodes with the names \tts{asch\_team\_1} or \tts{at5}.
Even though these names are quite similar, this is not enough to suspect that these nodes are controlled by the same player.
To determine whether the 80 nodes are owned by one player, we must investigate their activities.

Firstly, we determine when they became delegates.
Based on the transaction history, we can observe that the nodes named \tts{asch\_team\_1}$\sim${\tts 5} have simultaneously participated in the consensus protocol as delegates since Sep. 11, 2017.
Moreover, nodes named \tts{asch\_team\_6}$\sim$\tts{15} and those named \tts{asch\_team\_16}$\sim$\tts{35} simultaneously became delegates on Apr. 11, 2018, and Jun. 11, 2018, respectively.
Among these nodes, \tts{asch\_team\_31}$\sim$\tts{35} were inactive at the time of writing (Oct. 2018). 
In addition, all 50 nodes named \tts{at i} have become delegators simultaneously since Jul. 6, 2018. 

Secondly, all these nodes received 100 XAS (i.e., a unit of the Asch coin) from an address just before they became delegates. 
Even the address, which sent 100 XAS to \tts{asch\_team\_1$\sim$5}, is the same, and addresses for \tts{asch\_team\_6$\sim$15} and \tts{asch\_team\_16$\sim$34} are also the same, respectively. 
Furthermore, \tts{asch\_team\_35} and all nodes named \tts{at i} received 100 XAS from the same address.
Finally, these 80 nodes sent currencies to the address \tts{GADQ2bozmxjBfYHDQx3uwtpwXmdhafUdkN} at almost the same time on Aug. 20, 2018.
From this evidence, we can speculate that the 80 delegate nodes are run by the same player (or organization). 

\smallskip\noindent
\textbf{Summary. }
From these systems, we were able to observe that one player runs multiple nodes for a higher profit. In particular, GXChain and Asch systems seem to be controlled by only two players and one player, respectively, implying a severely low level of decentralization. 
In summary, even though DPoS systems can achieve an even power distribution among nodes, this even power distribution does not translate to the players, which implies that the system has a lower level of decentralization than expected. 

\subsubsection{DAG}
\label{subsubsec:DAG}

In this section, we describe the analysis result of IOTA and Nano, which adopt DAG.
In IOTA, transaction issuers are required to validate their transactions by themselves, and currently, there are not enough issuers to run IOTA stably. 
Therefore, to solve this problem, the IOTA foundation controls the system as a central authority, which implies that IOTA has only one player~\cite{iota-milestone, iota-centralization}. 
This result is in agreement with our protocol analysis in respect that many players do not exist in IOTA.
Meanwhile, at the time of writing, even though Nano does not have enough players, there are \textit{relatively} many players when compared to IOTA. 
Specifically, there are 64 players in Nano, and two players possess approximately 45\% of the power, indicating a significantly biased power distribution. 
This fact is derived referring to the data obtained from a node monitoring website~\cite{nanode}. 
We see that the situation of Nano is owing to incentives outside of the blockchain system. 
Indeed, we observe that at least 37 players get incentives outside of the blockchain system by participating in the system, and these players possess approximately 80\% of the power\footnote{We were not able to identify all such players because there are untraceable players.}.
For example, BrainBlocks~\cite{brainblocks}, which provides a platform related to Nano, is incentivized to run nodes in the Nano system for its business, and currently, it is a rich player in the Nano system. 
As a result, in Nano, most players participate in the consensus protocol to receive external incentives, and they possess most of the resource power.  
External incentives are discussed further in Section~\ref{subsec:debate}. 
\section{Discussion}
\label{sec:discuss}

\subsection{Debate on Incentive Systems}
\label{subsec:debate}

Recently, there was an interesting debate on the incentive system of Algorand~\cite{debate, gilad2017algorand,need}. 
Micali said that incentives are the hardest thing to do, and that existing incentivization has led to poor decentralization. 
Our study supports this notion by proving that it is impossible to design incentive systems for permissionless blockchains such that good decentralization is achieved. 

Can we then create a permissionless blockchain to achieve good decentralization without any incentive system?
The case where the incentive system does not exist is represented by $U_{n_i}=-c,$ where $c$ is the cost associated with running one node. 
This satisfies the second requirement of Def.~\ref{def:perfect} because NS-$\delta$ and ED-$(\varepsilon, \delta)$ are met as a trivial case. 
Meanwhile, the first two conditions, GR-$m$ and ND-$m$, cannot be satisfied. 
As examples, we can consider BridgeCoin, IOTA, and Byteball, which do not have incentive systems and have difficulty in attracting the participation of many players. 
BridgeCoin has only one player (refer to Tab.~\ref{tab:pow_level}), 
and IOTA is also controlled by just one player, the IOTA foundation~\cite{iota-milestone, iota-centralization}. 
Byteball is another system that adopts DAG, and there are only four players. 
These examples show that blockchain systems with no incentive system cannot have a sufficient number of players. 

However, our study considered only the incentives inside the system, and not incentives outside the system. 
Therefore, if there are some incentives that players can obtain outside the blockchain system, they can participate in the system.
For example, IBM is a validator in Stellar, which does business using Stellar, and BrainBlocks~\cite{brainblocks}  provides a payment platform related to Nano. 
This incentivizes IBM and BrainBlocks to participate in each system. 
Note that that fact does not ensure that these systems reach good decentralization. 
Indeed, both of these systems have poor decentralization~\cite{stellarbeat.io, mystellar, minjeong}. 
In other words, they do not have a sufficient number of players and have a biased power distribution. 
Besides, through these cases, we can empirically see that organizations related to the coin system (e.g., the coin foundation or companies that do business with the coin) control the blockchain system, which may deviate from the philosophy of permissionless blockchains. 

Note that we do not assert that blockchains without an incentive mechanism would always suffer from poor decentralization. 
Indeed, we can also find other peer-to-peer systems such as Tor and BitTorrent that attract many players without an incentive system. 
Of course, these systems are significantly different from a blockchain because they do not require resources such as computational power and stakes unlike a blockchain. 
In this paper, we do remain neutral on this debate.

\vspace{-2mm}
\subsection{Relaxation of Conditions from Consensus Protocol}
\label{subsec:relax}
We proved that an incentive system in permissionless blockchains cannot simultaneously satisfy the four conditions. 
Nevertheless, if there is a consensus protocol that relaxes part of the four conditions, we can expect to be able to design an incentive system such that good decentralization is achieved.
However, it seems to be quite difficult to design such consensus protocols.
We explain below the reason why the design of a consensus protocol relaxing the conditions is difficult by considering two methods of designing such protocols: 1) designing non-outsourceable puzzles and 2) finding non-delegable or non-divisible resources.

\para{Non-outsourceable puzzles. }
There exist several studies on the construction of non-outsourceable puzzles in PoW systems~\cite{two-phase, miller2014permacoin,miller2015nonoutsourceable, zeng2017nonoutsourceable}. 
In those puzzles, if players outsource the puzzles, their rewards can be stolen.
This risk can, therefore, cause a pool manager to refrain from outsourcing his work to pool miners. 
For example, in the proposed schemes, if a pool manager outsources the puzzles, the pool miner who finds a valid block might not submit that valid block to the pool manager and might steal the block reward. 

However, these puzzles still allow other types of mining pools, such as \textit{cloud mining}~\cite{cloud}, where individual miners buy hash rate from their service provider, and the provider directly solves the PoW puzzles using computing resources gathered by spending the received funds.
Miller et al.~\cite{miller2015nonoutsourceable} claimed that they can prevent cloud mining as well, since the cloud service provider can steal block rewards in their protocol.
However, with or without non-outsourceable puzzles, the provider can always steal the block reward without any clear evidence. 
Despite this risk, cloud mining has settled as one of the popular types of mining~\cite{popular-method} since cloud miners can reduce the cost of running hardware and nodes. 
Indeed, there exist several popular cloud mining services~\cite{recommend} such as Genesis Mining~\cite{genesis}, HashNest~\cite{hashnest} operated by BITMAIN~\cite{bitmain}, and Bitcoin.com~\cite{bitcoin.com}. 
This indicates that, if profitable, the delegation of resource power to part of the players would still occur even in non-outsourceable PoW protocols~\cite{two-phase, miller2014permacoin, miller2015nonoutsourceable, zeng2017nonoutsourceable}. 
Moreover, the more trust that the company providing the cloud mining service gets from users, the more popular the cloud service would become.

Even in the case of PoS coins, we can empirically see that players would delegate their resources to others for higher profits. 
One way is to delegate resources through investment in service providers, similar to cloud mining in PoW systems, and it seems to be difficult to prevent this if such a business is profitable.
As a result, it would be difficult to make the delegating behavior disappear by simply modifying the consensus protocol.

\para{Non-delegable/non-divisible resources. }
Another way to relax the four conditions is to find non-delegable or non-divisible resources. 
These resources make it impossible for players to delegate their resources to others and to run multiple nodes, respectively. 
Therefore, for each resource, it would be sufficient for the incentive systems to satisfy all conditions except for ND-$m$ and NS-$\delta$ in order to achieve full decentralization.

We can consider reputation as one such resource. 
Currently, GoChain uses proof-of-reputation (PoR) as a consensus algorithm in which nodes must have a high reputation score to participate.
In this system, only the company can be a validator, and it believes that PoR can achieve almost full decentralization~\cite{gochain, gochain2}.
In addition, trust can be one of the non-delegable and non-divisible resources. 
In the Stellar system, nodes have a trust-based relationship with one another. 
Specifically, Stellar uses FBA as a consensus algorithm, where nodes configure their quorum slice, which is a set of dependable nodes during a consensus process, according to their trust relationship. 
In addition, Bahri et al. proposed proof-of-trust (PoT), where more trusted nodes can easily solve puzzles~\cite{bahri2018trust}. 
However, both reputation and trust are not suitable for permissionless blockchains because players would need to know real identities of others. 
Even though Stellar is classified as a permissionless blockchain, for nodes to be effective validators, they should reveal identities.
As a result, it remains an open question as to whether we can find non-delegable or non-divisible resources that are suitable for permissionless blockchains. 
\section{Related Work}

\noindent
\textbf{Attacks. }
Eyal et al.~\cite{eyal2014majority} proposed selfish mining, which an attacker possessing over 33\% of the computing power can execute in PoW-based systems. They mentioned that this attack causes rational miners to join the attacker, eventually decreasing the level of decentralization.
Eyal~\cite{eyal2015miner} and Kwon et al.~\cite{kwon2017selfish} modeled a game between two pools. 
When considering block withholding attacks, the game is equivalent to \textit{the prisoner's dilemma}, and the attacks cause rational miners to leave their mining pools, and instead, directly run nodes in a consensus protocol~\cite{eyal2015miner}. 
Contrary to this positive result, a fork after withholding attack between two pools leads to a pool-size game, where a larger pool can earn extra profits, and thus, the Bitcoin system can become more centralized. 
Furthermore, two existing works analyzed the Bitcoin system in a transaction-fee regime where transaction fees in block rewards are not negligible~\cite{carlsten2016instability, tsabary2018gap}. 
They described that this regime incentivizes large miner coalitions and make a system more centralized. 

\smallskip
\noindent
\textbf{Analysis. }
Many papers have already examined centralization in the Bitcoin system. 
First, Gervais et al. described centralization of the Bitcoin system in terms of various aspects such as services, mining, and incident resolution processes~\cite{gervais2014bitcoin}. 
Miller et al. observed a topology in the Bitcoin network and found that approximately 2\% of high-degree nodes acquire three quarters of the mining power~\cite{miller2015discovering}.
Moreover, Feld et al. analyzed the Bitcoin network, focusing on its autonomous systems (ASes), and showed that routable peers are concentrated only in a few ASes~\cite{feld2014analyzing}. 
Recently, Gencer et al. analyzed the Bitcoin and Ethereum systems from the perspective of decentralization~\cite{gencer2018decentralization}.
Kwon et al. analyzed a game in which two PoW coins with compatible mining algorithms exist~\cite{kwon2019bitcoin}. 
They showed that fickle mining behavior between two coins can reduce the decentralization level of the lower-valued one of the two coins. 
In addition, Kim et al. analyzed the Stellar system and concluded that the system is significantly centralized~\cite{minjeong}. 

\smallskip
\noindent
\textbf{Solutions. }
There are several works that address the issue of poor decentralization in blockchains. 
Many works~\cite{two-phase, miller2014permacoin,miller2015nonoutsourceable, zeng2017nonoutsourceable} have proposed non-outsourceable puzzles to prevent mining pools from being popular.
However, they cannot fully prevent the delegation (Section~\ref{subsec:relax}).
As another solution, Luu et al. proposed an efficient decentralized mining pool, SMARTPOOL, where individual miners who directly run nodes in the consensus protocol can consistently earn profits~\cite{luu2017smart}.
However, this still does not incentivize players to run nodes directly (see Section~\ref{app:protocol}). 
Another work~\cite{blocki2016designing} proposed a proof-of-human-work requiring labor from players with CAPTCHA as a human-work puzzle. 
As mentioned by~\cite{blocki2016designing}, although the gap among labor abilities of people is relatively small by nature, rich players can hire more workers to solve more puzzles. 
Lastly, we are aware of a recent paper~\cite{brunjes2018reward} in which the authors addressed a similar problem to our paper.
Brünjes et al. proposed a reward scheme, which causes a system to reach a state where $k$ staking pools with similar resource power exist.
They assumed our third condition, NS-$\delta$ (i.e., all players can run only one node), and thus, it seems difficult for their incentive system to achieve good decentralization in practice. 
As described in previous sections, there is an incentive system that satisfies only GR-$m$, ND-$m$, and ED-$(\varepsilon, \delta)$. 

\section{Conclusion and direction}
\label{sec:conclude}

Developers are facing difficulties in designing blockchain systems to achieve good decentralization. 
Our study answers the question of why it is significantly difficult to design a system that achieves good decentralization, by proving that the achievement of good decentralization in the consensus protocol and non-reliance on a TTP contradict each other. 
More specifically, we prove that when the ratio between the resource power of the poorest and richest players is close to 0, the upper bound of the probability that systems without a Sybil cost will achieve full decentralization is close to 0. 
This result indicates that if we cannot narrow the gap between the rich and the poor in the real world or assign a Sybil cost without relying on a TTP, a high level of decentralization in systems will not occur forever with a high probability. 
Furthermore, through the protocol and data analysis, we observed the phenomena consistent with our theory.
From our result, we propose one direction to achieve good decentralization of the system; developing a method that can assign Sybil costs without relying on a TTP in blockchains.

\bibliographystyle{ACM-Reference-Format}
\bibliography{references}

\section*{Appendix}
\section{Proof of Theorem~4.3}
\label{app:example}

Because the function $U_{n_i}(\alpha_{n_i}, \bm{\bar{\alpha}_{-n_i}})$ is a strictly increasing function of $\alpha_{n_i}$, the players would want to increase their resource power and increase it at rate $r$ per earned profit. 
Therefore, the resource power $\alpha_{n_i}^{t}$ of node $n_i$ at time $t$ increases to 
$\alpha_{n_i}^{t+1}=\alpha_{n_i}^{t}+r\cdot R_{n_i}^{t}$ at time $t+1.$

Then we sequence nodes at time $t$ such that $\alpha_{n_i}^t \leq\alpha_{n_j}^t$ if $i<j$. 
Thus, $\alpha_{n_1}^t$ and $\alpha_{n_M}^t$ represent the smallest and largest resource power at time $t$, respectively.
In addition, we assume that there exist $M$ nodes (i.e., $|\mathcal{N}|=M$).
At time $t+1,$ the node $n_i$'s resource power $\alpha_{n_i}^{t+1}$ and other node $n_j$'s power $\alpha_{n_j}^{t+1}$ would be $\alpha_{n_i}^t+r\cdot f(\bm{\bar{\alpha}^t})$ and $\alpha_{n_j}^t,$ respectively, if node $n_i$ generates a block with probability $\Pr(R_{n_i}^t=f(\bm{\bar{\alpha}^t})\,|\,\bm{\bar{\alpha}^t})$.
Then, we resequence $M$ nodes at time $t+1$ such that $\alpha_{n_i}^{t+1} \leq\alpha_{n_j}^{t+1}$ if $i<j$.
Here, for simplicity, we denote by $\beta_{n_i}$ (or $\beta^t_{n_i}$) a resource power fraction of node $n_i$ (at time $t$). 
In other words, $\beta_{n_i}=\frac{\alpha_{n_i}}{\sum_{n_i}\alpha_{n_i}}$ and $\beta^t_{n_i}=\frac{\alpha^t_{n_i}}{\sum_{n_j}\alpha^t_{n_j}}.$
Moreover, $\frac{f(\bm{\bar{\alpha}^t})}{\sum_{n_i}\alpha^t_{n_i}}$ is denote by $B$.

Now, we show that $\lim_{t\rightarrow\infty}E[\beta_{n_1}^t]=\lim_{t\rightarrow\infty}E[\beta_{n_M}^t].$
First, the following is met. 
\begin{align}
&\frac{\beta_{n_i}}{\beta_{n_M}}\leq\frac{U_{n_i}(\alpha_{n_i}, \bm{\bar{\alpha}_{-n_i}})}{U_{n_M}(\alpha_{n_M},\bm{\bar{\alpha}_{-n_M}})} \Rightarrow \frac{1}{\beta_{n_M}}\leq\frac{\sum_{i}U_{n_i}(\alpha_{n_i},\bm{\bar{\alpha}_{-n_i}})}{U_{n_M}(\alpha_{n_M},\bm{\bar{\alpha}_{-n_M}})}\notag\\
&\Leftrightarrow U_{n_M}(\alpha_{n_M},\bm{\bar{\alpha}_{-n_M}})\leq \beta_{n_M}\sum_{i}U_{n_i}(\alpha_{n_i},\bm{\bar{\alpha}_{-n_i}}),\label{eq:M}\\    
&\frac{\beta_{n_i}}{\beta_{n_1}}\geq\frac{U_{n_i}(\alpha_{n_i},\bm{\bar{\alpha}_{-n_i}})}{U_{n_1}(\alpha_{n_1},\bm{\bar{\alpha}_{-n_1}})}\Rightarrow \frac{1}{\beta_{n_1}}\geq\frac{\sum_{i}U_{n_i}(\alpha_{n_i},\bm{\bar{\alpha}_{-n_i}})}{U_{n_1}(\alpha_{n_1},\bm{\bar{\alpha}_{-n_1}})}\notag\\
&\Leftrightarrow U_{n_1}(\alpha_{n_1},\bm{\bar{\alpha}_{-n_1}})\geq \beta_{n_1}\sum_{i}U_{n_i}(\alpha_{n_i},\bm{\bar{\alpha}_{-n_i}}).\label{eq:1}
\end{align}
In Eqs.~\eqref{eq:M} and \eqref{eq:1}, the equal sign is true only if all nodes have the same resource power fraction $\frac{1}{M}$.
Then we can derive the below equations.

\begin{equation*}
\begin{aligned}
&E[\beta_{n_i}^{t+1}|\bm{\bar{\alpha}^t}]=\Pr(R_{n_i}^t=f(\bm{\bar{\alpha}^t})\,|\,\bm{\bar{\alpha}^t})\Bigl(\frac{r\cdot B}{1+r\cdot B}\Bigr)+\\
&\sum_{j}\frac{\beta_{n_i}^t \Pr(R_{n_j}^t=f(\bm{\bar{\alpha}^t})|\bm{\bar{\alpha}^t})}{1+r\cdot B}\leq\frac{rU_{n_i}(\alpha_{n_i}^t,\bm{\bar{\alpha}^t_{-n_i}})}{1+r\cdot B}+\\
&\sum_{j}\frac{\beta_{n_M}^t \Pr(R_{n_j}^t=f(\bm{\bar{\alpha}^t})|\bm{\bar{\alpha}^t})}{1+r\cdot B}\\
&\leq\frac{r\beta_{n_M}^t\sum_{j}U_{n_j}(\alpha_{n_j}^t,\bm{\bar{\alpha}_{-n_j}^t})}{1+r\cdot B}+
\frac{\beta_{n_M}^t}{1+r\cdot B}=\beta_{n_M}^t 
\end{aligned}
\end{equation*}
Similarly, we also prove the following equation.
\begin{align}
&E[\beta_{n_i}^{t+1}|\bm{\bar{\alpha}^t}]=\Pr(R_{n_i}^t=f(\bm{\bar{\alpha}^t})|\bm{\bar{\alpha}^t})\Bigl(\frac{r\cdot B}{1+r\cdot B}\Bigr)+\\
&\sum_{j}\frac{\beta_{n_i}^t \Pr(R_{n_j}^t=f(\bm{\bar{\alpha}^t})|\bm{\bar{\alpha}^t})}{1+r\cdot B}\geq\frac{rU_{n_1}(\alpha_{n_1}^t,\bm{\bar{\alpha}_{-n_1}^t})}{1+r\cdot B}+\\
&\sum_{j}\frac{\beta_{n_1}^t\Pr(R_{n_j}^t=f(\bm{\bar{\alpha}^t})|\bm{\bar{\alpha}^t})}{1+r\cdot B}\\
&\geq\frac{r\beta_{n_1}^t\sum_{j}U_{n_j}(\alpha_{n_j}^t,\bm{\bar{\alpha}_{-n_j}^t})}{1+r\cdot B}+
\frac{\beta_{n_1}^t}{1+r\cdot B}=\beta_{n_1}^t \notag
\end{align}
Therefore, the following is satisfied:
$$\beta_{n_1}^t\leq E[\beta_{n_i}^{t+1}|\bm{\bar{\alpha}^t}]\leq \beta_{n_M}^t,\vspace{-1mm}$$ where two equal signs are true if all nodes have the same power fraction.
Because $E[\beta_{n_i}^{t+1}]=E[E[\beta_{n_i}^{t+1}|\bm{\bar{\alpha}^t}]],$ 
the below equation is satisfied: 
$$E[\beta_{n_1}^t]\leq E[\beta_{n_i}^{t+1}]\leq E[\beta_{n_M}^t].$$ 

By the above equation, $E[\beta_{n_1}^t]$ and $E[\beta_{n_M}^t]$ are increasing and decreasing functions of $t$, respectively, and converge according to the \textit{monotone convergence theorem}.
Moreover, if we assume that $\lim_{t\rightarrow\infty}E[\beta_{n_1}^t]=x<\lim_{t\rightarrow\infty}E[\beta_{n_M}^t]=y,$
$E[\beta_{n_1}^{t+1}|\beta_{n_1}^t=x]$ is greater than $x$ for any $t\geq 0$, and this is a contradiction because $E[\beta_{n_1}^{t+1}|\beta_{n_1}^t=x]$ should be $x$ for a large value of $t$. 
Thus, $x$ cannot be the limit, and $\lim_{t\rightarrow\infty}E[\beta_{n_1}^t]=\lim_{t\rightarrow\infty}E[\beta_{n_M}^t].$ 
In addition, because $\beta_{n_M}^t$ is always not less than $\beta_{n_1}^t$,
$$\lim_{t\rightarrow\infty}E[\beta_{n_1}^t]=\lim_{t\rightarrow\infty}E[\beta_{n_M}^t] \Leftrightarrow \lim_{t\rightarrow\infty}E[|\beta_{n_M}^t-\beta_{n_1}^t|]=0.$$
This fact implies that $\beta_{n_i}^t$ converges in mean to $\frac{1}{M}.$ 
Because convergence in mean implies convergence in probability, 
$$\lim_{t\rightarrow\infty}\Pr\Bigl[\frac{\beta_{n_M}^t}{\beta_{n_1}^t}=1\Bigr]=1.$$ 
As a result, Condition~\ref{con:density} is satisfied. 

On the contrary, if 
$$\frac{U_{n_i}(\alpha_{n_i}, \bm{\bar{\alpha}_{-n_i}})}{\alpha_{n_i}}>\frac{U_{n_j}(\alpha_{n_j}, \bm{\bar{\alpha}_{-n_j}})}{\alpha_{n_j}}$$ for any $\alpha_{n_i}>\alpha_{n_j},$ the following is met: $E[\beta_{n_M}^t]\leq E[\beta_{n_M}^{t+1}].$ 
As a result, $\lim_{t\rightarrow\infty}E[\beta_{n_M}^{t+1}]=1,$
and $\beta_{n_M}^{t+1}$ converges in probability to 1, where the case indicates extreme centralization.  
Lastly, when 
$$\frac{U_{n_i}(\alpha_{n_i}, \bm{\bar{\alpha}_{-n_i}})}{\alpha_{n_i}}=\frac{U_{n_j}(\alpha_{n_j}, \bm{\bar{\alpha}_{-n_j}})}{\alpha_{n_j}}$$ for any $\alpha_{n_i}>\alpha_{n_j},$ the following is satisfied:
$E[\beta_{n_i}^{t+1}]= E[\beta_{n_i}^t]=\beta_{n_i}^0.$
Therefore, if $\beta_{n_i}^t$ converges in mean to a value, the value would be $\beta_{n_i}^0.$ 
However, the fact that $\lim_{t\rightarrow\infty}E[\beta_{n_i}^t]=\beta_{n_i}^0$ does not imply 
$\lim_{t\rightarrow\infty} E[|\beta_{n_i}^t-\beta_{n_i}^0|]=0,$ and indeed the following would be met: $\lim_{t\rightarrow\infty} E[|\beta_{n_i}^t-\beta_{n_i}^0|]>0$.
As a result, $\beta_{n_i}^t$ does not converge in probability to $\beta_{n_i}^0$, which implies that there is no convergence in probability of $\beta_{n_i}^t$. 
These facts can be proven, similar to the above proof.

\section{Proof of Theorem~5.1}
\label{app:suff_nec}

In this section, we prove Theorem~\ref{thm:suff_nec}, and we introduce notations $\bm{\bar{EP}}=(EP_{p_i})_{p_i\in\mathcal{P}}$ and $\bm{\bar{EP}^t}=(EP_{p_i}^t)_{p_i\in\mathcal{P}^t}.$ 
In addition, we assume that there is a mechanism $\mathcal{M}$, which stochastically makes a system $(m,\varepsilon, \delta)$-decentralized. 
This mechanism $\mathcal{M}$ can be represented with two functions $\mathcal{M}_1^t$ and $\mathcal{M}_2^t$, which output the effective power distribution among players and resource power distribution among nodes after $t$ time from when entering $\mathcal{M}$, respectively. 
Formally, the two functions are presented as $\mathcal{M}^t_1: \Omega_{EP}\times \Omega_{\alpha} \mapsto \Omega_{EP}$ and $\mathcal{M}^t_2: \Omega_{EP}\times \Omega_{\alpha} \mapsto \Omega_{\alpha}$,
where 
$$\Omega_{EP}=\{(EP_{p_i})_{p_i\in\mathcal{P}}\,|\, EP_{p_i}\in\mathbb{R}^+\} \text{ and }$$
$$\Omega_{\alpha}=\{(\alpha_{n_i})_{n_i\in\mathcal{N}}\,|\, \alpha_{n_i}\in\mathbb{R}^+\}.$$
We also define $\Omega_{\alpha}(\bm{\bar{EP}})$ as follows:
$$\Omega_{\alpha}(\bm{\bar{EP}})=\left\{(\alpha_{n_i})_{n_i\in\mathcal{N}}\,\Big|\, \alpha_{n_i}\in\mathbb{R}^+, \sum_{n_i\in\mathcal{N}_{p_i}}\alpha_{n_i}=EP_{p_i}\right\}.$$
Moreover, note that, because a system has zero Sybil cost (i.e., $C=0$), the following equation is met: 
\begin{equation}
    U_{p_i}(\bm{\bar{EP}}, \bm{\bar{\alpha}})=\sum_{n_j\in\mathcal{N}_{p_i}^{0}}U_{n_i}(\bm{\bar{EP}^\prime}, \bm{\bar{\alpha}}) \quad\forall \bm{\bar{EP}}\not = \bm{\bar{EP}^\prime},
    \label{eq:mech}
\end{equation}
where $U_{p_i}$ indicates an utility of player $p_i$ and $\mathcal{N}_{p_i}^0$ indicates the set of nodes run by player $p_i$ at the state with the effective power distribution $\bm{\bar{EP}}$ and the resource power distribution $\bm{\bar{\alpha}}.$
In addition, we define $N(\bm{\bar{EP}^\star})$ as 
\begin{equation*}
\begin{aligned}
\bigcup_{\bm{\bar{EP}}\in\Omega_{EP}}\bigcap_{k=0}^{\infty}\Big\{\mathcal{M}_{c2}^t &\Big(\bm{\bar{EP}},f_{EP\rightarrow\alpha}(\bm{\bar{EP}})\Big)\,\Big|\,t>k, \\
&\mathcal{M}_{c1}^t\Big(\bm{\bar{EP}},f_{EP\rightarrow\alpha}(\bm{\bar{EP}})\Big)=\bm{\bar{EP}^\star}\Big\}, 
\end{aligned}
\end{equation*}
where the function $f_{EP\rightarrow\alpha}: \bm{\bar{EP}}\mapsto \bm{\bar{\alpha}}$ outputs the resource power distribution among nodes in which each player runs only one node (i.e., $f_{EP\rightarrow\alpha}(\bm{\bar{EP}})=(\alpha_{n_i})_{n_i\in\mathcal{N}}$ and $\alpha_{n_i}=EP_{p_i}$ for $\mathcal{N}_{p_i}=\{n_i\}$). 
Note that $f_{EP\rightarrow\alpha}(\bm{\bar{EP}})\in\Omega_{\alpha}(\bm{\bar{EP}}).$
In the definition of $N(\bm{\bar{EP}})$, $\mathcal{M}_{c1}^t(\bm{\bar{EP}},\bm{\bar{\alpha}})$ and $\mathcal{M}_{c2}^t(\bm{\bar{EP}},\bm{\bar{\alpha}})$ output an effective power distribution among players and a resource power distribution among nodes, respectively, and the outputs are the same as $\mathcal{M}_1^t(\bm{\bar{EP}},\bm{\bar{\alpha}})$ and $\mathcal{M}_2^t(\bm{\bar{EP}},\bm{\bar{\alpha}})$, respectively, under the assumption that \textit{a mechanism $\mathcal{M}$ does not change the resource power owned players} (note that the mechanism can change the effective power of players).

The set of all $(m,\varepsilon, \delta)$-decentralized distribution $\bm{\bar{EP}}$ is denoted by $\bm{S}.$ 
The probability to reach $(m,\varepsilon, \delta)$-decentralization is 
$$\lim_{t\rightarrow \infty}\Pr\left(\mathcal{M}_1^t(\bm{\bar{EP}^0},\bm{\bar{\alpha}}^0)\in \bm{S}\right).$$ 
Moreover, $I_{\bm{\bar{EP}_\delta}}$ denotes a parameter that shows whether the mechanism $\mathcal{M}$ can learn the information about $\bm{\bar{EP}_\delta}=(EP_a)_{a\geq \delta}$, where $I_{\bm{\bar{EP}_\delta}}=1$ (or 0) indicates that mechanism $\mathcal{M}$ gets (or does not get) the information about $\bm{\bar{EP}_\delta}$. 
In other words, when $I_{\bm{\bar{EP}_\delta}}=1$, a system can know the effective power distribution among players above the $\delta$-th percentile.

\begin{lemma}
$I_{\bm{\bar{EP}}_\delta}=1$ if and only if $N(\bm{\bar{EP}})\cap N(\bm{\bar{EP}^\prime})=\emptyset$ for any $\bm{\bar{EP}_\delta}\not = \bm{\bar{EP}^\prime_\delta}$, where $\bm{\bar{EP}_\delta}\subset\bm{\bar{EP}}$ and $\bm{\bar{EP}_\delta^\prime}\subset\bm{\bar{EP}^\prime}$.
\label{lem:ep_delta}
\end{lemma}

\begin{proof}
If $I_{\bm{\bar{EP}_\delta}}=1$, there is an incentive system such that, for any $\bm{\bar{EP}}$ and $\bm{\bar{EP}^\prime}$, which have $\bm{\bar{EP}_\delta}$ and $\bm{\bar{EP}^\prime_\delta}\,(\not=\bm{\bar{EP}_\delta})$, respectively,
$$U_{p_i}(\bm{\bar{EP}}, \bm{\bar{\alpha}})\not=\sum_{n_j\in\mathcal{N}_{p_i}^{0}}U_{n_j}(\bm{\bar{EP}^\prime}, \bm{\bar{\alpha}}) \quad \forall \bm{\bar{\alpha}}\in N(\bm{\bar{EP}})\cap N(\bm{\bar{EP}^\prime}).$$
However, the above equation contradicts Eq.~\eqref{eq:mech}, and thus, $N(\bm{\bar{EP}})\cap N(\bm{\bar{EP}^\prime})$ for $\bm{\bar{EP}_\delta}\not=\bm{\bar{EP}^\prime_\delta}$ should be the empty set. 
In addition, if $N(\bm{\bar{EP}})\cap N(\bm{\bar{EP}^\prime})=\emptyset$, a system can determine the
effective power distribution among players above the $\delta$-th percentile. 
Therefore, $I_{\bm{\bar{EP}_\delta}}=1$ if and only if $N(\bm{\bar{EP}})\cap N(\bm{\bar{EP}^\prime})=\emptyset$ for any $\bm{\bar{EP}_\delta}\not = \bm{\bar{EP}^\prime_\delta}$.
\end{proof}

\begin{lemma}
$N(\bm{\bar{EP}})\cap N(\bm{\bar{EP}^\prime})=\emptyset$ for any $\bm{\bar{EP}_\delta}\not = \bm{\bar{EP}^\prime_\delta}$ if and only if, for any effective power distribution $\bm{\bar{EP}^\star}$, $N(\bm{\bar{EP}^\star})=\emptyset$ or it is not more profitable for any player with effective power $EP_{p_i}^\star\geq EP^\star_{\delta}$ to run multiple nodes than to run only one node.
\label{lem:empty}
\end{lemma}

\begin{proof}
It is easy to prove $N(\bm{\bar{EP}})\cap N(\bm{\bar{EP}^\prime})=\emptyset$ for any $\bm{\bar{EP}_\delta}\not = \bm{\bar{EP}^\prime_\delta}$, when it is most profitable for players to collude or when a player with effective power $EP_{p_i}\geq EP_{\delta}$ runs one node. 
Therefore, we describe the proof of the other direction. 
To do this, we assume that a player with effective power greater than or equal to $EP^{\star}_\delta$ runs multiple nodes in the state with effective power distribution $\bm{\bar{EP}^\star}$ and so the state has the resource power distribution $\bm{\bar{\alpha}^\star}$ (i.e., $\bm{\bar{\alpha}^\star}\in N(\bm{\bar{EP}^\star})$).
Here, we define a function $f_{\alpha\rightarrow EP}: \bm{\bar{\alpha}}\mapsto \bm{\bar{EP}}$ as $f_{\alpha\rightarrow EP}(\bm{\bar{\alpha}})=(EP_{p_i})_{p_i\in\mathcal{P}}$, where the output represents a state in which each player runs only one node and $EP_{p_i}=\alpha_{n_i}$. 
Then $\bm{\bar{\alpha}^\star}$ belongs to the set $N(f_{\alpha\rightarrow EP}(\bm{\bar{\alpha}^\star}))$. 
This is certainly true when it is not more profitable for some players to delegate their resource to others or run more than one node in the state with $f_{\alpha\rightarrow EP}(\bm{\bar{\alpha}^\star}).$
Even if it is more profitable for some players to run more than one node in the state with $f_{\alpha\rightarrow EP}(\bm{\bar{\alpha}^\star}),$ the state can come back to itself after going through a process where a player runs multiple nodes and then delegates its resource power to others because 
$\bm{\bar{\alpha}^\star}\in N(\bm{\bar{EP}^\star}).$
Lastly, if it is more profitable for some players to delegate their resource power to others, the state can also come back to itself after a player delegates its resource power to others. 
As a result, $\bm{\bar{\alpha}^\star} \in N(f_{\alpha\rightarrow EP}(\bm{\bar{\alpha}^\star}))$ and $N(\bm{\bar{EP}^\star})\cap N(f_{\alpha\rightarrow EP}(\bm{\bar{\alpha}^\star}))\not =\emptyset$.
This fact implies that $N(\bm{\bar{EP}})\cap N(\bm{\bar{EP}^\prime})=\emptyset$ for any $\bm{\bar{EP}_\delta}\not = \bm{\bar{EP}^\prime_\delta}$ if and only if, for any $\bm{\bar{EP}}$, $N(\bm{\bar{EP}})=\emptyset$ or players above the $\delta$-th percentile should run only one node.
Note that, in order to satisfy $N(\bm{\bar{EP}})=\emptyset$, it should be more profitable for some players to delegate their resource to others in the state $\bm{\bar{EP}}.$ 
\end{proof}

In Lemma~\ref{lem:empty}, the fact that $N(\bm{\bar{EP}})$ is the empty set represents that a coalition for some players is more profitable at the state $\bm{\bar{EP}}.$

When a system can find out whether $\frac{EP_{\tt max}}{EP_{\delta}}\leq 1+\varepsilon$ for the current state and get ${EP}_{\tt max}$ if the ratio is greater than $1+\varepsilon,$
the probability to reach $(m,\varepsilon, \delta)$-decentralization would be certainly greater than that for when it is not. 
This is because if $\frac{{EP}_{\tt max}}{{EP}_{\delta}}$ is greater than $1+\varepsilon,$ the mechanism $\mathcal{M}$, which makes $\bm{\bar{EP}}$ belong to $\bm{S}$, should adjust $\frac{{EP}_{\tt max}}{{EP}_{\mu}}$\footnote{To get a fraction $\frac{{EP}_{\tt max}}{{EP}_\mu}$, the system should get ${EP}_{\tt max}$ and ${EP}_\mu.$} for some $\mu\geq \gamma$.
Also, if the system adjusts $\frac{{EP}_{\tt max}}{{EP}_{\mu}}$ while not knowing the value of $\frac{{EP}_{\tt max}}{{EP}_{\mu}}$, the state cannot move in the best direction to $(m,\varepsilon, \delta)$-decentralization. 
As a result, the following is met:
\begin{align}
&\max_{\mathcal{M}}\lim_{t\rightarrow \infty}\Pr(\mathcal{M}^t_1(\bm{\bar{EP}^0}, \bm{\bar{\alpha}^0})\in\bm{S}\,|\, I_{\bm{\mathcal{S}}}=0 \text{ or }\label{eq:max1}\\ 
&I^{\bm{\mathcal{S}^c}}_{\bm{\bar{EP}_{100}}}=0)\leq \max_{\mathcal{M}}\lim_{t\rightarrow\infty}\Pr(\mathcal{M}^t_1(\bm{\bar{EP}^0}, \bm{\bar{\alpha}^0})\in\bm{S}\,|\\
&\hspace{4.5cm} I_{\bm{\mathcal{S}}}=1, I^{\bm{\mathcal{S}^c}}_{\bm{\bar{EP}_{100}}}=1)=\notag\\
&\max_{\mathcal{M}}\lim_{t\rightarrow\infty}\Pr(\mathcal{M}^t_1(\bm{\bar{EP}^0}, \bm{\bar{\alpha}^0})\in\bm{S}\,|\,N(\bm{\mathcal{S}})\cap N(\bm{\mathcal{S}^c})=\emptyset \label{eq:max2}\\
& \text{ and } N(\bm{\bar{EP}})\cap N(\bm{\bar{EP}^\prime})=\emptyset \text{ for any } {EP}_{\tt max}\not = EP^\prime_{\tt max}),\notag
\end{align}
where $I_{\bm{\mathcal{S}}}=1$ (or 0) indicates that a system can (or cannot) learn the information about whether the current state is in $\bm{\mathcal{S}}$, and
$I^{\bm{\mathcal{S}^c}}_{\bm{\bar{EP}_{100}}}=1$ (or 0) indicates that a system can (or cannot) learn effective power of the richest when the current state is not in $\bm{\mathcal{S}}.$ 
Note that Eq.~\eqref{eq:max2} is derived by Lemma~\ref{lem:empty}.
Considering Lemma~\ref{lem:ep_delta} and \ref{lem:empty}, one can see that a mechanism satisfying 1) it is most profitable for all players to collude or for the richest to run only one node in a state that does not belong to $\bm{\mathcal{S}}$ and 2) $N(\bm{\mathcal{S}})\cap N(\bm{\mathcal{S}^c})=\emptyset$, can maximize the probability to achieve $(m,\varepsilon,\delta)$-decentralization. 
Moreover, $N(\bm{\mathcal{S}})\cap N(\bm{\mathcal{S}^c})=\emptyset$ implies that $N(\bm{\bar{EP}})=\emptyset$ or 
$f_{\alpha\rightarrow EP}(N(\bm{\bar{EP}}))\subset \bm{\mathcal{S}}$ for any $\bm{\bar{EP}}\in\bm{\mathcal{S}}.$

Next, we consider a mechanism where, for a state $\bm{\bar{EP}}$, it is most profitable for all players to form a grand coalition running only one node. 
Then all players would share reward $R=U_{n_i}(\bm{\bar{EP}})$. 
Here, we consider a scheme sharing the reward among joined accounts, and a player can have multiple accounts if the behavior is more profitable than that the one that is not. 
We also denote by $U_{a_i}(\alpha_{a_i}, \bm{\bar{\alpha}_{-a_i}})$ the received reward of account $a_i$ owned resource power $\alpha_{a_i}$.
Similar to the above progress, we can show that, in this case, the probability to reach $(m,\varepsilon,\delta)$-decentralization can be maximized when players above the $\delta$-th percentile should have one account. 
Note that when $A$ denotes the set of all accounts, $R=\sum_{a_i}U_{a_i\in A}(\alpha_{a_i}, \bm{\bar{\alpha}_{-a_i}})$ for any $A.$
Therefore, the conditions to maximize the probability to reach $(m,\varepsilon,\delta)$-decentralization in the sharing scheme correspond to the following: At least the richest player runs only one node, and ND-2 is satisfied. 
As a result, by Lemma~\ref{lem:m2}, we can derive that the probability to reach $(m, \varepsilon,\delta)$-decentralization is the maximum when the following is met: 
\begin{equation}
U_{a_i}(\alpha_{a_i}, \bm{\bar{\alpha}_{-a_i}})=\frac{R\cdot \alpha_{a_i}}{\sum_{a_i\in A}\alpha_{a_i}}.
\label{eq:acc_linear}
\end{equation}

Second, we consider a mechanism in which it is not most profitable for all players to collude and it is most profitable for the richest player to run only one node when the state is not in $\bm{\mathcal{S}}.$ 
In fact, this is equivalent to the case where GR-2 and ND-2 and NS-100 are satisfied. 
Therefore, from Lemma~\ref{lem:m2}, $U_{n_i}$ should be Eq.~\eqref{eq:linear} when the state is not in $\bm{\mathcal{S}}$. 

As a result, because Eq.~\eqref{eq:acc_linear} is also a form of Eq.~\eqref{eq:linear}, we can see that, through Lemma~\ref{lem:linear}, the probability to reach $(m,\varepsilon,\delta)$-decentralization can be maximized when GR-$|\mathcal{N}|$, ND-$|\mathcal{P}_\alpha|$, and NS-0 are met. 
Lastly, by presenting Lemma~\ref{lem:m2}, we completes the proof of Theorem~\ref{thm:suff_nec}. 

\begin{lemma}
Let us consider that GR-2, ND-2, and NS-100 are met. 
Then, in order that the probability of reaching $(m, \varepsilon, \delta)$-decentralization is the maximum, the following should be met:
\begin{equation}
U_{n_i}(\alpha_{n_i}, \bm{\bar{\alpha}_{-n_i}})=F\Bigl(\sum_{n_j\in\mathcal{N}}\alpha_{n_j}\Bigr)\cdot\alpha_{n_i}, \label{eq:linear2}
\end{equation}
where $F: \mathbb{R}^+\mapsto\mathbb{R}^+.$
\label{lem:m2}
\end{lemma}

\begin{proof}
According to ND-2 and NS-100, the following equation is satisfied for any $\alpha$ and set $\mathcal{N}_{\alpha}$ in which a node is an element and the total resource power of the elements is $\alpha$:
\begin{equation}
    \sum_{n_i\in\mathcal{N}_{\alpha}} U_{n_i}\Bigl(\alpha_{n_i},\bm{\bar{\alpha}_{-n_i}}(\mathcal{N}_{\alpha})\Bigr)
    = U_{n_j}(\alpha_{n_j}=\alpha),
    \label{eq:one}
\end{equation}
where node $n_j\in\mathcal{N}_{\alpha}$ and
$\bm{\bar{\alpha}_{-n_i}}(\mathcal{N}_{\alpha})=(\alpha_{n_k})_{n_k\in\mathcal{N}_{\alpha}, k\not=i}.$
Therefore, for all $n\in\mathbb{N},$ the following is met:  
\begin{equation}
    U_{n_i}\Bigl(\frac{\alpha}{n},\,\,\left[\frac{\alpha}{n}\right]^{n-1}\Bigr)
    = \frac{U_{n_i}(\alpha)}{n},
    \label{eq:n_linear}
\end{equation}
where $\left[\frac{\alpha}{n}\right]^{n-1}$ represents the array, which has $n-1$ elements $\frac{\alpha}{n}$. 
Note that $\left[\frac{\alpha}{n}\right]^n$ is one of possible candidates for $\mathcal{N}_\alpha$ because the sum of elements is $\alpha$. 

Moreover, according to Eq.~\eqref{eq:one} and Eq.~\eqref{eq:sybil} in NS-100, the following equations are met for any natural number $l<\frac{n}{2}$: 
\begin{equation*}
\begin{aligned}
    &U_{n_i}\Bigl(\frac{l\alpha}{n},\,\,\Bigl(\frac{(n-l)\alpha}{n}\Bigr)\Bigr)+U_{n_i}\Bigl(\frac{(n-l)\alpha}{n},\,\,\left(\frac{l\alpha}{n}\right)\Bigr)
    = U_{n_i}(\alpha),\\
    &U_{n_i}\Bigl(\frac{(n-l)\alpha}{n},\,\,\left(\frac{l\alpha}{n}\right)\Bigr)\geq (n-l)\cdot U_{n_i}\Bigl(\frac{\alpha}{n},\,\,\left[\frac{\alpha}{n}\right]^{n-1}\Bigr)
\end{aligned}
\end{equation*}
Because the lower the payoff of the richest, the more likely a system would reach $(m,\varepsilon,\delta)$-decentralization, the below equations should be met to maximize the probability to reach $(m,\varepsilon,\delta)$-decentralization. 
\begin{equation*}
\begin{aligned}
    &U_{n_i}\Bigl(\frac{(n-l)\alpha}{n},\,\,\left(\frac{l\alpha}{n}\right)\Bigr)= \frac{n-l}{n}\cdot U_{n_i}(\alpha),\\
    &U_{n_i}\Bigl(\frac{l\alpha}{n},\,\,\Bigl(\frac{(n-l)\alpha}{n}\Bigr)\Bigr)=\frac{l\cdot U_{n_i}(\alpha)}{n}
\end{aligned}
\end{equation*}
This fact implies that Eq.~\eqref{eq:linear2} is satisfied for any $\mathcal{P}$ of which size is two.

Next, we assume that Eq.~\eqref{eq:linear2} is satisfied for any $\mathcal{P}$ of which size is $k(<n)$. 
Then we show that 
$$U_{n_i}\Bigl(\frac{l_0\alpha}{n},\,\,\left(\frac{l_1\alpha}{n},\cdots, \frac{l_{k}\alpha}{n}\right)\Bigr)= \frac{l_0}{n}\cdot U_{n_i}(\alpha),$$
where $l_0, l_1, \cdots, l_k\in\mathbb{N}$ and $l_0=\max\{l_0, l_1, \cdots, l_k\}$. 
According to Eq.~\eqref{eq:sybil} and the assumption, the following is met for any $0<p \leq k$: 
\begin{equation*}
\begin{aligned}
    &\frac{l_0+l_p}{n}\cdot U_{n_i}(\alpha)=U_{n_i}\left(\frac{l_0\alpha}{n},\,\,\left(\frac{l_1\alpha}{n},\cdots, \frac{l_{k}\alpha}{n}\right)\right)+\\
    &U_{n_i}\left(\frac{l_p\alpha}{n},\,\,\left(\frac{l_0\alpha}{n},\cdots,\frac{l_{p-1}\alpha}{n},\frac{l_{p+1}\alpha}{n},\cdots,\frac{l_{k}\alpha}{n}\right)\right).
\end{aligned}
\end{equation*}
Moreover, the above equation derives the following. 
\begin{equation*}
\begin{aligned}
    k\cdot U_{n_i}&\left(\frac{l_0\alpha}{n},\,\,\left(\frac{l_1\alpha}{n},\cdots, \frac{l_{k}\alpha}{n}\right)\right)+\sum_{p=1}^{k}U_{n_i}\left(\frac{l_p\alpha}{n},\,\,\ast\right)
    \\&=\sum_{p=1}^{k}\frac{l_0+l_p}{n}\cdot U_{n_i}(\alpha),
\end{aligned}
\end{equation*}
where $\ast=\left(\frac{l_1\alpha}{n},\cdots, \frac{l_{k}\alpha}{n}\right).$
In addition, because 
$$\sum_{p=1}^{k}U_{n_i}\left(\frac{l_p\alpha}{n},\,\,\ast\right)=\sum_{p=1}^{k}\frac{l_p}{n}\cdot U_{n_i}(\alpha),$$
Eq.~\eqref{eq:linear2} is met for any $\mathcal{P}$ of which size is $k+1$. 
By mathematical induction, Eq.~\eqref{eq:linear2} holds for any $n$ and $k(<n)$, which implies that Eq.~\eqref{eq:linear2} is true when relative resource power of all nodes to total resource power is a rational number. 
As a result, by \textit{the density of the rational numbers}, Eq.~\eqref{eq:linear2} holds for any $\bm{\bar{\alpha}}.$ 
This completes the proof. 
\end{proof}

\section{Proof of Lemma~5.2}
\label{app:1}

The proof of Lemma~\ref{lem:linear} is similar to that for Lemma~\ref{lem:m2}. Thus, we briefly describe this proof. 
First, it is trivial for Eq.~\eqref{eq:linear} to satisfy GR-$|\mathcal{N}|$, ND-$|\mathcal{P}|$, and NS-$0$. 
Thus, we show the proof of the other direction. 
In other words, we prove that if the three conditions are met, the utility function should be Eq,~\eqref{eq:linear}. 
According to ND-$|\mathcal{P}|$ and NS-$0$, the following equation is satisfied for any $\alpha$:
\begin{equation*}
    \sum_{n_i\in\mathcal{N}_{\alpha}} U_{n_i}\Bigl(\alpha_{n_i},\bm{\alpha_{-n_i}^{+}}(\mathcal{N}_{\alpha})\Bigr)
    = U_{n_j}(\alpha_{n_j}=\alpha,\bm{\bar{\alpha}_{-\mathcal{N}_{\alpha}}}),
\end{equation*}
where node $n_j\in\mathcal{N}_{\alpha},$ the total resource power in the node set $\mathcal{N}_{\alpha}$ is $\alpha$, $\bm{\bar{\alpha}_{-\mathcal{N}_{\alpha}}}=(\alpha_{n_k})_{n_k\not\in\mathcal{N}_{\alpha}},$ and $\bm{\alpha_{-n_i}^{+}}(\mathcal{N}_{\alpha})=\bm{\bar{\alpha}_{-\mathcal{N}_{\alpha}}}\Vert (\alpha_{n_k})_{n_k\in \mathcal{N}_{\alpha}, n_k\not=n_i}.$
Therefore, for all $n\in\mathbb{N},$ the following is met:  
\begin{equation*}
    U_{n_i}\Bigl(\frac{\alpha}{n},\bm{\alpha_{-n_i}^{+}}(\mathcal{N}^{n_{\alpha}})\Bigr)
    = \frac{U_{n_j}(\alpha,\bm{\bar{\alpha}_{-\mathcal{N}^{n_{\alpha}}}})}{n},
\end{equation*}
where all nodes in $\mathcal{N}^n_\alpha$ possess $\frac{\alpha}{n}$ and $|\mathcal{N}^n_\alpha|=n$. 
Note that $\mathcal{N}^n_\alpha$ is one of possible candidates for $\mathcal{N}_\alpha$. 
The above equation derives the below equation: 
\begin{equation*}
    U_{n_i}\Bigl(\alpha_{n_i},\bm{\alpha_{-n_i}^{+}}(\mathcal{N}^\mathbb{Q}_{\alpha})\Bigr)
    = \frac{\alpha_{n_i}}{\alpha}\cdot U_{n_j}(\alpha,\bm{\bar{\alpha}_{-\mathcal{N}^{\mathbb{Q}_{\alpha}}}}),
\end{equation*}
where $\mathcal{N}^{\mathbb{Q}}_\alpha=\{n_i\,|\,\alpha_{n_i}=q_i\alpha, q_i\in\mathbb{Q}\}$ and node $n_j\in\mathcal{N}^{\mathbb{Q}}_\alpha$. Here, note that $\frac{\alpha_{n_i}}{\alpha}$ is a rational number.  
As a result, according to \textit{the density of the rational numbers}, 
the utility $U_{n_i}$ is a linear function for given the sum of resource power of nodes (i.e., $\sum_{n_i\in\mathcal{N}}\alpha_{n_i}$), where the coefficient is denoted by $F(\sum_{n_i\in\mathcal{N}}\alpha_{n_i})$ as a function of $\sum_{n_i\in\mathcal{N}}\alpha_{n_i}.$ 
Lastly, the coefficient $F(\sum_{n_i\in\mathcal{N}}\alpha_{n_i})$ should be positive to satisfy GR-$|\mathcal{N}|$.

\section{Proof of Theorem~5.3}
\label{app:3}

First, we consider that there is the minimum value of $\varepsilon(> 0)$ such that $\max_{x\leq A}xF(x)=(A-\varepsilon)F(A-\varepsilon)$ for a given value of $A$. 
Then, when $\sum\alpha_{n_i}$ is $A$, 
\begin{equation}
    \begin{aligned}
        &U\left(\alpha_{n_k}\cdot\frac{A-\varepsilon}{A},\bm{\bar{\alpha}_{-n_k}}\cdot \frac{A-\varepsilon}{A}\right)=F\left(A-\varepsilon\right)\cdot \alpha_{n_k}\frac{A-\varepsilon}{A}>\\ &U\left(\alpha_{n_k}\cdot \frac{A-\varepsilon^\prime}{A},\bm{\bar{\alpha}_{-n_k}}\cdot \frac{A-\varepsilon^\prime}{A}\right)=F\left(A-\varepsilon^\prime\right)\cdot \alpha_{n_k}\frac{A-\varepsilon^\prime}{A},
        \label{eq:control}
    \end{aligned}
\end{equation}
for any $\varepsilon^\prime<\varepsilon$.
Therefore, when all players reduce resource power of their node at the same rate, 
their node power would decrease from $\alpha_{n_k}$ to $\alpha_{n_k}\cdot\frac{A-\varepsilon}{A}$, and they earn a higher profit. 
We also consider the case where a node does not reduce its power by $\frac{\sum{\alpha_{n_i}}-\varepsilon}{\sum{\alpha_{n_i}}}$ times. 
However, the retaliation of other nodes can make this behavior less profitable when compared to the case where the node reduces its power by $\frac{\sum{\alpha_{n_i}}-\varepsilon}{\sum{\alpha_{n_i}}}$ times, where retaliation strategies are often used in a repeated game for cooperation. 
A possible strategy of node $n_i$ with resource power $\alpha_{n_i}^t$ is that the node updates its power $\alpha_{n_i}^t$ to $\alpha_{n_i}^{t+1}=\frac{A^{t+1}-\alpha_{n_i}^{t}}{A^t-\alpha_{n_i}^t}\cdot \alpha_{n_i}^t$ at time $t+1$, 
where $A^t$ denotes the total resource power of nodes at time $t$. 
Under this strategy, because of Eq.~\eqref{eq:control}, if even one node does not reduce its power by $\frac{A-\varepsilon}{A}$ times, all nodes earn a lower profit. 
As a result, there is a reachable equilibrium where all players reduce resource power of their node (i.e., effective power) by $\frac{A-\varepsilon}{A}$ times. 
Note that, in the equilibrium, the effective power distribution among players does not change. 

Second, we consider that $\max_{x\leq A}xF(x)=AF(A)$ for any $A$. 
This fact derives that
\begin{equation*}
    \begin{aligned}
    &U_{n_i}(\alpha_{n_i}+\varepsilon, \bm{\bar{\alpha}_{-n_i}})=(\alpha_{n_i}+\varepsilon)F\left(\sum_{n_i}\alpha_{n_i}+\varepsilon\right)>\\
    &U_{n_i}(\alpha_{n_i},\bm{\bar{\alpha}_{-n_i}})=\alpha_{n_i}F\left(\sum_{n_i}\alpha_{n_i}\right). 
    \end{aligned}
\end{equation*}
The above equation implies that the utility is a strictly increasing function for $\alpha_{n_i}$: $U_{n_i}(\alpha_{n_i}+\varepsilon,\bm{\bar{\alpha}_{-n_i}})>U_{n_i}(\alpha_{n_i}, \bm{\bar{\alpha}_{-n_i}})$ for any $\varepsilon>0.$ 
Thus, all nodes would increase their power for a higher profit. 

\begin{figure}[ht]
    \centering
    \includegraphics[width=0.65\columnwidth]{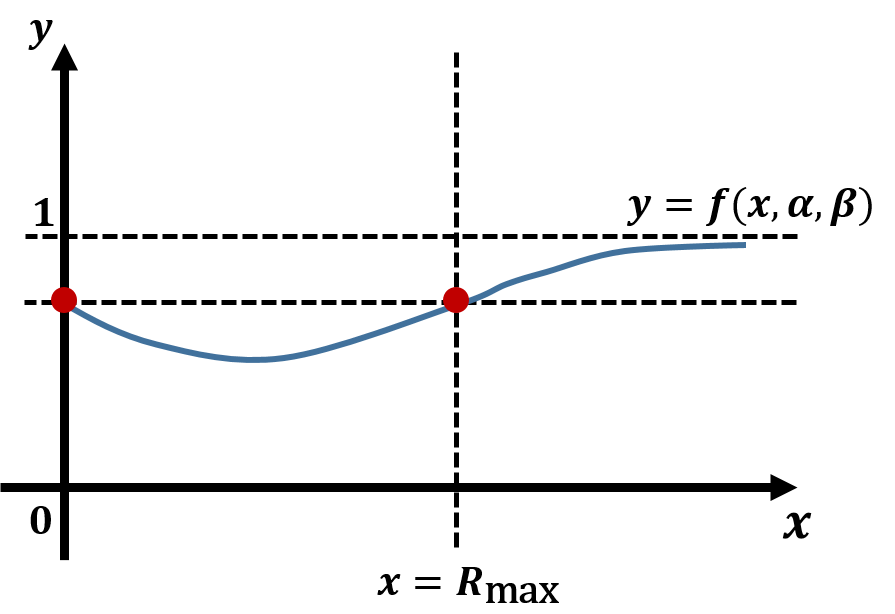}
    \caption{The function $f(x, \alpha,\beta)$ represents the right-hand side of Eq.~\eqref{eq:find_x}. 
    This graph shows that $f(R_{\tt max},\alpha,\beta)$ is the maximum in the range $x\leq R_{\tt max}.$}
    \label{fig:find_x}
\end{figure}

To satisfy the second requirement of Definition~\ref{def:perfect}, the following should be satisfied for any two players $p_i, p_j\in \mathcal{P}^t_{\delta}$: 
$$\frac{EP_{p_i}^t}{EP_{p_j}^t}\leq 1+\varepsilon,$$ where $EP_{p_i}^t \geq EP_{p_j}^t.$ 
Under the utility function Eq.~\eqref{eq:linear}, a player would run one node with its own resource, and the above equation can be expressed as follows: $\frac{\alpha_{n_i}^t}{\alpha_{n_j}^t}\leq 1+\varepsilon,$ 
where $\mathcal{N}_{p_i}=\{n_i\}$ and $\mathcal{N}_{p_j}=\{n_j\}.$
Because $U(\alpha_{n_i}, \bm{\bar{\alpha}_{-n_i}})$ is a strictly increasing function of $\alpha_{n_i},$ all nodes would increase their resource at rate $r$ per earned net profit. 
Then the ratio $\frac{\alpha_{n_i}^{t+1}}{\alpha_{n_j}^{t+1}}$ between the resource power of nodes $n_i$ and $n_j$ at time $t+1$ is 
\begin{equation*}
    \begin{aligned}
    \frac{\alpha_{n_i}^t+r\cdot R_{n_i}^t}{\alpha_{n_j}^t+r\cdot R_{n_j}^t}=\frac{\alpha_{n_i}^t}{\alpha_{n_j}^t}\cdot \frac{1+r\cdot \frac{R_{n_i}^t}{\alpha_{n_i}^t}}{1+r\cdot \frac{R_{n_j}^t}{\alpha_{n_j}^t}}>\frac{\alpha_{n_i}^t}{\alpha_{n_j}^t}\cdot\frac{1}{1+r\cdot \frac{R_{n_j}^t}{\alpha_{n_j}^t}}.
    \end{aligned}
\end{equation*}
For ease of reading, a state where $\alpha_{n_i}=\alpha$ and $\alpha_{n_j}=\beta$ is denoted by $(\alpha,\beta)$.
Here, note that $\alpha$ is not less than $\beta.$
Then we consider one step in which $(\alpha,\beta)$ moves to $(\alpha,\beta+ry)$ with probability $p$ and $(\alpha+rx,\beta)$ with probability $1-p$, where $x, y\leq R_{\tt max}$.
Because of $\frac{U_{n_j}(\beta)}{\beta}-\frac{U_{n_i}(\alpha)}{\alpha}=0,$ $p=\frac{x}{x+\frac{\alpha y}{\beta}}$. 
We also denote $\Pr\left(a\rightarrow b\,|\,(\alpha,\beta)\right)$ by the probability for ratio $\frac{\alpha_{n_i}}{\alpha_{n_j}}$ to reach from $a$ to less than $b$ when a state $(\alpha_{n_i},\alpha_{n_j})$ starts from $(\alpha, \beta)$.
Then the following holds:
\begin{equation}
\begin{aligned}
    &\Pr\left(\frac{\alpha}{\beta}\rightarrow 1+\varepsilon \,\Bigg|\,(\alpha,\beta)\right)\leq \frac{\beta x}{\beta x+\alpha y}\times \\
    &\max\Pr\left(\frac{\alpha}{\beta+ry}\rightarrow 1+\varepsilon \,\Bigg|\,(\alpha,\beta+ry)\right)+\frac{\alpha y}{\beta x+\alpha y}\\
    &\times\max\Pr\left(\frac{\alpha+rx}{\beta}\rightarrow 1+\varepsilon\,\Bigg|\,(\alpha+rx,\beta)\right),
\end{aligned}
\label{eq:find_x}
\end{equation}
where $\max\Pr(\frac{\alpha}{\beta+ry}\rightarrow 1+\varepsilon \,|\,(\alpha,\beta+ry))$ indicates the maximum probability for $(\alpha_{n_i},\alpha_{n_j})$ to reach from $(\alpha,\beta+ry)$ to a state satisfying that $\frac{\alpha_{n_i}}{\alpha_{n_j}}\leq 1+\varepsilon$, considering all possible random walks.
Similarly, $\max\Pr(\frac{\alpha+rx}{\beta}\rightarrow 1+\varepsilon\,|\,(\alpha+rx,\beta))$ represents the maximum probability for $(\alpha_{n_i},\alpha_{n_j})$ to reach from $(\alpha+rx,\beta)$ to a state satisfying that $\frac{\alpha_{n_i}}{\alpha_{n_j}}\leq 1+\varepsilon$.
Note that, in the range $0\leq x\leq R_{\tt max},$ the right-hand side of Eq.~\eqref{eq:find_x} is the maximum when $x=0$. 

We denote the right-hand side of Eq.~\eqref{eq:find_x} by $f(x, \alpha,\beta).$
Then, when assuming 1) $\lim_{\alpha\rightarrow\infty}f(x, \alpha,\beta)$ is a constant in terms of $x$ and 2) 
$f(x, \alpha,\beta)$ is the maximum when $x=R_{\tt max}$,
the probability to reach $(m,\varepsilon,\delta)$-decentralization is upper bounded by the maximum probability to reach $(m,\varepsilon,\delta)$-decentralization under \textit{a random walk where $\alpha_{n_i}$ changes to $\alpha_{n_i}+rR_{\tt max}$ if it increases.}
For the second assumption, Fig.~\ref{fig:find_x} describes an example. 
Note that the value of when $x=0$ cannot be greater than that for when $x=R_{\tt max}$ because $\max\Pr(\frac{\alpha}{\beta}\rightarrow 1+\varepsilon \,|\,(\alpha,\beta))$ is not greater than $f(x, \alpha,\beta)$. 
Moreover, the above fact derives that, even if we extend to one step in which $(\alpha,\beta)$ can move to $(\alpha,\beta+ry)$,  $(\alpha+rx_1,\beta)$, $(\alpha+rx_2,\beta), \cdots, (\alpha+rx_n,\beta)$, the probability for the ratio $\frac{\alpha_{n_i}}{\alpha_{n_j}}$ to reach from $\frac{\alpha}{\beta}$ to less than $1+\varepsilon$ can be the maximum when $x_i=R_{\tt max}$ for $1\leq i\leq n$.
Also, when considering one step where $(\alpha,\beta)$ can move to $(\alpha,\beta+ry_1)$, $(\alpha,\beta+ry_2)$, $\cdots$, $(\alpha,\beta+ry_n)$, $(\alpha+rx,\beta)$, the probability for the ratio $\frac{\alpha_{n_i}}{\alpha_{n_j}}$ to reach from $\frac{\alpha}{\beta}$ to less than $1+\varepsilon$ can be the maximum if $x=R_{\tt max}$.
This is because such steps can be expressed as a linear combination of a step $s_i$ for $i\leq n$ in which $(\alpha,\beta)$ can move to $(\alpha,\beta+ry_i)$ or $(\alpha+rx_i,\beta)$.
As a result, these facts imply that it is sufficient to find a function $G(\alpha,\beta)$ satisfying the following.
\begin{enumerate}
    \item The function $G(\alpha,\beta)$ is equal to or greater than 
    $$\max_{x=R_{\tt max}}\Pr\left(\frac{\alpha}{\beta} \rightarrow 1+\varepsilon\,\big|\,(\alpha,\beta) \right).$$
    \item The following equation is the maximum when $x=R_{\tt max}.$
    \begin{equation}
        \begin{aligned}
\max_{y}&\left\{\frac{\beta x}{\beta x+\alpha y}\cdot G(\alpha,\beta+ry)+\right.\\
&\hspace{1.5cm}\left.\frac{\alpha y}{\beta x+\alpha y}\cdot G(\alpha+rx,\beta)\right\}.\label{eq:const}
        \end{aligned}
    \end{equation}
    \item The limit value of Eq.~\eqref{eq:const} when $\alpha$ goes to infinity is a constant in terms of $x.$
    \item The below equation holds:
    \begin{equation*}
    \begin{aligned}
        G(\alpha,\beta)\geq \max_{y}&\left\{\frac{\beta R_{\tt max}}{\beta R_{\tt max}+\alpha y}\cdot G(\alpha,\beta+ry)+\right.\\
&\hspace{0.8cm}\left.\frac{\alpha y}{\beta R_{\tt max}+\alpha y}\cdot G(\alpha+rR_{\tt max},\beta)\right\}.
    \end{aligned}
    \end{equation*}
\end{enumerate}

\begin{figure*}[ht]
	\centering
    \subfloat[]{
      \includegraphics[height=0.2\textwidth]{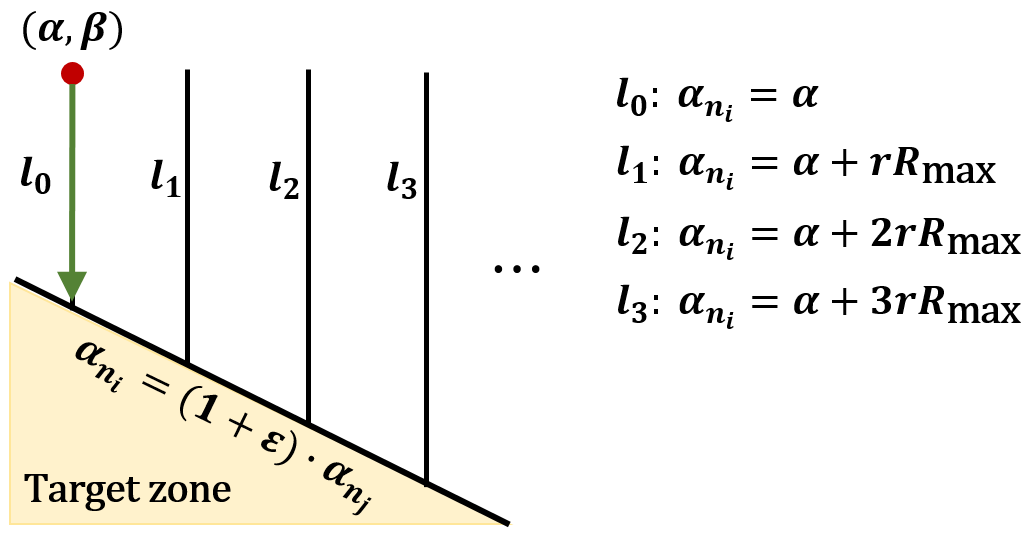}
      \label{fig:ex1}
   }
    \subfloat[]{
      \includegraphics[height=0.2\textwidth]{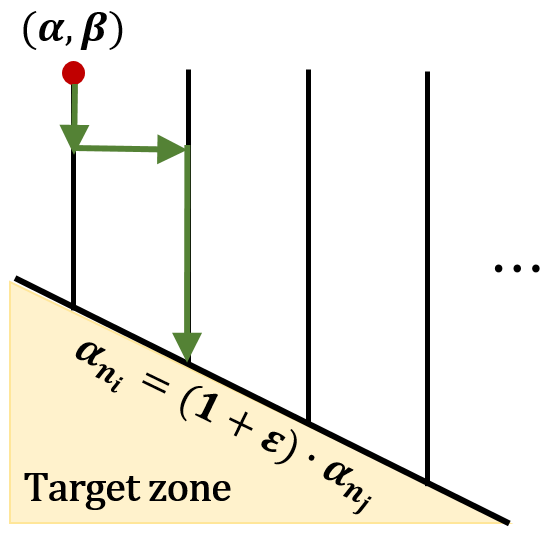}
      \label{fig:ex2}
   }
    \subfloat[]{
      \includegraphics[height=0.2\textwidth]{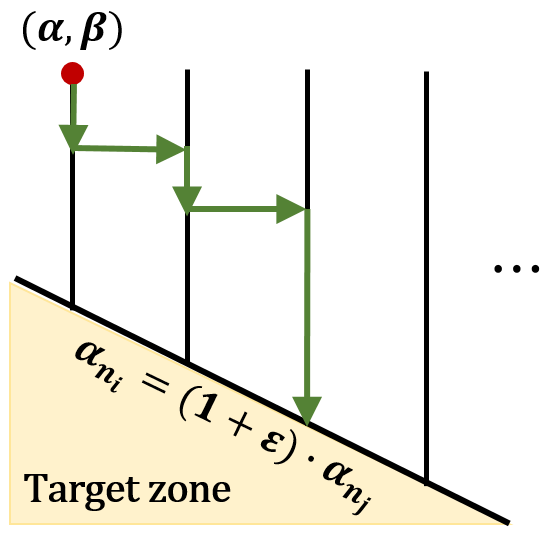}
      \label{fig:ex3}
   }
\caption{The figures represent examples for events of which probabilities are $P_0^\varepsilon(\alpha,\beta), P_1^\varepsilon(\alpha,\beta)$, and $P_2^\varepsilon(\alpha,\beta)$, respectively. 
The red point $(\alpha,\beta)$ is a start point, and a random walk aims to enter the target zone in which $\frac{\alpha_{n_i}}{\alpha_{n_j}}\leq 1+\varepsilon$. 
The lines $l_0, l_1, l_2$ and $l_3$ represent $\alpha_{n_i}=\alpha, \alpha_{n_i}=\alpha+R_{\tt max}, \alpha_{n_i}=\alpha+2R_{\tt max},$ and  $\alpha_{n_i}=\alpha+3R_{\tt max}$, respectively.
The point $(\alpha_{n_i},\alpha_{n_j})$ would descend along the current line or move to the next line.}
\label{fig:examples}
\end{figure*}

Next, we consider the case where the ratio $\frac{\alpha_{n_i}}{\alpha_{n_j}}$ changes from $\frac{\alpha}{\beta}$ to less than $1+\varepsilon$ without a process in which $\alpha_{n_i}$ increases from $\alpha$ to $\alpha+rR_{\tt max}$. 
The probability for the case is denoted by $P_0^\varepsilon(\alpha,\beta)$. 
In addition, for the case where $\frac{\alpha_{n_i}}{\alpha_{n_j}}$ changes from $\frac{\alpha}{\beta}$ to less than $1+\varepsilon$ with a process in which $\alpha_{n_i}$ increases from $\alpha$ to $\alpha+krR_{\tt max}$ but not to $\alpha+(k+1)rR_{\tt max}$,
its probability is denoted by $P_k^\varepsilon(\alpha,\beta)$.
Fig.~\ref{fig:examples} represents examples for events of which probabilities are $P_0^\varepsilon(\alpha,\beta), P_1^\varepsilon(\alpha,\beta)$, and $P_2^\varepsilon(\alpha,\beta)$, respectively.
For ease of reading, we also denote $\frac{R_{n_j}}{\alpha_{n_j}}-\frac{R_{n_i}}{\alpha_{n_i}}$ by $D$,
and then, the following holds:
\begin{align*}
    &\frac{U_{n_j}(\alpha_{n_j},\bm{\bar{\alpha}_{-n_j}})}{\alpha_{n_j}}-\frac{U_{n_i}(\alpha_{n_i},\bm{\bar{\alpha}_{-n_i}})}{\alpha_{n_i}}=0\\=&\int_{ D\geq d}D\Pr(D)+\int_{D< d}D\Pr(D)\\ \geq \,&d\Pr(D\geq d)-\frac{R_{\tt max}}{\alpha_{n_i}}(1-\Pr(D\geq d))\\
    \Rightarrow &\Pr(D\geq d)\leq \frac{R_{\tt max}}{R_{\tt max}+d\alpha_{n_i}}\\
    \Rightarrow &\Pr(\frac{R_{n_j}}{\alpha_{n_j}}\geq d, R_{n_i}=0)\leq \frac{R_{\tt max}}{R_{\tt max}+d\alpha_{n_i}}
\end{align*}

By the above equation, we can also derive the following:
\begin{align}
    &\Pr\Bigl(\frac{\alpha_{n_i}^{t+1}}{\alpha_{n_j}^{t+1}}\leq x \Big| \alpha_{n_i}^t,\alpha_{n_j}^t\Bigr)\notag
    \\\leq&\Pr\Bigl(\frac{R_{n_j}^t}{\alpha_{n_j}^t}\leq\frac{1}{r}\Bigl(\frac{1}{x}\cdot\frac{\alpha_{n_i}^t}{\alpha_{n_j}^t}-1\Bigr)=d \Big| \alpha_{n_i}^t,\alpha_{n_j}^t\Bigr)\notag
    \\\leq&\frac{R_{\tt max}}{R_{\tt max}+d\alpha_{n_i}^t}=\frac{R_{\tt max}(1+rd)}{(R_{\tt max}+d\alpha_{n_i}^{t})(1+rd)}\notag
    \\\leq& \frac{1}{1+rd}\leq x\cdot \frac{\alpha_{n_j}^t}{\alpha_{n_i}^t} \quad \text{ if } \quad R_{\tt max}\cdot r\leq\alpha_{n_i}^t 
    \label{eq:pf}
\end{align}
Assuming that $\beta\prod_{t=1}^{n}(1+rd^t)=\frac{\alpha}{1+\varepsilon}$, Eq.~\eqref{eq:pf} implies 
\begin{align}
    &P_0^\varepsilon(\alpha,\beta)= \prod_{t=1}^{n}\frac{R_{\tt max}}{R_{\tt max}+d^t\alpha}=\prod_{t=1}^{n}\frac{R_{\tt max}(1+rd^t)}{(R_{\tt max}+d^t\alpha)(1+rd^t)}\notag\\&
    \leq (1+\varepsilon)\cdot \frac{\beta}{\alpha} \quad \text{ if } \quad R_{\tt max}\cdot r\leq\alpha. \notag 
\end{align}
Furthermore, 
\begin{align*}
\max P_0^\varepsilon(\alpha,\beta)&=\max_{(d^1,\cdots, d^n)\in S_1}\prod_{t=1}^{n}\frac{R_{\tt max}}{R_{\tt max}+d^t\alpha}\\
&\leq\max_{(d^1,\cdots, d^n)\in S_2}\prod_{t=1}^{n}\frac{R_{\tt max}}{R_{\tt max}+d^t\alpha},    
\end{align*}
where 
\begin{equation*}
    \begin{aligned}
&S_1=\left\{(d^1,\cdots, d^n)\,\Big|\, 0\leq d^t\leq \frac{R_{\tt max}}{\beta\prod_{i=1}^{t-1}(1+rd^i)},\right.\\ 
&\hspace{2.5cm}\left.\beta\prod_{t=1}^{n}(1+rd^t)=\frac{\alpha}{1+\varepsilon}\right\}\subset \\
&S_2=\left\{(d^1,\cdots, d^n)\,\Big|\, 0\leq d^t, \beta\prod_{t=1}^{n}(1+rd^t)=\frac{\alpha}{1+\varepsilon}\right\}.    
    \end{aligned}
\end{equation*}
Because $\prod_{t=1}^{n}\frac{R_{\tt max}}{R_{\tt max}+d^t\alpha}$ is a symmetric and convex function for variables $d^1, d^2, \cdots, d^n,$
it would be the maximum when a point $(d^1, d^2, \cdots, d^n)$ is on the boundary of a set $A_2.$
In other words, if 
\begin{equation*}
d^1=\frac{1}{r}\left(\frac{\alpha}{\beta(1+\varepsilon)}-1\right) \,\,\text{and}\,\, d^t=0 \quad\forall t>1,    
\end{equation*}
the value of $\prod_{t=1}^{n}\frac{R_{\tt max}}{R_{\tt max}+d^t\alpha}$ is the maximum. 
Meanwhile, $\prod_{t=1}^{n}\frac{R_{\tt max}}{R_{\tt max}+d^t\alpha}$ is the minimum if $d^1, d^2, \cdots, d^n$ are the same. 
In addition, when $R_{\tt max}\cdot r=\alpha$, $P^\varepsilon_0(\alpha, \beta)$ can be maximized, and the value is $(1+\varepsilon)\frac{\beta}{\alpha}$. 

We define $Pr_k\left((\alpha, \beta)\rightarrow(\alpha+krR_{\tt max}, \beta^\prime)\right)$ as the probability of an event where a point $(\alpha_{n_i}, \alpha_{n_j})$ starting from $(\alpha, \beta)$ reaches the line $\alpha_{n_i}=\alpha+krR_{\tt max}$ before satisfying $\frac{\alpha_{n_i}}{\alpha_{n_j}}\leq 1+\varepsilon$, and the value of $\alpha_{n_j}$ of the point at which $(\alpha_{n_i}, \alpha_{n_j})$ meets the line $\alpha_{n_i}=\alpha+kR_{\tt max}$ for the first time is $\beta^\prime.$
Then, for the probability $P^\varepsilon_k(\alpha, \beta)$, the following holds: 
\begin{equation}
\begin{aligned}
&P^\varepsilon_k(\alpha, \beta)=
\sum_{\beta^\prime}Pr_k\left((\alpha, \beta)\rightarrow(\alpha+krR_{\tt max}, \beta^\prime)\right)\times \\
& P^\varepsilon_0(\alpha+krR_{\tt max}, \beta^\prime)\leq \sum_{\beta^\prime}Pr_k\left((\alpha, \beta)\rightarrow(\alpha+krR_{\tt max}, \beta^\prime)\right) \\
&\hspace{1cm}\times \frac{rR_{\tt max}}{rR_{\tt max}+(\alpha+krR_{\tt max})\cdot\left(\frac{\alpha+krR_{\tt max}}{\beta^\prime(1+\varepsilon)}-1\right)}
\label{eq:hk}
\end{aligned}
\end{equation}
We denote the right-hand side of Eq.~\eqref{eq:hk} by $H_k(\alpha, \beta)$.
Note that the value of $H_k(\alpha, \beta)$ indicates the probability of an event in which the point $(\alpha_{n_i}, \alpha_{n_j})$ meets the line $\alpha_{n_i}=\alpha+krR_{\tt max}$ and moves from $(\alpha, \beta)$ to a point satisfying $\frac{\alpha_{n_i}}{\alpha_{n_j}}\leq 1+\varepsilon$.
In this event, if $(\alpha_{n_i}, \alpha_{n_j})$ is on the point $(\alpha+krR_{\tt max},\beta^\prime)$, it can reach a point satisfying $\frac{\alpha_{n_i}}{\alpha_{n_j}}\leq 1+\varepsilon$ with probability 
\begin{equation}
\frac{rR_{\tt max}}{rR_{\tt max}+(\alpha+krR_{\tt max})\cdot\left(\frac{\alpha+krR_{\tt max}}{\beta^\prime(1+\varepsilon)}-1\right)}.  
\label{eq:max_0}
\end{equation}
Therefore, the value of $H_k(\alpha, \beta)$ depends on how the point $(\alpha_{n_i}, \alpha_{n_j})$ reaches the line $\alpha_{n_i}=\alpha+krR_{\tt max}.$ 

\begin{figure*}[ht]
	\centering
    \subfloat[Random walk $\mathcal{W}_1$]{
      \hspace{-5mm}\includegraphics[height=0.23\textwidth]{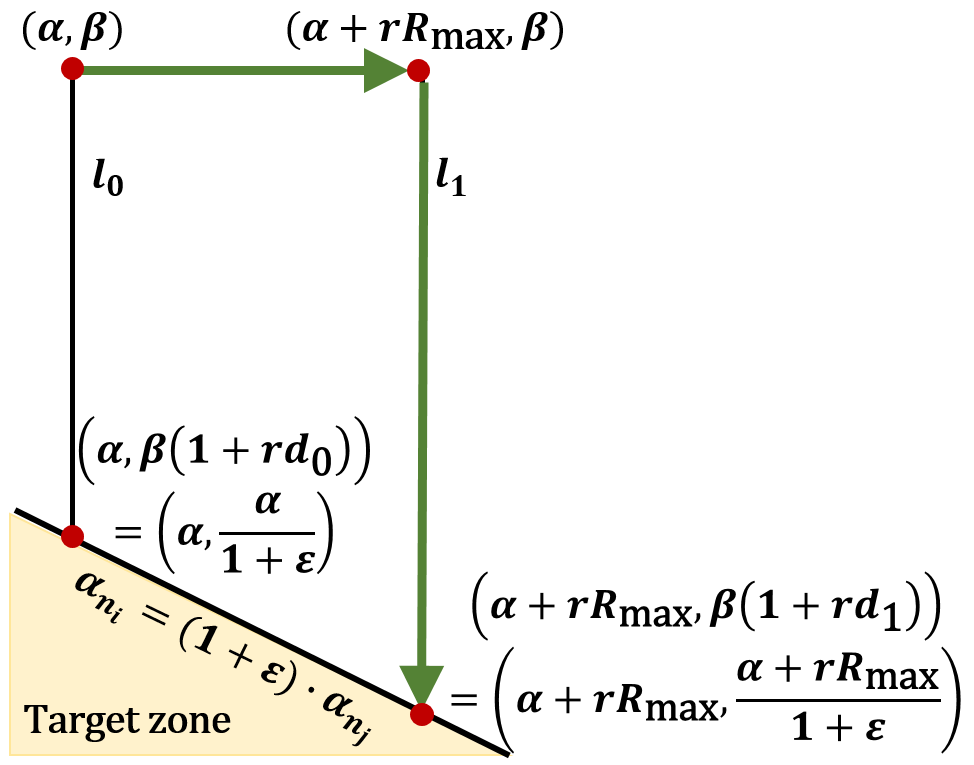}
      \label{fig:w1}
   }
    \subfloat[Random walk $\mathcal{W}_2$]{
      \includegraphics[height=0.23\textwidth]{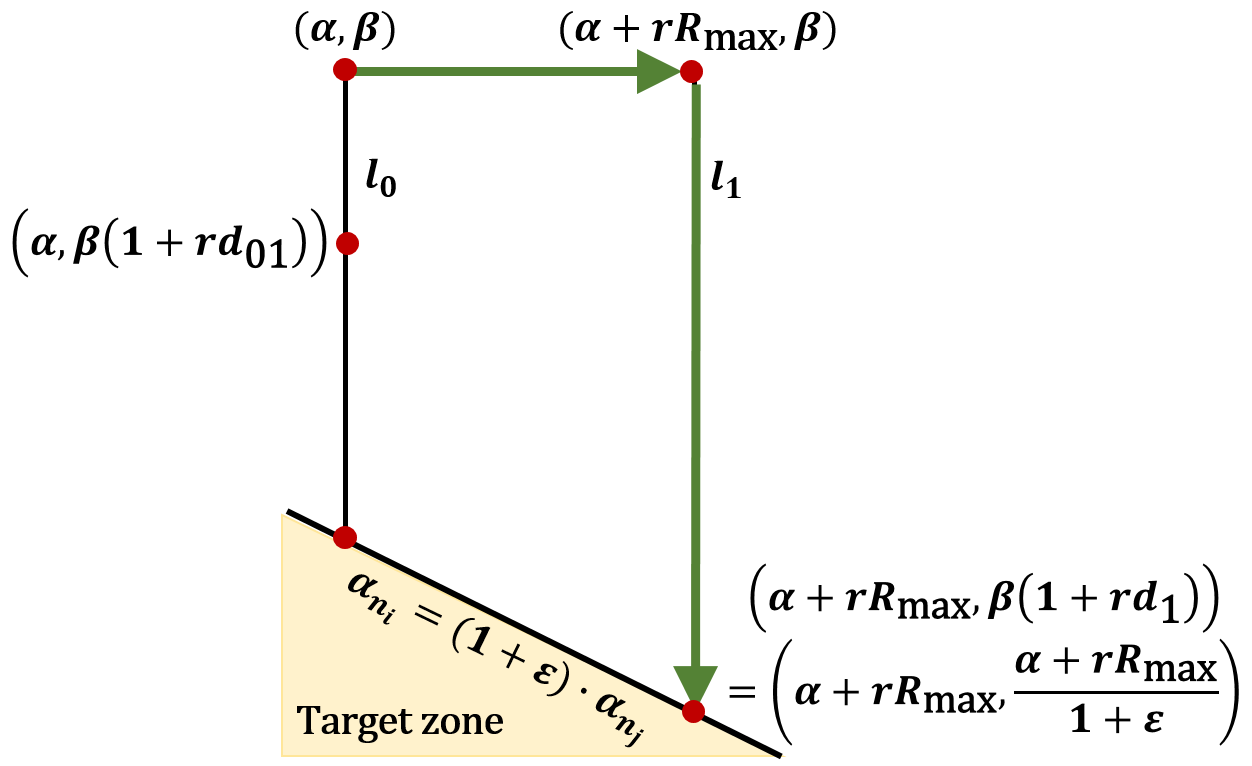}
      \label{fig:w2_1}
   }
   \subfloat[Random walk $\mathcal{W}_2$]{
      \includegraphics[height=0.23\textwidth]{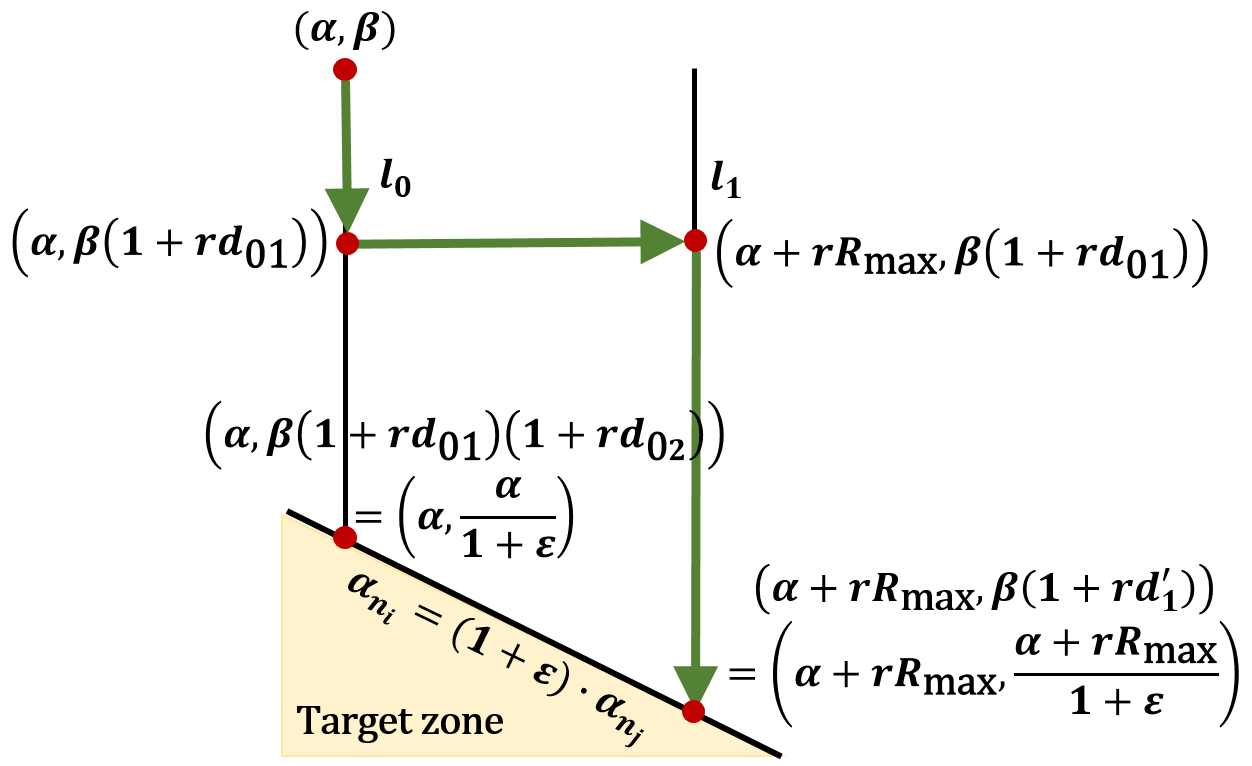}
      \label{fig:w2_2}
   }
\caption{The figures represent two random walks $\mathcal{W}_1$ and $\mathcal{W}_2$, respectively.
The red points indicate points to which the state $(\alpha_{n_i}, \alpha_{n_j})$ can move through each random walk. 
Moreover, green paths indicate the possible path in each random walk. 
In $\mathcal{W}_2$, there is one red point $(\alpha, \beta(1+rd_{01}))$ on line $l_0$ in addition to the red point of $\mathcal{W}_1$. 
Here, $1+rd_0=(1+rd_{01})(1+rd_{02})$.}
\label{fig:walk_ex}
\end{figure*}

\begin{figure*}[ht]
	\centering
    \subfloat[Random walk $\mathcal{W}_3$]{
      \hspace{-5mm}\includegraphics[height=0.25\textwidth]{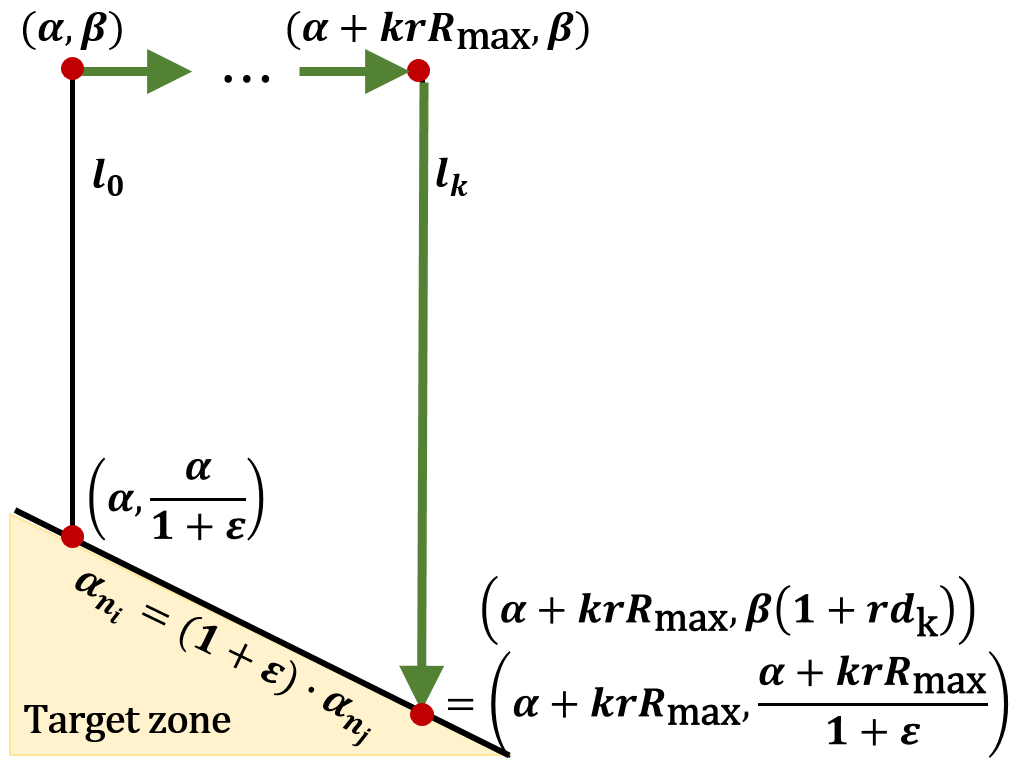}
      \label{fig:wk3}
   }
    \subfloat[Random walk $\mathcal{W}_4$]{
      \includegraphics[height=0.25\textwidth]{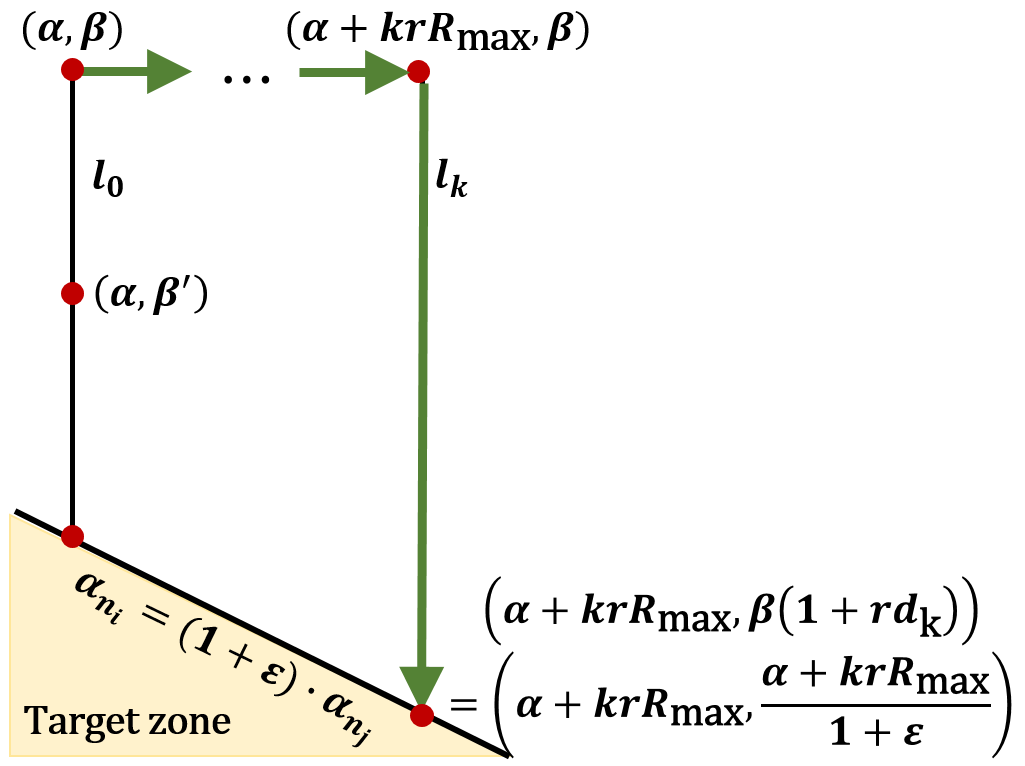}
      \label{fig:wk4_1}
   }
   \subfloat[Random walk $\mathcal{W}_4$]{
      \includegraphics[height=0.25\textwidth]{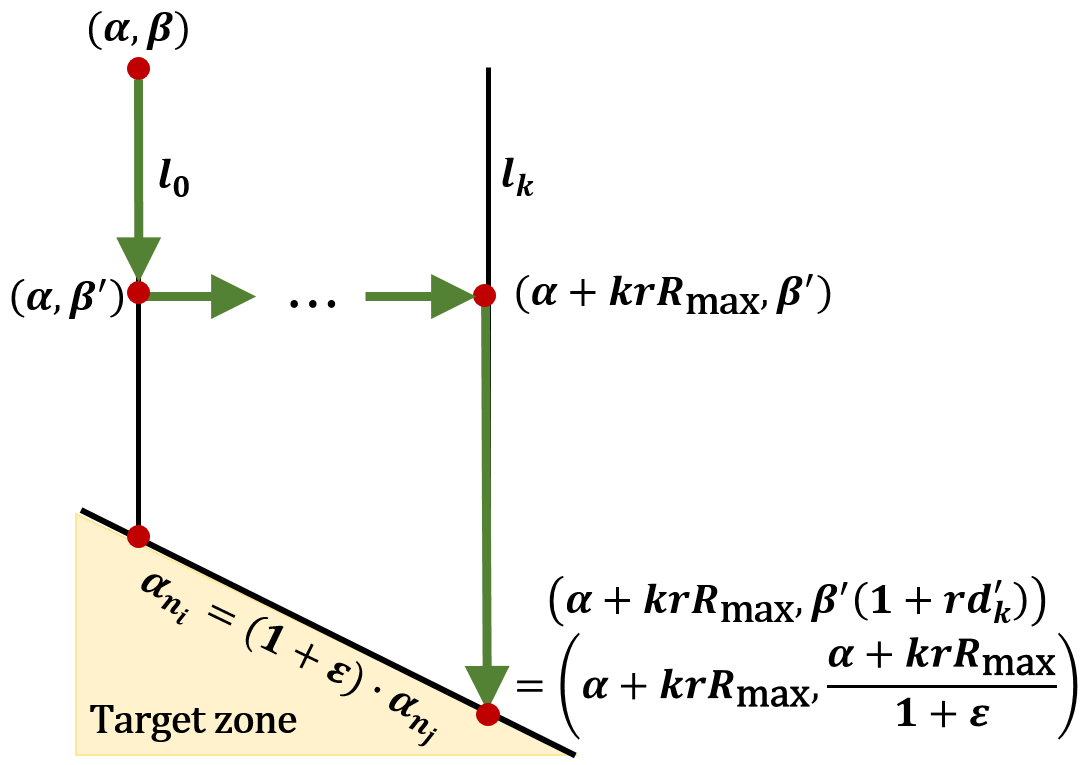}
      \label{fig:wk4_2}
   }
\caption{The figures represent two random walks $\mathcal{W}_3$ and $\mathcal{W}_4$, respectively.
The red points indicate points to which the state $(\alpha_{n_i}, \alpha_{n_j})$ can change the moving direction. 
Moreover, green paths indicate the possible path in each random walk. 
In $\mathcal{W}_4$, there is another red point $(\alpha, \beta^\prime)$ on line $l_0$ in addition to the red point of $\mathcal{W}_3$. }
\label{fig:walk_exk}
\end{figure*}

Next, we find when $H_k(\alpha, \beta)$ can be maximized. 
Note that the value of $H_0(\alpha, \beta)$ is determined as Eq.~\eqref{eq:max_0}.
Thus, we first consider when $k=1$ and denote the value of $H_k(\alpha, \beta)$ under a random walk $\mathcal{W}$ by $H_k^{\mathcal{W}}(\alpha, \beta)$. 
Also, we assume that two random walks $\mathcal{W}_1$ and $\mathcal{W}_2$ exist.
In $\mathcal{W}_1$, the point $(\alpha_{n_i}, \alpha_{n_j})$ on the line $\alpha_{n_i}=\alpha$ can move to either the point $(\alpha+rR_{\tt max}, \alpha_{n_j})$ or the point $(\alpha, \frac{\alpha}{1+\varepsilon}).$ 
If the point is on the line $\alpha_{n_i}=\alpha+rR_{\tt max},$ it can move to either the point $(\alpha+2rR_{\tt max}, \alpha_{n_j})$ or the point $(\alpha+rR_{\tt max}, \frac{\alpha+rR_{\tt max}}{1+\varepsilon}).$
The random walk $\mathcal{W}_2$ is similar to $\mathcal{W}_1$ except that there is one additional path from the line $\alpha_{n_i}=\alpha$ to the line $\alpha_{n_i}=\alpha+rR_{\tt max}$ when compared to $\mathcal{W}_1$.
Fig.~\ref{fig:walk_ex} represents $\mathcal{W}_1$ and $\mathcal{W}_2$.
While the random walk $\mathcal{W}_1$ has only one point $(\alpha+rR_{\tt max}, \beta)$ at which a state $(\alpha_{n_i}, \alpha_{n_j})$ can meet the line $l_1$, $(\alpha_{n_i}, \alpha_{n_j})$ can meet the line $l_1$ at two points $(\alpha+rR_{\tt max}, \beta)$ and $(\alpha+rR_{\tt max}, \beta(1+rd_{01}))$ in random walk $\mathcal{W}_2$. 
Fig.~\ref{fig:w1} represents the possible path of random walk $\mathcal{W}_1$, and Figs.~\ref{fig:w2_1} and \ref{fig:w2_2} show two possible paths of random walk $\mathcal{W}_2$.

We show that $H_1^{\mathcal{W}_2}(\alpha, \beta)$ is greater than $H_1^{\mathcal{W}_1}(\alpha, \beta)$.
Referring to Fig.~\ref{fig:walk_ex}, the following is met: 
\begin{align}
&\beta(1+rd_0)=\frac{\alpha}{1+\varepsilon},\quad\notag
\beta(1+rd_1)=\frac{\alpha+rR_{\tt max}}{1+\varepsilon},\notag\\
&\beta(1+rd_{01})(1+rd_{02})=\frac{\alpha}{1+\varepsilon},\notag\\
&\beta(1+rd_{01})(1+rd_{1}^\prime)=\frac{\alpha+rR_{\tt max}}{1+\varepsilon},\notag\\
&R_{\tt max}(R_{\tt max}+\alpha d_{0})\leq(R_{\tt max}+\alpha d_{01})(R_{\tt max}+\alpha d_{02}).\label{eq:h1}
\end{align}
Also, $H_1^{\mathcal{W}_1}(\alpha, \beta)$ and $H_1^{\mathcal{W}_2}(\alpha, \beta)$ are 
$$\frac{R_{\tt max}}{R_{\tt max}+(\alpha+rR_{\tt max})d_1}\cdot \frac{\alpha d_0}{R_{\tt max}+\alpha d_0} \text{ and }$$ 
\begin{equation*}
    \begin{aligned}
&\frac{\alpha d_{01}}{R_{\tt max}+\alpha d_{01}}\cdot \frac{R_{\tt max}}{R_{\tt max}+(\alpha+rR_{\tt max})d_1}+
\frac{R_{\tt max}}{R_{\tt max}+\alpha d_{01}}\\
&\times \frac{\alpha d_{02}}{R_{\tt max}+\alpha d_{02}}\cdot\frac{R_{\tt max}}{R_{\tt max}+(\alpha+rR_{\tt max})d_1^\prime},
    \end{aligned}
\end{equation*} respectively.
Because of Eq.~\eqref{eq:h1}, $H_1^{\mathcal{W}_2}(\alpha, \beta)$ is less than 
\begin{equation}
    \begin{aligned}
&\frac{\alpha d_{01}}{R_{\tt max}+\alpha d_{01}}\cdot \frac{R_{\tt max}}{R_{\tt max}+(\alpha+rR_{\tt max}d_1)}+\left(\frac{R_{\tt max}}{R_{\tt max}+\alpha d_{01}}\right.\\
&\left.-\frac{R_{\tt max}}{R_{\tt max}+\alpha d_{0}}\right)\cdot \frac{R_{\tt max}}{(\alpha+rR_{\tt max})d_1^\prime}.  
\end{aligned}
\label{eq:h1_2}
\end{equation}
By the below equations, Eq.~\eqref{eq:h1_2} is greater than $H_1^{\mathcal{W}_1}(\alpha, \beta)$.
\begin{align*}
    &\frac{1}{R_{\tt max}+(\alpha+rR_{\tt max})d_1^\prime}\geq
    \frac{1}{R_{\tt max}+(\alpha+rR_{\tt max})d_1}\Leftrightarrow\\ 
    &\frac{1}{R_{\tt max}+(\alpha+rR_{\tt max})d_1^\prime}\times \left(\frac{R_{\tt max}}{R_{\tt max}+\alpha d_{01}}-\frac{R_{\tt max}}{R_{\tt max}+\alpha d_{0}}\right)\\
    &\geq\frac{1}{R_{\tt max}+(\alpha+rR_{\tt max})d_1}\times \left(\frac{\alpha d_{0}}{R_{\tt max}+\alpha d_{0}}-\frac{\alpha d_{01}}{R_{\tt max}+\alpha d_{01}}\right)
    \\&\Leftrightarrow H_1^{\mathcal{W}_1}(\alpha, \beta)<\text{Eq.~\eqref{eq:h1_2}}
\end{align*}
Here, note that $d_1^\prime<d_1.$
As a result, $H_1^{\mathcal{W}_2}(\alpha, \beta)>H_1^{\mathcal{W}_1}(\alpha, \beta).$
Moreover, $H_1^{\mathcal{W}_2}(\alpha, \beta)$ is a concave function of $d_{01},$ which implies that the value of $H_1^{\mathcal{W}_2}(\alpha, \beta)$ would more efficiently increase when $d_{01}$ is closer to 0. 
Considering this fact, we can see that the more densely there exist points at which $(\alpha_{n_i},\alpha_{n_j})$ can meet the line $\alpha_{n_i}=\alpha+rR_{\tt max}$ for the first time, the greater the value of $H_1(\alpha, \beta)$ is. 

Next, we consider two random walks, $\mathcal{W}_3$ and $\mathcal{W}_4$, and find when $H_k^\mathcal{W}(\alpha, \beta)$ can be maximized.
Hereafter, a point $(\alpha_{n_i}, \alpha_{n_j})$, which can move to the next line, (e.g., red points represented in Fig.~\ref{fig:walk_ex}) is called a break point. 
The random walk $\mathcal{W}_4$ has one additional break point on the line $l_0: \alpha_{n_i}=\alpha$ in comparison with $\mathcal{W}_3$. 
Therefore, the number of points at which $\mathcal{W}_4$ can meet the line $\alpha_{n_i}=\alpha+krR_{\tt max}$ for the first time is greater than that for $\mathcal{W}_3$ by 1.
Fig.~\ref{fig:walk_exk} represents the two random walks $\mathcal{W}_3$ and $\mathcal{W}_4$, and the following holds:
\begin{equation*}
    \begin{aligned}
    &\beta^\prime =\beta(1+rx),\,\, \beta(1+rd_k)=\alpha+krR_{\tt max}, \\
    &\beta^\prime(1+rd_k^\prime)=\alpha+krR_{\tt max},\\
    &\beta(1+rd_{k+1})=\alpha+(k+1)rR_{\tt max},\\
    &\beta^\prime(1+rd_{k+1}^\prime)=\alpha+(k+1)rR_{\tt max}
    \end{aligned}
\end{equation*}
for $\beta^\prime>\beta$. 
Then we find $\frac{\partial (H_k^{\mathcal{W}_4}-H_k^{\mathcal{W}_3})}{\partial x}\Big|_{x=0}.$

First, $H_k^{\mathcal{W}_3}$ and $H_k^{\mathcal{W}_4}$ can be expressed as follows:
\begin{align*}
&H_k^{\mathcal{W}_3}=\prod_{i=0}^{k-1}\frac{(\alpha+irR_{\tt max}) \left(\frac{\alpha+irR_{\tt max}}{(1+\varepsilon)\beta}-1\right)}{rR_{\tt max}+(\alpha+irR_{\tt max}) \left(\frac{\alpha+irR_{\tt max}}{(1+\varepsilon)\beta}-1\right)}\times\\
&\hspace{1cm}\frac{rR_{\tt max}}{rR_{\tt max}+(\alpha+krR_{\tt max}) \left(\frac{\alpha+krR_{\tt max}}{(1+\varepsilon)\beta}-1\right)},\\
&H_k^{\mathcal{W}_4}=\prod_{i=1}^{k-1}\frac{(\alpha+irR_{\tt max}) \left(\frac{\alpha+irR_{\tt max}}{(1+\varepsilon)\beta}-1\right)}{rR_{\tt max}+(\alpha+irR_{\tt max}) \left(\frac{\alpha+irR_{\tt max}}{(1+\varepsilon)\beta}-1\right)}\\
&\times\frac{\alpha x}{R_{\tt max}+\alpha x}\cdot\frac{rR_{\tt max}}{rR_{\tt max}+(\alpha+krR_{\tt max}) \left(\frac{\alpha+krR_{\tt max}}{(1+\varepsilon)\beta}-1\right)}\\
&+\frac{R_{\tt max}}{R_{\tt max}+\alpha x}\cdot \prod_{i=0}^{k-1}\frac{(\alpha+irR_{\tt max}) \left(\frac{\alpha+irR_{\tt max}}{(1+\varepsilon)\beta(1+rx)}-1\right)}{rR_{\tt max}+(\alpha+irR_{\tt max}) \left(\frac{\alpha+irR_{\tt max}}{(1+\varepsilon)\beta(1+rx)}-1\right)}\\
&\times\frac{rR_{\tt max}}{rR_{\tt max}+(\alpha+krR_{\tt max}) \left(\frac{\alpha+krR_{\tt max}}{(1+\varepsilon)\beta(1+rx)}-1\right)}.
\end{align*}
In fact, when $\frac{\partial (H_k^{\mathcal{W}_4}-H_k^{\mathcal{W}_3})}{\partial x}\Big|_{x=0}$ is positive, it is always greater than $\frac{H_k^{\mathcal{W}_4}-H_k^{\mathcal{W}_3}}{x}$ for any $0<x<\frac{1}{r}\cdot\left(\frac{\alpha}{(1+\varepsilon)\beta}-1\right)$. 
In addition, if $\frac{\partial (H_k^{\mathcal{W}_4}-H_k^{\mathcal{W}_3})}{\partial x}\Big|_{x=0}$ is negative, $H_k^{\mathcal{W}_3}$ is greater than  $H_k^{\mathcal{W}_4}$. 
These facts implies that if $\frac{\partial (H_k^{\mathcal{W}_4}-H_k^{\mathcal{W}_3})}{\partial x}\Big|_{x=0}$ is positive, $H_k^{\mathcal{W}}$ can be maximized when there exist densely break points on the line $l_0$. 
Meanwhile, if $\frac{\partial (H_k^{\mathcal{W}_4}-H_k^{\mathcal{W}_3})}{\partial x}\Big|_{x=0}$ is negative, $H_k^{\mathcal{W}}$ can be maximized when there is no break point on line $l_0$.

The derivative $\frac{\partial (H_k^{\mathcal{W}_4}-H_k^{\mathcal{W}_3})}{\partial x}\Big|_{x=0}$ is equal to $\frac{\partial H_k^{\mathcal{W}_4}}{\partial x}\Big|_{x=0}$ because $\mathcal{W}_3$ is constant in terms of $x$. 
In addition, the value of $\frac{\partial H_k^{\mathcal{W}_4}}{\partial x}\Big|_{x=0}$ is equal to the value of $\frac{\partial A^k}{\partial x}\Big|_{x=0}$, where
\begin{equation*}
    \begin{aligned}
&A^k=\prod_{i=0}^{k-1}\frac{(\alpha+irR_{\tt max}) \left(\frac{\alpha+irR_{\tt max}}{(1+\varepsilon)\beta(1+rx)}-1\right)}{rR_{\tt max}+(\alpha+irR_{\tt max}) \left(\frac{\alpha+irR_{\tt max}}{(1+\varepsilon)\beta(1+rx)}-1\right)}\\
&\times\frac{rR_{\tt max}}{rR_{\tt max}+(\alpha+krR_{\tt max}) \left(\frac{\alpha+krR_{\tt max}}{(1+\varepsilon)\beta(1+rx)}-1\right)}
    \end{aligned}
\end{equation*}
The value of $\frac{\partial A^k}{\partial x}\Big|_{x=0}$ is expressed as 
\begin{align*}
&-r\cdot A^k\cdot\sum_{i=0}^{k-1}\frac{(l+i)^2}{((l+i)^2\frac{R^\prime_{\tt max}}{(1+\varepsilon)\beta}-l-i+1)^2}\times\\
&\frac{1+(l+i)((l+i)\frac{R^\prime_{\tt max}}{(1+\varepsilon)\beta}-1)}{(l+i)((l+i)\frac{R^\prime_{\tt max}}{(1+\varepsilon)\beta}-1)}+r\cdot A^k\times\\
&\frac{(l+k)^2}{((l+k)^2\frac{R^\prime_{\tt max}}{(1+\varepsilon)\beta}-l-k+1)^2}\times\\
&\left(1+(l+k)((l+k)\frac{R^\prime_{\tt max}}{(1+\varepsilon)\beta}-1)\right),
\end{align*}
where $R_{\tt max}^\prime=rR_{\tt max}$ and $l=\frac{\alpha}{R_{\tt max}^\prime}$.
Through the above equation, one can see that if $\frac{\partial (H_k^{\mathcal{W}_4}-H_k^{\mathcal{W}_3})}{\partial x}\Big|_{x=0}$ is positive when $l=l_0,$ $\frac{\partial (H_k^{\mathcal{W}_4}-H_k^{\mathcal{W}_3})}{\partial x}\Big|_{x=0}$ is also positive for all $l\geq l_0.$  
In other words, when the derivative value
is positive for $\alpha=\alpha_0$, it is positive for all $\alpha>\alpha_0.$ 

Also, we assume that $H_k^{\mathcal{W}_3}(\alpha,\beta)>H_k^{\mathcal{W}_4}(\alpha,\beta).$ 
This fact implies that 
\begin{equation*}
\resizebox{\hsize}{!}{$
\begin{aligned}
    &f_1\cdot\frac{R_{\tt max}}{R_{\tt max}+(\alpha+krR_{\tt max})d_{k}} +f_2\cdot \frac{R_{\tt max}}{R_{\tt max}+(\alpha+krR_{\tt max})d_{k}^\prime}\\
    &<f_3\cdot \frac{R_{\tt max}}{R_{\tt max}+(\alpha+krR_{\tt max})d_{k}},
\end{aligned}$}
\end{equation*}
where $f_1, f_2$, and $f_3$ are determined by $\mathcal{W}_3$ and $\mathcal{W}_4$. 
To prove that $H_{k+1}^{\mathcal{W}_3}(\alpha,\beta)>H_{k+1}^{\mathcal{W}_4}(\alpha,\beta),$ it is sufficient to show the following:
\begin{equation}
    \begin{aligned}
    &\frac{f_1\cdot (\alpha+krR_{\tt max})d_{k}}{R_{\tt max}+(\alpha+krR_{\tt max})d_{k}}\cdot \frac{R_{\tt max}}{R_{\tt max}+(\alpha+(k+1)rR_{\tt max})d_{k+1}}+\\
    &\frac{f_2\cdot(\alpha+krR_{\tt max})d_{k}^\prime}{R_{\tt max}+(\alpha+krR_{\tt max})d_{k}^\prime}\cdot \frac{R_{\tt max}}{R_{\tt max}+(\alpha+(k+1)rR_{\tt max})d_{k+1}^\prime}<\\
    &\frac{f_3\cdot (\alpha+krR_{\tt max})d_k}{R_{\tt max}+(\alpha+krR_{\tt max})d_k}\cdot \frac{R_{\tt max}}{R_{\tt max}+(\alpha+(k+1)rR_{\tt max})d_{k+1}}.
    \end{aligned}
    \label{eq:walk_eq}
\end{equation}
Then the above equation can be derived as follows: 
\begin{equation}
    \begin{aligned}
&(\alpha+krR_{\tt max})^2(\beta^\prime-\beta)+(\alpha+(k+1)rR_{\tt max})^2(\beta-\beta^\prime)<0 \\
&\Leftrightarrow(\alpha+krR_{\tt max}-\beta)(\beta^\prime rR_{\tt max}+(\alpha+(k+1)rR_{\tt max})\times\\
&(\alpha+(k+1)rR_{\tt max}-\beta^\prime)>(\alpha+krR_{\tt max}-\beta^\prime)\times\\
&(\beta rR_{\tt max}+(\alpha+(k+1)rR_{\tt max})(\alpha+(k+1)rR_{\tt max}-\beta))\\
&\Leftrightarrow d_k(R_{\tt max}+(\alpha+(k+1)rR_{\tt max})d_{k+1}^\prime)>d_k^\prime \times\\
&(R_{\tt max}+(\alpha+(k+1)rR_{\tt max})d_{k+1})\Rightarrow
\text{Eq.}~\eqref{eq:walk_eq}.
\end{aligned}
\label{eq:pf1}
\end{equation}
This fact implies that if $\frac{\partial H_k^{\mathcal{W}_4}}{\partial x}\Big|_{x=0}$ is negative when $k=k_0$, $\frac{\partial H_k^{\mathcal{W}_4}}{\partial x}\Big|_{x=0}$ is negative for all $k>k_0.$ 

Now, we consider when $l_k$ for $k\geq 1$ has an additional break point. 
Let us assume that there are two random walks $\mathcal{W}_1^k$ and $\mathcal{W}_2^k$, where $\mathcal{W}_2^k$ has an additional break point $(\alpha+krR_{\tt max}, \beta_2)$ on $l_k$ $(k\geq 1)$ below the final break point $(\alpha+krR_{\tt max}, \beta_1)$ located on $l_k$ $(k\geq 1)$ in the random walk $\mathcal{W}_1^k.$ 
Here, we assume that $\beta_2=(1+rx)\beta_1.$
Then, $H_{k+1}^{\mathcal{W}_1^k}<H_{k+1}^{\mathcal{W}_2^k}$, and this is easily proven by using the proof of that $H_{1}^{\mathcal{W}_1}<H_{1}^{\mathcal{W}_2}$, which is described above. 
In addition, if $\frac{\partial H_{k+1}^{\mathcal{W}_2^{k}}-H_{k+1}^{\mathcal{W}_1^{k}}}{\partial x}\Big|_{x=0}$ positive, it is always greater than $\frac{H_{k+1}^{\mathcal{W}_2^{k}}-H_{k+1}^{\mathcal{W}_1^{k}}}{x}$ for any $0<x<\frac{1}{r}\cdot \left(\frac{\alpha+krR_{\tt max}}{(1+\varepsilon)\beta_1}-1\right)$, and thus $H_{k+1}^{\mathcal{W}_2^{k}}(\alpha, \beta)$ can more efficiently increase when $x$ is closer to 0. 

Next, we consider $H_{k+N}^{\mathcal{W}_1^{k}}(\alpha, \beta)$ and $H_{k+N}^{\mathcal{W}_2^{k}}(\alpha, \beta)$.
The derivative $\frac{\partial (H_{k+N}^{\mathcal{W}_2^k}-H_{k+N}^{\mathcal{W}_1^k})}{\partial x}\Big|_{x=0}$ is equal to $\frac{\partial H_{k+N}^{\mathcal{W}_2^k}}{\partial x}\Big|_{x=0},$ and it can be expressed as 
\begin{align*}
&-rA_k^{N}\cdot\sum_{i=0}^{N-1}\frac{(l+k+i)^2}{((l+k+i)^2\frac{R^\prime_{\tt max}}{(1+\varepsilon)\beta}-l-k-i+1)^2}\times\\
&\frac{1+(l+k+i)((l+k+i)\frac{R^\prime_{\tt max}}{(1+\varepsilon)\beta}-1)}{(l+k+i)((l+k+i)\frac{R^\prime_{\tt max}}{(1+\varepsilon)\beta}-1)}+rA_k^{N}\times\\
&\frac{(l+k+N)^2}{((l+k+N)^2\frac{R^\prime_{\tt max}}{(1+\varepsilon)\beta}-l-k-N+1)^2}\times\\
&\left(1+(l+k+N)((l+k+N)\frac{R^\prime_{\tt max}}{(1+\varepsilon)\beta}-1)\right),
\end{align*}
where $R_{\tt max}^\prime=rR_{\tt max}$, $l=\frac{\alpha}{R_{\tt max}^\prime}$, and 
\begin{equation*}
    \begin{aligned}
&A_k^{N}=\prod_{i=k}^{k+N-1}\frac{(\alpha+irR_{\tt max}) \left(\frac{\alpha+irR_{\tt max}}{(1+\varepsilon)\beta(1+rx)}-1\right)}{rR_{\tt max}+(\alpha+irR_{\tt max}) \left(\frac{\alpha+irR_{\tt max}}{(1+\varepsilon)\beta(1+rx)}-1\right)}\\
&\times\frac{rR_{\tt max}}{rR_{\tt max}+(\alpha+(k+N)rR_{\tt max}) \left(\frac{\alpha+(k+N)rR_{\tt max}}{(1+\varepsilon)\beta(1+rx)}-1\right)}.
\end{aligned}
\end{equation*}
This implies that if $\frac{\partial H_{k+N}^{\mathcal{W}_2^k}}{\partial x}\Big|_{x=0}$ is positive when $k=k_0$, $\frac{\partial H_{k+N}^{\mathcal{W}_2^k}}{\partial x}\Big|_{x=0}$ is positive for all $k>k_0.$
In fact, when $k=1,$ $\frac{\partial H_{k+N}^{\mathcal{W}_2^k}}{\partial x}\Big|_{x=0}$ is positive regardless of $N$ and $\alpha$. 
Therefore, for all $k>0,$ $\frac{\partial H_{k+N}^{\mathcal{W}_2^k}}{\partial x}\Big|_{x=0}$ is positive regardless of $N$ and $\alpha$. 
In other words, $H_k^{\mathcal{W}}(\alpha,\beta)$ can be maximized when line $l_i$ has infinitely many break points for all $0<i<k.$

When we define the random walk $\mathcal{W}_{\tt max}^k$ as $\mathcal{W}_{\tt max}^k=\argmax_{\mathcal{W}}H_{k}^{\mathcal{W}}(\alpha,\beta)$,
the random walk $\mathcal{W}_{\tt max}^k$ has infinite break points on $l_i$ for $0<i<k$.
Formally, there always exist break points in interval $(\alpha+irR_{\tt max}, (\beta_1,\beta_2))$, for $\beta\leq\beta_1<\beta_2\leq \frac{\alpha+irR_{\tt max}}{1+\varepsilon}$.
Meanwhile, $\mathcal{W}_{\tt max}^k$ has no break point on $l_k$. 
In other words, in $\mathcal{W}_{\tt max}^k$, whenever a point moves to the line $l_k: \alpha_{n_i}=\alpha+krR_{\tt max}$, the point can reach the target zone where $\frac{\alpha_{n_i}}{\alpha_{n_j}}\leq 1+\varepsilon$, without break points. 
Considering the above facts, the following holds:
\begin{align}
&\max_{\mathcal{W}}H_{k}^{\mathcal{W}}(\alpha,\beta)=H_{k}^{\mathcal{W}^k_{\tt max}}(\alpha,\beta)=\notag\\
&\lim_{d\rightarrow 0}\sum_{\substack{\forall j<k:\sum_{i=0}^j x_i< m_j^d}}\left\{\frac{rR_{\tt max}}{rR_{\tt max}+(\alpha+krR_{\tt max})\cdot D_k}\right.\notag\\
&\hspace{4cm}\left.\times\prod_{i=0}^{k-1}h_i(x_i,d)\right\},\label{eq:final}
\end{align}
where 
\begin{equation*}
\resizebox{\hsize}{!}{$
\begin{aligned}
&m_j^d=\log_{1+rd}\left(\frac{\alpha+jrR_{\tt max}}{(1+\varepsilon)\beta}\right)\text{ for } j>0,\,\, m_0^d=\log_{1+rd}\left(\frac{\alpha+jrR_{\tt max}}{(1+\varepsilon)\beta^\star}\right),\\ 
&D_k=\frac{1}{r}\cdot\left(\frac{\alpha+krR_{\tt max}}{(1+\varepsilon)\beta(1+rd)^{\sum_{i=0}^{k-1}{x_i}}}-1\right),\\
&h_i(x_i,d)=\left(\frac{rR_{\tt max}}{rR_{\tt max}+(\alpha+irR_{\tt max})d}\right)^{x_i}\cdot \left(\frac{(\alpha+irR_{\tt max})d}{rR_{\tt max}+(\alpha+irR_{\tt max})d}\right)
\end{aligned}$}
\end{equation*}
The notation $\beta^\star$ denotes the root of the following equation for $\beta:$
\begin{equation*}
\begin{aligned}
&\sum_{i=0}^{k-1}\frac{(l+i)^2}{\left((l+i)^2\frac{R^\prime_{\tt max}}{(1+\varepsilon)\beta}-l-i+1\right)^2}\cdot\frac{1+(l+i)((l+i)\frac{R^\prime_{\tt max}}{(1+\varepsilon)\beta}-1)}{(l+i)((l+i)\frac{R^\prime_{\tt max}}{(1+\varepsilon)\beta}-1)}=\\
&\frac{(l+k)^2}{\left((l+k)^2\frac{R^\prime_{\tt max}}{(1+\varepsilon)\beta}-l-k+1\right)^2}\cdot
\left(1+(l+k)((l+k)\frac{R^\prime_{\tt max}}{(1+\varepsilon)\beta}-1)\right),
\end{aligned}
\end{equation*}
where $R_{\tt max}^\prime=rR_{\tt max}$ and $l=\frac{\alpha}{R_{\tt max}^\prime}$. 
Note that the root is unique. 
Then we denote $g_k(\alpha,\beta)$ by $H_k^{\mathcal{W}^k_{\tt max}}(\alpha,\beta)$ for ease of reading.

Finally, because 
\begin{equation*}
\begin{aligned}
\max_{x=R_{\tt max}}&\Pr\left(\frac{\alpha}{\beta} \rightarrow 1+\varepsilon \,\Big|\, (\alpha,\beta)\right)=\max\sum_{k=0}^{\infty}P^\varepsilon_k(\alpha,\beta)\\
&\leq\sum_{k=0}^{\infty}\max P^\varepsilon_k(\alpha,\beta)=\sum_{k=0}^{\infty}g_i(\alpha,\beta),    
\end{aligned}    
\end{equation*}
the probability for a state $(\alpha_{n_i},\alpha_{n_j})$ starting from $(\alpha,\beta)$ to reach the target zone in which satisfies $\frac{\alpha_{n_i}}{\alpha_{n_j}}\leq 1+\varepsilon$ is upper bounded by
\begin{equation}
\begin{aligned}
\lim_{\substack{d\rightarrow 0\\n\rightarrow \infty}}
&\sum_{k=0}^{n}\Bigg\{\sum_{\substack{\forall j<k:\sum_{i=0}^j x_i< m_j^d}}\Big\{\frac{rR_{\tt max}}{rR_{\tt max}+(\alpha+krR_{\tt max})\cdot D_k}\\
&\hspace{2cm}\times\prod_{i=0}^{k-1}h_i(x_i,d)\Big\}\Bigg\},
\end{aligned}    
\label{eq:final2}
\end{equation}
which is denoted by $G(\alpha,\beta).$
Note that 
\begin{equation*}
\begin{aligned}
    g_0(\alpha,\beta)\geq &\frac{R_{\tt max}}{R_{\tt max}+\alpha d}\cdot g_0(\alpha,\beta(1+rd)) \text{ and }\\
g_i(\alpha,\beta)\geq &\frac{R_{\tt max}}{R_{\tt max}+\alpha d}\cdot g_i(\alpha,\beta(1+rd))+\\
&\frac{\alpha d}{R_{\tt max}+\alpha d}\cdot g_{i-1}(\alpha+rR_{\tt max},\beta) \,\,\,\,\forall i>0.
\end{aligned}
\end{equation*}
Therefore, the following holds: 
$$G(\alpha,\beta)\geq\frac{R_{\tt max}}{R_{\tt max}+\alpha d}\cdot G(\alpha,\beta(1+rd))+\frac{\alpha d}{R_{\tt max}+\alpha d}\cdot G(\alpha+rR_{\tt max},\beta).$$
Also, Eq.~\eqref{eq:const} is the maximum when $x=R_{\tt max}.$
More specifically, Eq.~\eqref{eq:const} has a similar form to that shown in Fig.~\ref{fig:find_x}.
Lastly, because the limit value of $G(\alpha,\beta)$ when $\alpha$ goes to infinity is 0, it is a constant in terms of $x.$ 
As a result, 
$$\lim_{t\rightarrow\infty}\Pr\Bigr[\frac{EP_{\tt max}^t}{EP^t_{\delta}}<1+\varepsilon\Bigl]<G(\alpha_{\tt MAX}, \alpha_{\delta}),$$ 
and $G(\alpha_{\tt MAX}, \alpha_{\delta})$ is denoted by $G^\varepsilon(f_\delta, \frac{rR_{\tt max}}{\alpha_{\tt MAX}})$ in Theorem~\ref{thm:impossible}. 
Moreover, the limit value of $G^\varepsilon(f_\delta, \frac{rR_{\tt max}}{\alpha_{\tt MAX}})$ when $f_\delta$ goes to 0 is 0. 
This completes the proof of Theorem~\ref{thm:impossible}. 

\section{Simulation}
\label{app:sim}

\begin{figure}[ht]
    \centering
    \includegraphics[width=\columnwidth]{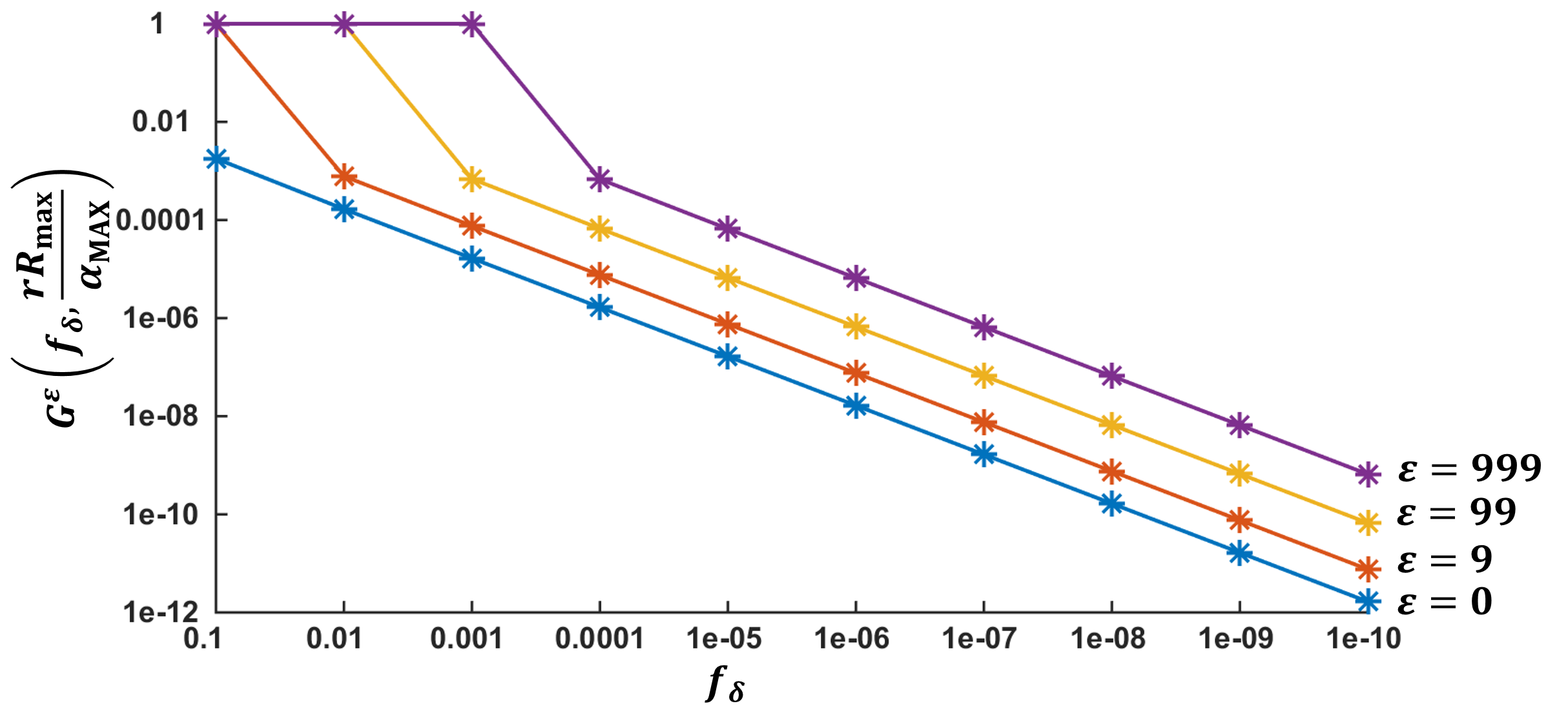}
    \caption{In this figure, when $\frac{rR_{\tt max}}{\alpha_{\tt MAX}}$ is $10^{-2}$, $G^\varepsilon(f_\delta, \frac{rR_{\tt max}}{\alpha_{\tt MAX}})$ ($y$-axis) is presented with regard to $f_\delta$ ($x$-axis) and $\varepsilon$. }
    \label{fig:sim2}
\end{figure}

\begin{figure}[ht]
    \centering
    \includegraphics[width=\columnwidth]{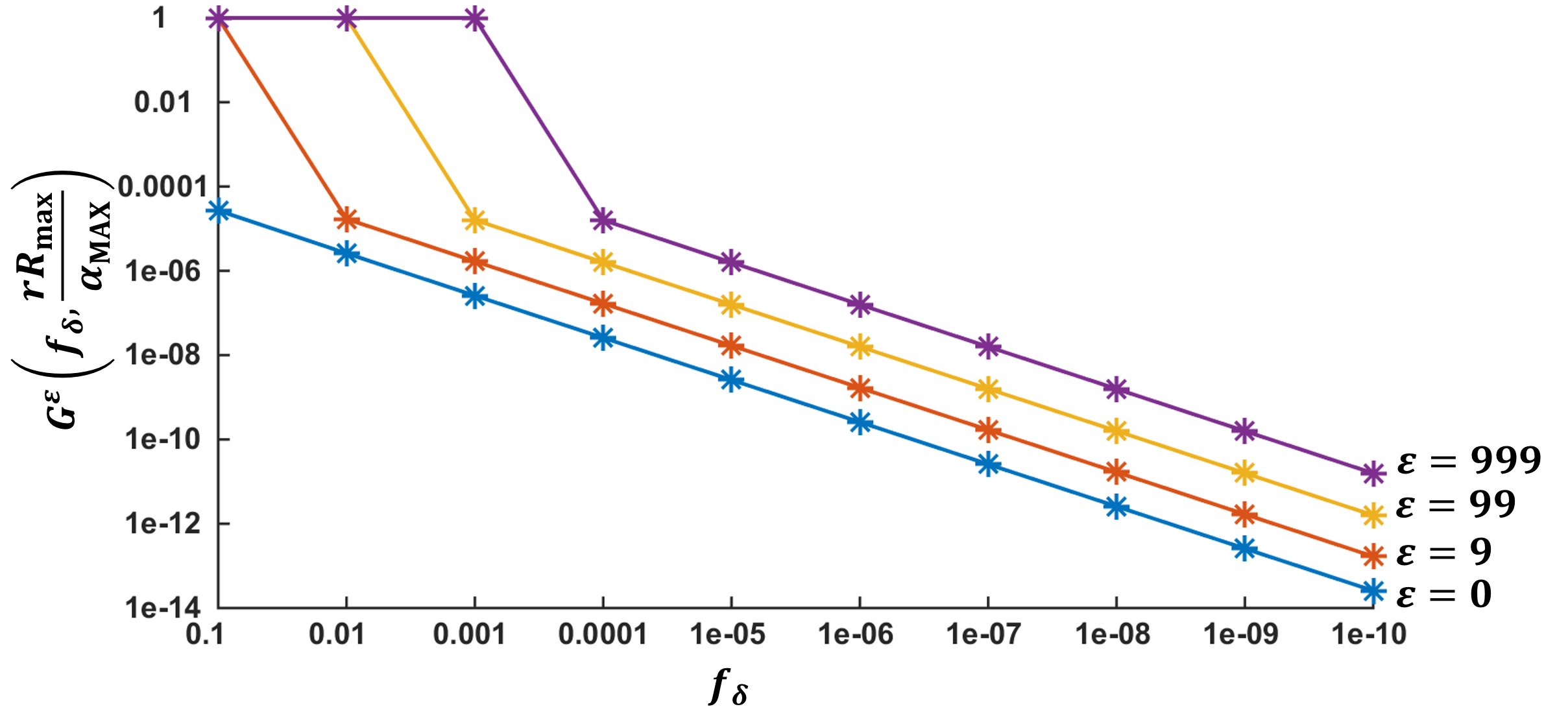}
    \caption{In this figure, when $\frac{rR_{\tt max}}{\alpha_{\tt MAX}}$ is $10^{-4}$, $G^\varepsilon(f_\delta, \frac{rR_{\tt max}}{\alpha_{\tt MAX}})$ ($y$-axis) is presented with regard to $f_\delta$ ($x$-axis) and $\varepsilon$. }
    \label{fig:sim3}
\end{figure}

\end{document}